\documentclass[prologue,dvipsnames,sigconf,screen]{acmart}

\def\fullversion

\setcopyright{acmcopyright}
\copyrightyear{2018}
\acmYear{2018}
\acmDOI{XXXXXXX.XXXXXXX}

\usepackage{booktabs}   
\usepackage{subcaption} 
\usepackage{caption}
\usepackage{colortbl}
\usepackage{bigstrut}
\usepackage{microtype}
\usepackage{tabularx}
\usepackage{multirow}
\usepackage{multicol}
\usepackage{booktabs}
\usepackage{rotating}
\usepackage{lipsum}
\usepackage{balance}
\usepackage{mfirstuc}
\usepackage{titlecaps}
\usepackage{listings}
\usepackage{float}
\usepackage[rightcaption]{sidecap}
\usepackage{setspace}
\usepackage{tcolorbox}
\usepackage[normalem]{ulem}
\useunder{\uline}{\ul}{}

\usepackage{xcolor}
\usepackage[font=normalfont,labelfont=bf,skip=2pt]{caption}

\newcommand{\oref}[1]{\hyperref[#1]{#1}}

\usepackage{xr}
\usepackage{etoolbox}
\newcommand{\ifconference}[1]{{{\ifx\fullversion\undefined{#1}\fi}\xspace}}
\newcommand{\iffullversion}[1]{{{\ifx\conference\undefined{#1}\fi}\xspace}}

\usepackage{graphicx}  
\usepackage{lipsum}  
\newcommand{\hide}[1]{} 
\usepackage{xspace}
\usepackage{textcomp}
\usepackage{comment} 
\usepackage{verbatim}
\usepackage{url}

\usepackage{pifont}

\usepackage[shortlabels]{enumitem}

\setlist{topsep=0.3em,itemsep=0.2em,parsep=0.1em,leftmargin=*}

\usepackage{float}
\usepackage[font={small},aboveskip=0em, belowskip=0em]{caption}

\usepackage[labelfont=bf,list=true,skip=0em]{subcaption}
\usepackage[rightcaption]{sidecap}

\usepackage{wrapfig}

\usepackage{array}
\newcolumntype{L}[1]{>{\raggedright\let\newline\\\arraybackslash\hspace{0pt}}m{#1}}
\newcolumntype{C}[1]{>{\centering\let\newline\\\arraybackslash\hspace{0pt}}m{#1}}
\newcolumntype{R}[1]{>{\raggedleft\let\newline\\\arraybackslash\hspace{0pt}}m{#1}}
\newcolumntype{B}{>{\bf}c}
\usepackage{rotating}

\usepackage{booktabs} 
\usepackage{multicol,multirow}
\usepackage{longtable} 
\usepackage{supertabular} 
\usepackage{colortbl}
\usepackage{bigstrut}

\usepackage{titlesec}
\titleformat{\subsection}{\normalfont\large\bfseries}{\thesubsection}{1em}{}

\titlespacing{\section}{0pt}{0.3em}{0.2em} 
\titlespacing{\subsection}{0pt}{0.3em}{0.2em} 
\titlespacing{\subsubsection}{0pt}{0.1em}{1em} 
\newcommand{\mysubsubsection}[1]{\underline{#1}.}
\titleformat{\subsubsection}[runin]
{\normalfont\normalsize\bfseries}{\thesubsubsection}{1em}{\mysubsubsection}

\newcommand{\myparagraph}[1]{\vspace{.03in}\noindent {\boldmath\bf #1\unboldmath}}

\newcommand{\rmv}[1]{\textcolor{red}{\sout{}}}

\usepackage[ruled,lined,linesnumbered,noend]{algorithm2e}
\usepackage[noend]{algpseudocode}

\makeatletter
\patchcmd{\@algocf@start}
  {-1.5em}
  {0pt}
  {}{}
\newcommand{\nosemic}{\renewcommand{\@endalgocfline}{\relax}}
\newcommand{\dosemic}{\renewcommand{\@endalgocfline}{\algocf@endline}}

\SetSideCommentLeft
\SetKwInput{notations}{Notations}
\SetKwInput{notes}{Notes}
\SetKwInput{maintains}{Maintains}

\SetKwProg{myfunc}{Function}{}{}
\SetKwFor{parForEach}{ParallelForEach}{do}{endfor}
\SetKwFor{Justrepeat}{Repeat}{}{}

\definecolor{dpcol}{RGB}{0,160,240}

\newcommand{\yihan}[1]{{\color{blue}{\bf Yihan:} #1}}

\newcommand{\ziyang}[1]{{\color{BlueGreen}{\bf Ziyang:} #1}}

\newcommand{\mb}[1]{{\mbox{\emph{#1}}}}

\SetCommentSty{mycommfont}

\usepackage{mdframed}
\definecolor{framelinecolor}{RGB}{68,114,196}
\mdfdefinestyle{mystyle}{linecolor=framelinecolor,innertopmargin=1pt,innerbottommargin=2pt,backgroundcolor=gray!20,skipabove=2pt,skipbelow=0pt}
\mdfdefinestyle{densestyle}{linecolor=framelinecolor,innertopmargin=0,innerbottommargin=0,leftmargin=0,rightmargin=0,backgroundcolor=gray!20}
\mdfdefinestyle{compactcode}{linecolor=framelinecolor,innertopmargin=1pt,innerbottommargin=1pt,backgroundcolor=gray!20,skipabove=0pt,skipbelow=0pt,leftmargin=0,rightmargin=0}

\usepackage{framed}

\usepackage{listings}

\newdimen\zzsize
\newdimen\kwsize
\newcommand{\basicstyle}{\fontsize{\zzsize}{1\zzsize}\ttfamily}
\newcommand{\keywordstyle}{\fontsize{\kwsize}{1\kwsize}\ttfamily\bf}

\newdimen\zzlstwidth
\usepackage{cleveref}
\crefname{appendix}{Appendix}{Appendix}
\crefname{theorem}{Thm.}{Thm.}
\crefname{lemma}{Lem.}{Lem.}
\crefname{corollary}{Col.}{Col.}
\crefname{table}{Tab.}{Tab.}
\crefname{algorithm}{Alg.}{Alg.}
\crefname{figure}{Fig.}{Fig.}
\crefname{fact}{Fact}{Fact}
\Crefname{table}{Tab.}{Tab.}
\crefname{problem}{Problem}{Problem}

\usepackage{tikz} 

\setcopyright{none}
\renewcommand\footnotetextcopyrightpermission[1]{} 

\newcommand*{\rom}[1]{\expandafter\@slowromancap\romannumeral #1@}
\newcommand{\type}[1]{\texttt{TYPE} \rom{#1}\xspace}

\newcommand{\kdtree}{$k$d-tree\xspace}

\newcommand{\ourtree}{\mbox{\textsf{Pkd-tree}}\xspace}
\newcommand{\ours}{\mbox{\textsf{Pkd}}\xspace}
\newcommand{\bhltree}{\mbox{\textsf{BHL-tree}}\xspace}
\newcommand{\logtree}{\mbox{\textsf{Log-tree}}\xspace}
\newcommand{\oursbb}{\mbox{\textsf{Pkd-bb}}\xspace}
\newcommand{\pargeo}{\textsf{ParGeo}\xspace}
\newcommand{\rtree}{\textsf{R-tree}\xspace}
\newcommand{\zdtree}{\textsf{Zd-tree}\xspace}
\newcommand{\cgal}{\textsf{CGAL}\xspace}
\newcommand{\ourlib}{\textsf{Pkd-tree}\xspace}
\newcommand{\balpara}{\alpha}
\newcommand{\varden}{\texttt{Varden}\xspace}
\newcommand{\uniform}{\texttt{Uniform}\xspace}

\newcommand{\lson}{\ell}
\newcommand{\rson}{\textit{r}}

\newcommand{\buildTreeSkeleton}{\textsc{BuildTreeSkeleton}}

\newcommand{\build}{\textsc{BuildTree}}

\newcommand{\fetchTreeSkeleton}{\textsc{FetchTreeSkeleton}}
\newcommand{\batchinsert}{\textsc{BatchInsert}}
\newcommand{\deleteTree}{\textsc{DeleteTree}}

\newcommand{\batchdelete}{\textsc{BatchDelete}}

\newcommand{\leafwrap}{\phi}
\newcommand{\skheight}{\lambda}
\newcommand{\samplerate}{\sigma}
\newcommand{\knn}{$k$-NN}

\newcommand{\os}{\sigma\xspace}
\newcommand{\levels}{\lambda\xspace}
\newcommand{\ourkdtree}{\ourtree}
\newcommand{\lc}{\mathit{lc}}
\newcommand{\rc}{\mathit{rc}}

\SetKw{MIN}{min}
\SetKw{MAX}{max}
\SetKw{OR}{or}
\SetKw{AND}{and}

\newcommand{\mtext}[1]{{\mbox{{#1}}}} 
\newcommand{\emp}[1]{\emph{\textbf{\boldmath #1\unboldmath}}} 

\newcommand{\ip}[2]{\langle#1,#2\rangle}
\newcommand{\bdita}[1]{{\boldmath\textbf{\textit{#1}}\unboldmath}}

\newcommand{\R}{\mathbb{R}}
\newcommand{\N}{\mathbb{N}}

\newcommand{\matht}{\mathcal{T}\xspace}
\newcommand{\varmatht}{\mathscr{T}\xspace}
\newcommand{\whp}{\emph{whp}\xspace}

\DeclareMathOperator*{\polylog}{polylog}

\newtheorem{theorem}{Theorem}[section]
\newtheorem{lemma}[theorem]{Lemma}

\let \originalleft \left
\let\originalright\right
\renewcommand{\left}{\mathopen{}\mathclose\bgroup\originalleft}
\renewcommand{\right}{\aftergroup\egroup\originalright}

\usepackage{scalerel} 

\newtheoremstyle{exampstyle}
{.5em} 
{1em} 
{\it} 
{.5em} 
{\it \bfseries} 
{.} 
{.5em} 
{} 

\makeatletter
\renewenvironment{proof}[1][\proofname]{\par
	\vspace{-2\topsep}
	\pushQED{\qed}%
	\normalfont
	\topsep0pt \partopsep0pt 
	\trivlist
	\item[\hskip\labelsep
	\itshape
	#1\@addpunct{.}]\ignorespaces
}{%
	\popQED\endtrivlist\@endpefalse
}

\crefname{section}{Sec.}{Sec.}
\crefname{table}{Tab.}{Tab.}
\crefname{figure}{Fig.}{Fig.}

\makeatletter
\newcommand{\pushright}[1]{\ifmeasuring@#1\else\omit\hfill$\displaystyle#1$\fi\ignorespaces}
\newcommand{\pushleft}[1]{\ifmeasuring@#1\else\omit$\displaystyle#1$\hfill\fi\ignorespaces}
\makeatother

\begin{document}

\makeatletter
\gdef\@copyrightpermission{
	\begin{minipage}{0.2\columnwidth}
		\href{https://creativecommons.org/licenses/by/4.0/}{\includegraphics[width=0.90\textwidth]{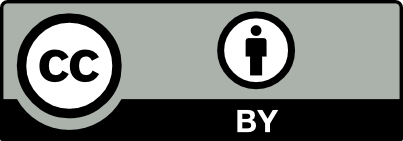}}
	\end{minipage}\hfill
	\begin{minipage}{0.8\columnwidth}
		\href{https://creativecommons.org/licenses/by/4.0/}{This work is licensed under a Creative Commons Attribution International 4.0 License.}
	\end{minipage}
	\vspace{5pt}
}
\makeatother

\title{Parallel \kdtree{} with Batch Updates}


\settopmatter{authorsperrow=4}

\author{Ziyang Men}
\email{zmen002@ucr.edu}
\affiliation{\institution{UC Riverside}\country{}}
\author{Zheqi Shen}
\email{zshen055@ucr.edu}
\affiliation{\institution{UC Riverside}\country{}}
\author{Yan Gu}
\email{ygu@cs.ucr.edu}
\affiliation{\institution{UC Riverside}\country{}}
\author{Yihan Sun}
\email{yihans@cs.ucr.edu}
\affiliation{\institution{UC Riverside}\country{}}


\begin{abstract}
	The \kdtree{} is one of the most widely used data structures to manage multi-dimensional data.
	Due to the ever-growing data volume, it is imperative to consider parallelism in \kdtree{s}.
	However, we observed challenges in existing parallel \kdtree{} implementations, for both constructions and updates.

	The goal of this paper is to develop efficient in-memory \kdtree{s} by supporting high parallelism and cache-efficiency.
	We propose the \ourtree{} (Parallel \kdtree{}), a parallel \kdtree{} that is efficient both in theory and in practice.
	The \ourtree{} supports parallel tree construction, batch update (insertion and deletion), and various queries including $k$-nearest neighbor search, range query, and range count.
	We proved that our algorithms have strong theoretical bounds in work (sequential time complexity), span (parallelism), and cache complexity.
	Our key techniques include 1) an efficient construction algorithm that optimizes work, span, and cache complexity simultaneously, and 2) reconstruction-based update algorithms that guarantee the tree to be weight-balanced.
	With the new algorithmic insights and careful engineering effort, we achieved a highly optimized implementation of the \ourtree{}.

	We tested \ourtree{} with various synthetic and real-world datasets, including both uniform and highly skewed data.
	We compare the \ourtree{} with state-of-the-art parallel \kdtree{} implementations.
	In all tests, with better or competitive query performance,
	\ourtree{} is much faster in construction and updates consistently than all baselines.
	We released our code.
	\let\thefootnote\relax\footnote{To appear: International Conference on Management of Data (SIGMOD), 2025}
	\setcounter{footnote}{0}
\end{abstract}

\hide{
	We present a theoretically efficient parallel tree construction and batch-update algorithm for \kdtree{} with a high-performance parallel implementation. Although the use of \kdtree{} is common in many scenarios such as database, clustering, machine learning, etc, current implementation of \kdtree{} either falls short of supporting efficient batch update or integrates limited parallelism as a software library. This gap significantly restricts the usage of \kdtree{} especially when dealing with large volumes of data. We propose a parallel \kdtree{} construction algorithm that achieves optimal $O(n\log n)$ work, $O(\log n\log_M n)$ span, and $O(\text{Sort}(n))$ I/O bound for point sets of size $n$. This is motivated by emerging the widely used techniques in \kdtree{} community, such as sampling and partial rebuilding, with the I/O-friendly bucket-distribution from recent parallel sorting algorithm. The tree has height $\log n+O(1)$ and thus preserves most query bounds. Batch update of size $m$ requires $O(\log n\log_M n)$ span, and amortized $O(\log^2 n)$ work and $O(\log(n/m)+(\log n\log_Mn/B))$ I/O cost per element. All bounds above are held with high probability.

	With a combination of new algorithmic insights and careful engineering effort, we provide a highly optimized implementation of parallel \kdtree{}, named \ourlib, that supports parallel tree construction, batch-update and various query types. The \ourlib outperforms other state-of-the-art work in almost every benchmark we tested. In addition to a comprehensive experimental study, we also investigate the effect of applied techniques in tree construction quantitatively. Our experiments show that \kdtree{} is reasonably consistent with minor imbalance for different types of queries.}



\renewcommand\footnotetextcopyrightpermission[1]{} 
\fancyhead{} 

\maketitle
\setcounter{page}{1}
\section{Introduction}\label{sec:intro}
The \kdtree{} is one of the most widely-used data structures for managing multi-dimensional data.
A \kdtree{} maintains a set of points in $D$ dimensions\footnote{Based on the original terminology, \kdtree{} deals with $k$-dimensional data. To avoid overloading $k$ in different scenarios such as the ``\knn{}'' query (i.e., finding $k$ nearest neighbors of a given point), we use $D$ as the number of dimensions in this paper.}, and supports various queries such as $k$-nearest neighbor (\knn{}), orthogonal range count and range report.
Compared to other counterparts,
the \kdtree{} has its unique advantages, such as linear space,
simple algorithms,
being comparison-based (and thus resistant to skewed data),
scaling to reasonably-large dimensions (being efficient up to $D\approx 10$)
and supporting a wide range of query types.
Due to these advantages, the \kdtree{} is the choice of data structure in many applications.
Indeed, after its invention by Bentley in 1975~\cite{bentley1975multidimensional}, \kdtree{} has been widely used and cited by over ten thousand times across multiple areas such as databases~\cite{guting1994introduction,Graefe93,bohm2001searching,ma2018query}, data science~\cite{tang2016visualizing, yue2016healthcare,muja2014scalable,deng2016efficient}, machine learning~\cite{li2018so,shakhnarovich2005nearest,li2019graph,chen2016gene}, clustering algorithms~\cite{schubert2017dbscan,silva2013data,likas2003global,kanungo2002efficient,mcinnes2017accelerated}, and computational geometry~\cite{malkov2018efficient,smith2016structure,guo2013rotational,blonder2018new}.

Due to the ever-growing data volume, it is imperative to consider parallelism in \kdtree{s}.
For instance, the North American region of
OpenStreetMap~\cite{haklay2008openstreetmap} contains 1.29 billion nodes, and
building a \kdtree{} for this dataset on a single core using \cgal{}~\cite{cgal51}, a widely-adopted geometry library,
takes over 2000 seconds.
However, we observe a \emph{significant gap} between the \emph{wide usage} of \kdtree{s},
and a \emph{lack of high-performance parallel implementation} of \kdtree{s} for all three aspects of construction, updates, and queries.
Some existing parallel implementations (e.g.~\cite{reif2022scalable}) are static and do not support updates.
The two parallel libraries for dynamic \kdtree{s} that we are aware of,
\cgal{}~\cite{cgal51}, and ParGeo~\cite{wang2022pargeo} (which includes two \kdtree{} implementations \mbox{\logtree} and {\bhltree}),
both have difficulties scaling to today's large-scale data size (see a summary of results in \cref{tab:introTable}).
\cgal{} and the \bhltree do not support parallel updates.
Even in the sequential updates, they fully rebuild the tree for rebalancing, which is inefficient.
The \logtree parallelizes updates using the classic \emph{logarithmic method}~\cite{bentley1975multidimensional,procopiuc2003bkd,agarwal2003cache}.
The logarithmic method avoids fully rebuilding the tree upon update by maintaining $O(\log n)$ perfectly balanced trees with different sizes, such that an update reorganizes the trees by merging some of them in parallel.
Accordingly, a query processes all $O(\log n)$ trees and combines the results, which can be significantly more expensive than it on a single \kdtree{}.
As shown in \cref{tab:introTable}, the \logtree, despite being faster on updates, can be up to an order of magnitude slower than the \bhltree or \cgal{} 
on \knn{} queries.
In addition, the construction for these \kdtree{s} is also much slower than the time reported in recent works of other parallel tree structures, such as binary search trees~\cite{sun2018pam,dhulipala2022pac} and quad/octrees~\cite{blelloch2022parallel}\footnote{For example, the construction time of the parallel binary search tree reported in~\cite{sun2018pam} is 28s on $n=10^{10}$ elements on a similar machine, which is faster than all previous \kdtree{} implementations on $n=10^{9}$ shown in \cref{tab:introTable}.}, indicating significant space for improvements on the construction algorithm.

\begin{table*}[htbp]
	\centering
	\small
	\setlength\tabcolsep{6pt} 
	\renewcommand{\arraystretch}{0.8} 

	\begin{tabular}{cc|c|cccc|cccc|c|c}
		\toprule
		\textbf{Benchmark}                    & \multirow{2}[2]{*}{\textbf{Baselines}} & \multirow{2}[2]{*}{\textbf{Build}} & \multicolumn{4}{c|}{\textbf{Batch Insert}} & \multicolumn{4}{c|}{\textbf{Batch Delete}} & \textbf{10-NN }  & \textbf{Range Report}                                                                                                                                 \\
		\textbf{\boldmath($10^9$-2D)\unboldmath}                  &                                        &                                    & \textbf{0.01\%}                            & \textbf{0.1\%}                             & \textbf{1\%}     & \textbf{10\%}         & \textbf{0.01\%}  & \textbf{0.1\%}   & \textbf{1\%}     & \textbf{10\%}    & \textbf{\boldmath$10^7$ \unboldmath queries} & \textbf{\boldmath$10^4$ \unboldmath queries} \\
		\midrule
		\multirow{4}[2]{*}{\textbf{\uniform}} & Ours                                   & \underline{3.15}                   & \underline{.004}                           & \underline{.020}                           & \underline{.104} & \underline{.495}      & \underline{.004} & \underline{.022} & \underline{.121} & \underline{.526} & \underline{.381}        & \underline{.391}        \\
		                                      & \logtree                               & 37.9                               & .008                                       & .059                                       & 2.16             & 30.7                  & .436             & .168             & .396             & 3.01             & 2.96                    & 2.62                    \\
		                                      & \bhltree                               & 31.7                               & 31.0                                       & 31.2                                       & 31.4             & 39.9                  & 30.8             & 31.1             & 30.9             & 30.6             & .487                    & 2.06                    \\
		                                      & \cgal                                  & 1147                               & 1614                                       & 1562                                       & 1631             & 1660                  & .400             & 3.89             & 41.2             & 427              & 1.04                    & 311                     \\
		\midrule
		\multirow{4}[2]{*}{\textbf{\varden}}  & Ours                                   & \underline{3.66}                   & \underline{.002}                           & \underline{.007}                           & \underline{.055} & \underline{.473}      & \underline{.002} & \underline{.006} & \underline{.049} & \underline{.477} & \underline{.172}        & \underline{.382}        \\
		                                      & \logtree                               & 34.2                               & .008                                       & .057                                       & 2.01             & 28.0                  & .799             & .848             & 1.06             & 3.47             & 2.05                    & 2.63                    \\
		                                      & \bhltree                               & 30.2                               & 29.1                                       & 29.2                                       & 29.4             & 37.3                  & 29.3             & 29.2             & 29.0             & 28.0             & .239                    & 1.95                    \\
		                                      & \cgal                                  & 429                                & 867                                        & 867                                        & 849              & 836                   & .113             & 1.06             & 13.0             & 153              & .511                    & 296                     \\
		\bottomrule
	\end{tabular}%

	\vspace{.2em}
	\caption{\textbf{Running time (in seconds) for \ourtree{} and other baselines on $10^9$ points in 2 dimensions. Lower is better.} \normalfont
		``\logtree'': the parallel \kdtree using logarithmic method from the \pargeo{} library~\cite{wang2022pargeo}. 
        ``\bhltree'': the single parallel \kdtree from \pargeo{} library~\cite{wang2022pargeo}.
        ``\cgal'': the \kdtree{} from the \cgal{} library~\cite{cgal51}.        
        ``Varden'': a skewed distribution from~\cite{gan2017hardness}.
		``10-NN'': 10-nearest-neighbor queries on $10^7$ points.
		``Range report'': orthogonal range report queries on $10^4$ rectangles, with output sizes in $10^4$--$10^6$.
		Experiments are run on a 96-core machine. More details are in \cref{sec:exp}.
		The fastest time for each test is underlined. 
		\vspace{-1em}	}
	\label{tab:introTable}%
\end{table*}%

\textbf{In this paper, we overcome the above challenges in existing work by proposing the \emp{\ourtree{}} (\boldmath Parallel \kdtree{}\unboldmath), a parallel in-memory \kdtree{} that is efficient both in theory and in practice. }
\ourtree{} supports efficient construction, batch-update, and various query types. 
Our algorithms have strong theoretical bounds in work (sequential time complexity), span (parallelism), and cache complexity.
Our key techniques include 1) an efficient construction algorithm that optimizes work, span, and cache complexity simultaneously, and 2) reconstruction-based update algorithms that guarantee the tree to be weight-balanced.
With the new algorithmic insights and careful engineering effort, we achieved a highly optimized implementation of the \ourtree{}.



\myparagraph{Construction.}
Our first contribution is a new parallel algorithm to construct weight-balanced \kdtree{s}, given in \cref{sec:constr}.
To the best of our knowledge, this is the first \kdtree{} construction algorithm with optimal $O(n\log n)$ work and $O((n/B)\log_M n)$ cache complexity, and polylogarithmic span, all with high probability\footnote{The \emph{work} of a parallel algorithm is the total number of operations (i.e., sequential time complexity), and its \emph{span} is the longest dependence chain. All the terms here are formally defined in \cref{sec:prelim}.}, where $n$ is the tree size, $M$ is the cache size, and $B$ is the cacheline size.
To achieve good bounds on all three metrics simultaneously, the algorithmic highlight here
is to 1) determine the splitting hyperplane using a carefully designed sampling scheme, and 2) a \emp{sieving algorithm} to partition all points into subspaces of $\skheight$ levels in the \kdtree{} by \emph{one round of data movement}.
By picking $\skheight=\Theta(\log M)$, we achieve strong theoretical bounds for the construction algorithm and good performance in practice due to the saving of memory accesses.

\myparagraph{Updates.}
The \kdtree{} differs from other classic trees
and does not support rebalancing primitives for updates, such as overflow/underflow (as in B-trees),
or rotations (as in binary search trees).
Our idea is to keep the tree weight-balanced, and use
a \emp{lazy strategy} that tolerates the difference of sibling subtree sizes by a predefined and controllable weight-balancing factor of $\balpara$ before invoking rebalancing by \emp{locally reconstructing} the affected subtrees.
The idea of rebalancing via local reconstruction was originally proposed by Overmars from early 80s, and has been studied in the sequential setting ~\cite{overmars1983design,blelloch2018geometry,overmars1981maintenance,galperin1993scapegoat,andersson1989improving}.
However, it remained previously unknown about the efficiency on parallelism and cache complexity.
Interestingly, the efficiency of our update algorithm is achieved by making use of our new construction algorithm---first, rebalancing the tree relies on efficient reconstruction; second, both insertion and deletion use the \emph{sieving process} in construction as a subroutine to also achieve good cache complexity and parallelism. We present our update algorithms and the cost analysis in \cref{sec:update}.

\myparagraph{Queries.}
Since the \ourtree{} remains as a single \kdtree{}, the same query algorithms for static \kdtree{s} can directly work on \ourtree{s} without any modifications.
In \cref{sec:exp}, we will show that the query performance for \ourtree{s} is faster or as fast as existing solutions.

\medskip

We implemented the \ourtree{} and conducted extensive experiments in \cref{sec:exp}.
A short summary is given in \cref{tab:introTable} on two datasets with $10^9$ 2D points.
In a nutshell, \ourtree{s} are much faster in all aspects. 
For construction, the performance gain is from better cache complexity---data movement can be greatly saved by constructing multiple levels in one round.
Compared to the logarithmic method (\logtree{}), the \ourtree{} is 2.02--62.0$\times$ faster on insertion, 3.27--400$\times$ faster on deletion, and 6.71--11.9$\times$ faster on queries by avoiding keeping $O(\log n)$ trees.
Compared to full reconstruction on updates (\bhltree{} and \cgal{}), the \ourtree{} is orders of magnitude faster on updates and has better query performance.
We show running time on real-world datasets, and in-depth experiments with varying dimensions, query types and parameters, individual techniques, scalability, and more, in \cref{sec:exp}.
We believe that the \ourtree{} is the first \kdtree{} that is highly performant, parallel, and dynamic.
We release our code at \cite{pkdCode}.
\ifconference{More experiments and analysis are provided in the full paper in the supplementary material. }

\hide{

	===========

	We experimentally show that even a reasonably large $\balpara$
	only mildly affects the query performance, while leading to much better update time.

	we provide \emp{high-performance parallel implementation for \kdtree{s} that supports efficient construction, updates, and queries.}
	To do this, we have to deal with several challenges, which are also
	what make state-of-the-art solutions fall short of achieving high performance.

	The first challenge is to support efficient dynamic updates on \kdtree{s}.
	\kdtree{s} differ from other classic trees
	and do not support rebalancing primitives for updates, such as overflow/underflow (as in B-trees),
	or rotations (as in binary search trees).
	Sequentially, a few solutions have been proposed (see \cref{sec:related}), but it is not clear how they incorporate parallelism.
	The recent attempt in ParGeo adopted the logarithmic method with strong bounds,
	but incurs significant overhead in queries.
	It is therefore worth asking 1) what is the best way to achieve dynamism for \kdtree{s} in parallel,
	and 2) how to achieve a tradeoff to give overall best performance and theoretical bounds in both updates and queries.

	The second challenge comes from the inherent complication of the \emph{design space} for \kdtree{s}.
	As a classic data structure, the \kdtree{} is widely-used in various applications with different design goals.
	The performance evaluation is thus monolithic,
	encompassing construction, updates, and a \emph{variety} of queries. 
	To achieve the best performance, one also has to consider factors such as work-efficiency (i.e., low time complexity), I/O-efficiency, and parallelism.
	Most existing work focused on one or a subset of such goals,
	and thus overlooked the other factors and did not incorporate them in their implementation collectively.
	A good implementation should require a co-design of all these factors, but it is highly challenging to do so.
	It is therefore worth asking 1) how to incorporate all these factors (work, parallelism, I/O) to design a full-featured interface (construction, update, queries) of \kdtree{s},
	and 2) how much each factor affects the performance.

	In this paper, we answer these questions by proposing \emp{\ourtree{} (\boldmath Parallel \kdtree{}\unboldmath)},
	a parallel \kdtree{} implementation supporting efficient construction, batch-update and multiple query types.
	This requires both new algorithmic insights to maintain good theoretical guarantees,
	and careful engineering effort to achieve high performance.
	One key question we investigated is, whether being perfectly balanced is needed in \kdtree (i.e., $\log_2n+O(1)$ tree height), or whether a reasonably relaxed criterion is good enough (e.g., $O(\log n)$ tree height).
	We conducted experimental studies and showed that the query performance of \kdtree{} is reasonably consistent with minor imbalance (see \cref{sec:balancingParameterRevisted}).
	Motivated by this observation, we adopted a \emph{randomized weight-balanced} approach, and such relaxation allows for much better performance in \emph{construction, update and query} compared to state-of-the-art implementations.
	We also provide careful \emph{theoretical analysis} to guarantee that under our relaxation, the algorithms are \emph{efficient in} \emph{work} (i.e., low time complexity),
	\emph{span} (i.e., good parallelism)\footnote{The \emph{work} of a parallel algorithm is the total number of operations (i.e., sequential time complexity),
		and its \emph{span} is the longest dependence chain. They are formally defined in \cref{sec:prelim}.}, and \emph{I/O cost}.

	To enable efficient parallel \kdtree{} algorithms both in theory and in practice,
	\ourtree{} is a careful co-design of construction and update algorithms.
	For construction, to reduce the I/O cost to find splitting hyperplane and move all data points per level,
	\ourtree{} employs a careful \emp{sampling scheme} to determine the hyperplanes for $\skheight$ levels,
	and applies an efficient parallel algorithm to \emp{sieve all points down $\skheight$ levels by one round of data movement}.
	In experiments, we observe that sampling and multi-level construction improve performance by 1.6$\times$ and 2.4$\times$, respectively.

	For batch updates (insertion and deletion),
	we design a \emp{lazy strategy} that tolerates the difference of sibling subtree sizes by a controllable factor of $\balpara$ before invoking the rebalancing scheme,
	which identifies the unbalanced substructures and performs a \emp{local reconstruction} on the affected subtrees.
	The parameter $\balpara$ controls the degree of balance, enabling a tradeoff between update and query costs---a stronger balance condition allows for better query performance (due to shallower tree height),
	at the cost of more expensive updates (due to more frequent rebalancing).
	We experimentally show that even a reasonably large $\balpara$
	only mildly affects the query performance, while leading to much better update time.

	While the high-level ideas of sampling~\cite{bentley1990k,bentley1990experiments,al2000parallel,agarwal2016parallel}, multi-level construction~\cite{reif2022scalable,agarwal2016parallel} and reconstruction-based balancing scheme~\cite{overmars1983design,galperin1993scapegoat,cai2021ikd,jo2017progressive} have
	all been studied in previous work (for both \kdtree{s} and other data structures),
	the unique challenge here is to integrate these elements cohesively in theory,
	such that they facilitate a unified implementation for \kdtree{}.
	Indeed, we are not aware of any existing work on \kdtree{s} that are theoretically-efficient in all measurements of work, span and I/O
	for either construction \emp{or} update.
	We present a more detailed description of the existing work in \cref{sec:related}.
	Such a synergy of construction and update algorithms is crucial in \ourtree{}, since its update relies on efficient reconstruction.
	\ourtree{} achieves this by 1)
	designing the sieving step as a building block by borrowing ideas from recent parallel sorting algorithms~\cite{dong2023high},
	which provides work, span and I/O bounds for both construction and updates,
	and 2) configuring the parameters for sampling, multi-level construction and rebalancing collectively.
	In theory, we present the parameter configuration
	and the cost bounds based on the parameters
	in \cref{lem:sampling,lem:tree-height,thm:update,thm:constr}.
	In practice, we provide a highly-optimized implementation that supports parallel construction, update, and three query types: $k$-nearest neighbor query (\knn{}), range count and range search.
	\hide{We also carefully surveyed existing optimizations in previous theoretical and experimental work, across different settings (sequential, parallel, distributed),
		and apply what we believe are the most relevant ones to \ourtree{s} (overviewed in \cref{sec:impl}).}

	As illustrated in \cref{tab:introTable}, \ourlib{} achieves much better performance than previous \kdtree{s} due to the newly proposed techniques. 
	On 100M points in 2D, \ourlib{} is at least 11.6$\times$ faster than all \kdtree{} baselines on construction, 13.4$\times$ on batch insertion, 4.1$\times$ on batch deletion, 1.7$\times$ on $k$-NN, and $1.5\times$ on range query.
	\ourlib{} is also faster than \zdtree{} in construction, queries, and most of batch update operations. 
	Results in higher dimension and larger data are summarized in \cref{table:summary}.
	We also design experiments to understand the performance gain of the proposed techniques in both construction and updates, discussed in \cref{sec:tech:analysis}.
	On our machine,
	\ourtree{} can process large-scale data. It constructs one billion points in 3D in about 4 seconds, and in 9D in about 10s.
	Even in 3D, none of the baseline \kdtree{s} were able to perform all point 10-NN query on 1 billion points due to space- and/or time-inefficiency, 
	while \ourtree{} only uses 11 seconds.
	We believe \ourtree{} is the first parallel dynamic \kdtree{} implementation that scales to billions of points with high performance.
	Our code is
	publicly available in~\cite{pkdCode}.
	For page limit, some proofs are deferred to the
	full version of this paper~\cite{pkdPaperFull}.

	\hide{
		In this paper, we answer these questions by proposing \emp{\ourtree{} (\boldmath Parallel \kdtree{}\unboldmath)},
		a parallel \kdtree{} implementation supporting efficient construction, batch-update and multiple query types.
		This requires both new algorithmic insights to maintain good theoretical guarantees,
		and careful engineering effort to achieve high performance.
		Instead of the strong condition of perfect balancing as in logarithmic methods, we allow the tree to be \emph{weight-balanced} with \emph{randomization}
		by re-designing both the construction and update algorithms.
		Such relaxation allows for much higher performance in \emph{both update and query} compared to all existing implementations.
		In addition, we provide careful \emp{theoretical analysis} to guarantee that under our relaxation, the algorithms are still \emph{efficient in} \emp{work} (i.e., low time complexity),
		\emp{span} (i.e., good parallelism)\footnote{The \emph{work} of a parallel algorithm is the total number of operations (i.e., sequential time complexity),
			and its \emph{span} is the longest dependence chain. They are formally defined in \cref{sec:prelim}.}, and \emp{I/O cost}.

		Our first contribution is to support efficient parallel batch update (insertion and deletion)  on \kdtree{s}.
		As mentioned, instead of making the sibling subtrees perfectly balanced,
		we allow their sizes to be off by a factor of $\balpara$.
		When the weight-balance criteria is broken, we identify the unbalanced substructure and perform a
		local rebuild on the affected subtree.
		While this local rebuild idea has been explored in various data structures sequentially (e.g., the scapegoat tree~\cite{galperin1993scapegoat}),
		applying it to \kdtree{s} in the parallel setting is highly non-trivial, let alone other desired properties such as I/O efficiency.
		We carefully design the update algorithms to be highly-parallel and I/O-efficient, while keeping the tree reasonably balanced.
		The parameter $\balpara$ captures extend of balance, which enables a tradeoff between update and query performance.
		A stronger balance condition allows for shallower tree height (and thus better query performance),
		at the cost of more expensive updates (due to more aggressive rebalancing).
		With $O(xx)$ cost to update (insertion or deletion) a batch of $m$ points, one can guarantee $O(?)$ tree height.
		Spending $O(xx)$ cost to update allows for $\log_2 n + O(1)$ tree height,
		which retains all asymptotical bounds as the perfectly-balanced \kdtree{}.

		Our second contribution is an efficient parallel construction algorithm on \kdtree{s}.
		While the original \kdtree{} is a simple algorithmic idea, existing \kdtree{} algorithms
		employ a combination of \emph{some} optimizations, such as leaf wrapping~\cite{robinson1981kdb,wang2022pargeo,wald2009state}, constructing multiple levels at a time~\cite{agarwal2016parallel,procopiuc2003bkd,agarwal2003cache,blelloch2018geometry,wald2009state},
		finding medians by sampling~\cite{agarwal2016parallel,blelloch2018geometry,wald2009state}, etc.
		Unfortunately, they feel short in achieving work-efficiency~\cite{robinson1981kdb}, I/O-efficiency~\cite{agarwal2016parallel,blelloch2018geometry,wald2009state,wang2022pargeo}, high parallel bound~\cite{agarwal2016parallel,wald2009state,procopiuc2003bkd}; some are only theoretical and do not have implementations~\cite{agarwal2016parallel,blelloch2018geometry}.
		Another unique challenge here is to incorporate the construction algorithm to work with our new reconstruction-based update algorithm.
		Since our update algorithms relies on reconstruction upon imbalancing,
		the performance and theoretical guarantee both rely on the efficiency of the construction algorithm. 
		We borrow the idea from parallel sorting algorithms~\cite{blelloch2010low} with both theoretical and practical efficiency.
		However, we need non-trivial adaptations to multi-dimensional \kdtree{s} to construct $\skheight$ levels in the tree in one round with a sampling scheme with rate $\samplerate$ to select the splitting hyperplane to reduce I/Os.
		To make construction and update algorithms compatible, we need to set the parameters $\skheight$, $\samplerate$ and the balancing parameter $\balpara$ synergistically.

		With the new proposed algorithms, our implementation is a careful co-design of all operations on \kdtree{}.
		In addition to construction and update, \ourtree{} provides three query types: kNN, range count and range report. 
		We carefully surveyed existing optimizations in previous theoretical and experimental work, across different settings (sequential, parallel, distributed), and apply what we believe are the most relevant ones to \ourtree{}.
		They all give significant improvement in performance.
		As a result, \ourtree{} outperforms all existing parallel \kdtree{} implementations by xx times [describe].
		Compared to alternative data structures supporting similar queries, \kdtree{} outperforms \zdtree{} by xx times on uniform distributed points,
		while having the benefit of being more resistant to skewed data (xx times faster), and scaling to higher dimensions (xxx).
		\ourtree{} can process very-large scale data. It constructs 1 billion 3D points in xx seconds,
		performs 10-NN query for all points in the dataset in xx seconds,
		while the other parallel \kdtree{} libraries fails to do it due to space- and time-inefficiency.
		We believe \ourtree{} is the first parallel dynamic \kdtree{} implementation that scales to billions of points with high performance.
	}
	Our contributions include:

	\begin{itemize}
		\item Parallel construction and update algorithms for parallel \kdtree{} with theoretical efficiency;
		\item Parallel library (open-source) with high performance; and
		\item In-depth experimental study on parallel \kdtree{s}.
	\end{itemize}
}

\hide{
\section{Introduction - old}

The \kdtree is one of the most widely used data structure in managing multi-dimensional data.
It can be considered as a binary search tree (BST) but in multiple dimensions.
After its invention by Jon Louis Bentley in 1975~\cite{bentley1975multidimensional}, \kdtree has been widely used in real-world applications and cited by over ten thousand times.
Compared to the counterparts such as quadtrees~\cite{finkel1974quad} and R-trees~\cite{guttman1984r,beckmann1990r,manolopoulos2006r,kamel1993packing}, \kdtree{s} are purely comparison-based (similar to regular BSTs), and support strong theoretical guarantees for construction, updates, and queries.
We will review the \kdtree more in \cref{sec:related}.

Due to the ever-growing data volume and the prevalence of parallel processors, in this paper, we consider a high-performance parallel solution for \kdtree{s} that supports efficient construction, updates, and queries.
Among these three operations, supporting efficient updates has the most challenging task with no doubt---\kdtree{s} cannot apply ``rotation-based'' rebalancing as in BSTs~\cite{sun2018pam,blelloch2016just,blelloch2022joinable}, since the \kdtree{} is a space-partitioning data structure.
Some existing algorithms~\cite{wang2020theoretically,cgal51} do not rebalance the tree, but obviously it can cause degeneracy for query performance.
Most existing solutions~\cite{robinson1981kdb, procopiuc2003bkd, agarwal2003cache, wang2022pargeo} are based on the ``logarithmic method'' proposed by Bentley and Saxe~\cite{bentley1979decomposable}---an up to $\lfloor\log_2 n\rfloor$ trees are maintained for a set of $n$ points, all with sizes as powers of 2.
The decomposition is based on the binary representation of $n$.
When a batch of insertions of size $m$ arrives, the algorithm acts as adding $n$ and $m$ as binary numbers.
When both $m$ and $n$ contains batches of size $2^i$, a new batch of $2^{i+1}$ is constructed, which may cascade for larger batches.
While the logarithmic method is conceptually simple and has some theoretical guarantees (e.g., $O(\log^2 n)$ cost per inserted elements), it has several major issues.
First, maintaining a total of $O(\log n)$ trees significantly complicates the queries on \kdtree{s} and slows down the query performance.
Second, the logarithmic method does not directly support deletions, and existing solutions usually mark tombs for deletions but it will cause overheads.
Lastly, the construction algorithm used existing solutions cannot achieve parallelism and cache efficiency simultaneously.
Combining all reasons together, in the setting as shown in \cref{tab:teaser}, the SOTA \kdtree{} from~\cite{wang2022pargeo} is slower than the SOTA quad/octree~\cite{blelloch2022parallel} by 3$\times$ on construction, 3$\times$ on insertion, 10$\times$ on deletion, and 3$\times$ $k$-nn queries---currently the \kdtree is not a competitive solution in almost all scenarios.

To achieve good performance in \kdtree{s}, 
we would like our algorithms to achieve work-efficiency (i.e., using no asymptotically more operations than the best sequential algorithm), high parallelism, and good cache performance, both in theory and in practice.
To achieve all these goals and overcome the issues in existing solutions, we need several synergic techniques.

The first goal in our \kdtree design is to avoid multiple trees as in the logarithmic method.
Our idea is to keep a single tree, and rebuild the subtrees that are imbalanced due to insertions or deletions.
This idea is known---Overmars mentioned in his PhD thesis in 1983~\cite{overmars1983design} that this is a possible solution for dynamic \kdtree.
However, Overmars only gave some preliminary thoughts and analysis, with no implementations, nor to mention parallelism or I/O efficiency.
Unfortunately, to the best of our knowledge, no later work followed up on this idea.

To fulfill this idea of reconstruction-based rebalancing efficiently, we need to overcome two major challenges.
The first challenge is an efficient algorithm for \kdtree construction that rebuilds the subtrees.
Theoretically, our goal is to achieve 1) work-efficiency (using $O(n\log n)$ operations), 2) highly parallel, and 3) good cache-efficiency (having $O(\mb{Sort}(n))$ cache bound).
We want to keep our solution simple so it should yield good practical performance and be general to different hardware platforms.
The common approach for \kdtree construction builds one level at a time, which leads to $O(\log n)$ rounds of global data movement and is thus inefficient.
Our key insight is to build multiple level of trees in one round.
Again this approach is known (cite PODS paper), but previous work did not have implementations or show the I/O and span bounds.
Here we borrow the idea from parallel sorting algorithms~\cite{blelloch2010low,dong2023high}.
This solution has two main parts---a careful but lightweight sampling scheme to balance the subproblem sizes, and a highly parallel and I/O-efficient ``transpose'' process to distribute the objects to the recursive subproblems.
However, adapting this idea to \kdtree is non-trivial since \kdtree maintains multi-dimensional data and has unique properties to support query and update efficiency.
We give more details of our approach in \cref{sec:construct}.
We believe our algorithm is conceptually easy to understand, but by setting up the parameters carefully, it can achieve strong theoretical guarantees---$O(n\log n)$ work, $O(\log^2 n)$ span, and $O(\mb{Sort}(n))=O((n/B)\log_M n)$ I/O cost---and significantly outperform any existing implementations.

\hide{
The theoretical study of dynamic \kdtree{s} dated back to at least 40 years ago, and the key idea is to use the construction algorithm for dynamizing \kdtree{s}.
The mainly adopted solution is referred to as the ``logarithmic method'' proposed by Bentley and Saxe (cite)---an up to $\lfloor\log_2 n\rfloor$ trees are maintained for a set of $n$ points, all with sizes as powers of 2.
The decomposition is based on the binary representation of $n$.
When a batch of insertions of size $m$ arrives, the algorithm acts as adding $n$ and $m$ as binary numbers.
When both $m$ and $n$ contains batches of size $2^i$, a new batch of $2^{i+1}$ is constructed, which may cascade for larger batches.
One can show that the amortized cost per element is $O(\log^2 n)$ (check), and this logarithmic method is widely studied in the literature (check and cite all relavant papers in [1]-[8] from Julian's paper).
Despite the theoretical efficiency, the logarithmic method, unfortunately, does not perform well in practice.
The main issue is the existence of $O(\log n)$ trees, which significantly complicates the queries.
(Evidence needs to be supplemented here.  I read Julian's paper again and they said query time is the same as zd-tree, which I don't quite convinced.)

Therefore, this paper seeks for a single-tree solution and efficient update algorithms on top of it.
Since only one \kdtree is kept, the query cost will naturally remains efficient as for the static \kdtree{} version.
Our key idea is to use the ``reconstruction-based'' rebalancing---once a subtree becomes ``sufficiently'' unbalanced, we will simply rebuild this subtree into a perfectly balanced shape.
As far as we know, this idea was first proposed by Lueker~\cite{lueker1982data} on B-trees, and later reinvented multiple times~\cite{overmars1982dynamic,galperin1993scapegoat,andersson1989improving}.
The only known work on applying this to \kdtree{s} was \cite{overmars1982dynamic}, but only very briefly with some theoretical anlaysis---no implementation and experiments were given, nor to mention parallelism and batch updates.
}

~

~

The rest of the story: to make it work, we need efficient construction.
We will use I/O efficient distribution, but finding the partition hyperplanes are hard.
We then proposed sampling based solution.
Analyze I/O, work, span, and space.

Using them, we can then design our dynamic version.
Give theoretical analysis.

\ziyang{No parallel partial and batch, no theoretical analysis (work, dpeth), no high performance implementation}

Experimentally, ...

}

\section{Preliminaries}
\label{sec:prelim}

We present a table of notations used in this paper in \cref{tab:notations}.
We use \bdita{with high probability} (\whp{}) in terms of $n$ to mean probability at least $1-n^{-c}$ for any constant $c>0$.
With clear context, we omit ``in terms of $n$''.
We use $\log n$ as a short term of $\log_2 (1+n)$.


\myparagraph{Computational Model.}
We analyze our algorithms using the \bdita{work-span model} in the classic fork-join paradigm with {binary-forking}~\cite{BL98,arora2001thread,blelloch2020optimal}. We assume multiple threads that share memory. Each thread is a sequential RAM augmented with a \textsf{fork} instruction, which spawns two child threads that run in parallel. The parent thread resumes execution upon the completion of both child threads. A parallel for-loop can be simulated by a logarithmic number of steps of forking.
A computation can be viewed as a directed acyclic graph (DAG). The \bdita{work $(W)$} of a parallel algorithm is the total number of operations within its DAG (aka.\ time complexity in the sequential setting), and the \bdita{span $(S)$} depicts the longest path in the DAG.
Using a randomized work-stealing scheduler, a computation with work $W$ and span $S$ can be executed in $W/\rho+O(S)$ time \whp{} (in~$W$) with $\rho$ processors~\cite{BL98,arora2001thread,gu2022analysis}.


We use the classic \bdita{ideal-cache model}~\cite{Frigo99} to measure the cost of memory accesses, which is
widely used to analyze the cache complexity of algorithms~\cite{BG2020,arge2004cache,dong2024parallel,blelloch2010low}.
In this model, the memory is divided into two levels.
The CPU is connected to the small-memory (aka.\ the cache) of finite size $M$, which is connected to a large-memory (the main memory) of infinite size. Both memories are organized as \bdita{blocks} with size $B$. 
The CPU can only access the data in small-memory with free cost, and there is a unit cost to transfer one block from large-memory to small-memory, assuming the optimal offline cache replacement policy.
The cache complexity of an algorithm is number of block transfers during the algorithm.

The ideal-cache model assumes optimal eviction strategy for theoretical analysis. It has been shown that
practical cache policies (e.g., LRU) enables the same asymptotic cache I/Os as the optimal strategy~\cite{Frigo99,Sleator85}.
The ideal-cache model is only used for theoretical analysis.
In our implementation, we do not control the cache, so the optimal eviction strategy is guaranteed.
In reality, the eviction strategy is usually a combination of multiple strategies, considering more complicated components such as set associativity, parallelism, and a few optimizations.
However, the theoretical model is still extremely widely used as a good estimation of the practice.

\myparagraph{The \boldmath $k$d-Tree. \unboldmath}
We study points in Euclidean space in $D$ dimensions,
and the distance between two points is their Euclidean distance.
A partition hyperplane can be represented by a pair $\ip{d}{x}$,  where $d$ ($1\le d\le D$) is the \emp{splitting dimension} and $x\in \R$ is the \emp{splitting coordinate}.
We refer to such a pair $s$ as a \bdita{splitter}.

The \bdita{$k$d-tree} ($k$-dimensional tree), is a spatial-partitioning binary tree data structure.
To avoid overloading the commonly-used parameter $k$ in \knn{} query, we use $D$ to refer to the number of dimensions of the dataset.
Each interior (non-leaf) node in a \kdtree{} signifies an axis-aligned splitting hyperplane $\ip{d}{x}$.
Points to the left of the hyperplane (those with the $d$-th dimension coordinates smaller than $x$)
are stored in the left subtree and the remaining are in the right subtree.
Each subtree is split recursively until the number of points drops below a \emph{leaf wrap} threshold $\leafwrap$ (a small constant),
where all the points are directly stored in a leaf.
Common approaches for choosing the dimension of the splitting hyperplane include cycling among the $D$ dimensions~\cite{bentley1975multidimensional}, choosing the dimension with the widest stretch~\cite{friedman1977algorithm}, etc.
\ourtree{} also uses the widest dimension as the cutting dimension.
The cutting coordinate $x$ is usually the median of the points on the $d$-th dimension, yielding two balanced subtrees.
Given $n$ points in $D$ dimensions, a balanced \kdtree{} has a height of $\log_2n+O(1)$ and can be constructed in $O(n\log n)$ work using $O(n)$ space.

The \kdtree{} can answer various types of queries.
Since the \ourtree{} remains a single \kdtree{}, the same query algorithms for classic \kdtree{s} also work on \ourtree{s}.
In our experiments, we focus on \emph{\knn{} queries} (finding the $k$ nearest points to a query point), rectangle \emph{range report queries} (reporting all points within an axis-aligned bounding box) and
rectangle \emph{range count queries} (reporting the number of points within an axis-aligned bounding box).



We use the (subtree) root pointer $T$ to denote a (sub-)$k$d-tree.
With clear context, we also use $T$ to represent the set of all points in~$T$.
Every interior node in $T$ maintains two pointers $T.\lc$ and $T.\rc$ to its left and right children, respectively. 
As we mentioned, \ourtree{} is \emph{weight-balanced}.
Given the balancing parameter $\balpara\in[0,0.5]$,
we say a \kdtree{} is (weight-)balanced if $0.5-\balpara \le |T.\lc|/|T| \le 0.5+\balpara$, and \emph{unbalanced} otherwise. 
Essentially, this means that the two subtrees can be off from perfectly balanced by a factor of $\balpara$.

\begin{table}
	\small
	\begin{tabular}{|>{\boldmath}l<{\unboldmath}p{3cm}>{\boldmath}l<{\unboldmath}p{3.2cm}|}
		\hline
		$T$         & \multicolumn{3}{l|}{a (sub-)\kdtree, also the set of points in the tree}                                   \\
		$\leafwrap$ & \multicolumn{3}{l|}{leaf wrap threshold (leaf size upper bound)}                                           \\
		$k$         & \multicolumn{3}{l|}{required number of nearest neighbors in a query}                                       \\
		$\skheight$ & \multicolumn{3}{l|}{number of levels in a tree sketch (i.e., that are built at a time)}                      \\
		$\matht$    & \multicolumn{3}{l|}{tree skeleton at $T$ with maximum levels $\skheight$} \\
		$P$         & \multicolumn{3}{l|}{input point set (for updates, $P$ is the batch to be updated)}                        \\
		$T.\textit{lc}$      & left child of $T$ & $T.\textit{rc}$  & right child of $T$                                                  \\
		$n$         & tree size              &
		$m$         & batch size for batch updates                                                          \\
		$D$         & number of dimensions                                                                  &
		$d$         & a certain dimension                                                                   \\
		$S$         & samples from $P$                                                                      &
		$s$         & size of the $S$                                                                       \\
		$\os$       & oversampling rate                                                                     &
		$\balpara$  & balancing parameter                                                                   \\
		$M$         & small-memory (cache) size                                                             &
		$B$         & block (cacheline) size                                                                        \\
		\hline
	\end{tabular}
  \caption{\textbf{Notations used in this paper.} }
	\label{tab:notations}
\end{table}

\hide{
\begin{table}
	\small
	\begin{tabular}{|cp{7.5cm}|}
		\hline
		$M$         & small-memory (cache) size                                                             \\
		$B$         & cacheline size                                                                        \\
		$D$         & number of dimensions                                                                  \\
		$d$         & a certain dimension                                                                   \\
		$k$         & required number of nearest neighbors in a query                                       \\
		$T$         & a (sub-)\kdtree, also the set of points in the tree                                   \\
		$T.lc$      & left child of $T$, so as for $T.rc$                                                   \\
		$\skheight$ & number of levels in tree sketch (i.e., that are built at a time)                      \\
		$\matht$    & shorthand of $\matht_\skheight$, tree skeleton at $T$ with maximum levels $\skheight$ \\
		$P$         & input point set (for updates, $P$ is the batch to be updated)                        \\
		$n$         & tree size              \\
		$m$         & batch size for batch updates                                                          \\
		$S$         & samples from $P$                                                                      \\
		$s$         & size of the $S$                                                                       \\
		$\os$       & oversampling rate                                                                     \\
		$\balpara$  & balancing parameter                                                                   \\
		$\leafwrap$ & upper bound for size of leaf wrap                                                     \\
		\hline
	\end{tabular}
	\caption{Notations used in this paper. }
	\label{tab:notations}
\end{table}
}

\hide{
	\myparagraph{Metric Space.}
	The \bdita{metric space} is a set of elements, usually called points, together with a notion of distance function between its elements. In the formal setting, one is giving a metric space $\mathcal{M}(M,d)$, where $M$ is the collection of points and $d:M\times M\to \R$ is the distance function. For all pairs of points $x,y\in M$, $d$ satisfies following axioms: (1) $d(x,y)\geq 0$ and $d(x,y)=0 \Leftrightarrow x=y$, (2) symmetry, (3) triangle inequality. For example, the $k$-dimensional Euclidean metric space is defined as $\mathcal{M}=(\R^k, L_p)$.

	Let $S\subseteq M$ be a finite point set and denote $B_p(r)\subseteq S$ the set of points enclosed by a ball with radius $r$ centered at $p$. Then $S$ has $(\rho,c)-$\bdita{expansion} if and only if $\forall p\in M$ and $r>0$
	\begin{equation}
		|B_p(r)|\geq \rho \Rightarrow |B_p(2r)|\leq c\cdot|B_p(r)|
	\end{equation}
	The constant $c$ is refereed to \bdita{expansion rate} and $\rho$ is usually set to be $O(\log |S|)$. We say the expansion rate is \bdita{low} if $c=O(1)$.
	The coefficient $2$ above can be substituted to any other constant with a suitable modification of $c$. Applies the expansion property recursively implies that the $|B_p(r)|\leq poly(r)$. Intuitively, the low expansion property ensures the points distributed "smoothly" in the space, i.e., it rules out the possibility that a region of dense points locate next to a large patch of empty space.
}
\hide{
	Another common assumptions used in computation geometry and computer graphics to describe a property of geometric shapes or structures is called the \bdita{bounded aspect ratio}, whose focus is about maintaining a balance between the different dimensions of a shape to avoid extreme elongation or flattening while working with geometric structures. Formally, given metric space $\mathcal{M}(M,d)$ and a subset $S\subseteq M$ containing $n$ points, the aspect ratio $\Delta$ is defined as:

	\begin{equation}
		\Delta = \frac{\max d(x,y)}{\min d(x,y)}\quad \forall x,y\in S
	\end{equation}
	and is said to be \textit{bounded} if $\Delta< n^c$ holds for some constant $c>0$.
}
\hide{
	\myparagraph{Range spaces and $\epsilon$-approximation}

	Given a finite set $(X,\mathcal R)$, and a parameter $0\leq \epsilon\leq 1$, a set $A\subseteq X$ is called an $\epsilon$-approximation if, for each $R\in\mathcal R$,
	\begin{equation}
		\left|\frac{|R|}{|X|}-\frac{|R\cap A|}{|A|}\right|\leq \epsilon\label{epsilonApproximation}
	\end{equation}
	\ziyang{ADD more explanation and paraphrase124}.
}



\section{Parallel Algorithm for Tree Construction}\label{sec:constr}

We start with our parallel $k$d-tree construction algorithm.
Constructing a \kdtree{} requires partitioning the points into nested sub-spaces recursively 
based on the median of the splitting dimension. 
It directly implies a parallel construction algorithm with $O(n\log n)$ work and $O(\log^2 n)$ span 
using the standard parallel partition algorithm ($O(n)$ work and $O(\log n)$ span). 
However, it requires $O(\log n)$ rounds of data movement and is not cache-efficient when the input is larger than the cache size.
We will refer to this algorithm as the \emph{plain parallel \kdtree{} construction} algorithm.
This is also the algorithm used by the \bhltree{} from ParGeo~\cite{wang2022pargeo}.

Instead of partitioning all points into the left and right subtrees and pushing the points to the next level in the recursive calls,
the high-level idea of our approach is to build $\skheight$ levels at a time by one round of the data movement.
To avoid the data movement for finding splitting coordinates,
our algorithm uses samples to decide all splitters for $\skheight$ levels.
It then distributes all points into the corresponding subtrees ($2^{\skheight}$ of them) and recurses.

The main challenges here are 1) to use only a subset of the points (samples) to decide the splitters and make the tree nearly balanced, and
2) to move each point exactly once to its final destination in a cache-efficient and parallel manner.
Below, we will first elaborate on our parallel and cache-efficient construction algorithm in \cref{sec:constr-algo}, and show the cost analysis in \cref{sec:constr-analysis}.

\hide{
	Using the samples, we find $2^{\skheight}$
	after which the points are distributed into ranges divided by the samples. Recursion continues until the base case is met. Sequentially this method is much slower than the median-based strategy mentioned above due to the overhead of sub-lists division and median finding in each level~\cite{al2000parallel}. In the parallel setting, we resolve this issue by constructing the interior tree from samples directly and then sieving points to the corresponding bucket using the parallel algorithm for the \bdita{semisort} problem.
}


\begin{figure*}[t]
	\includegraphics[width=\textwidth]{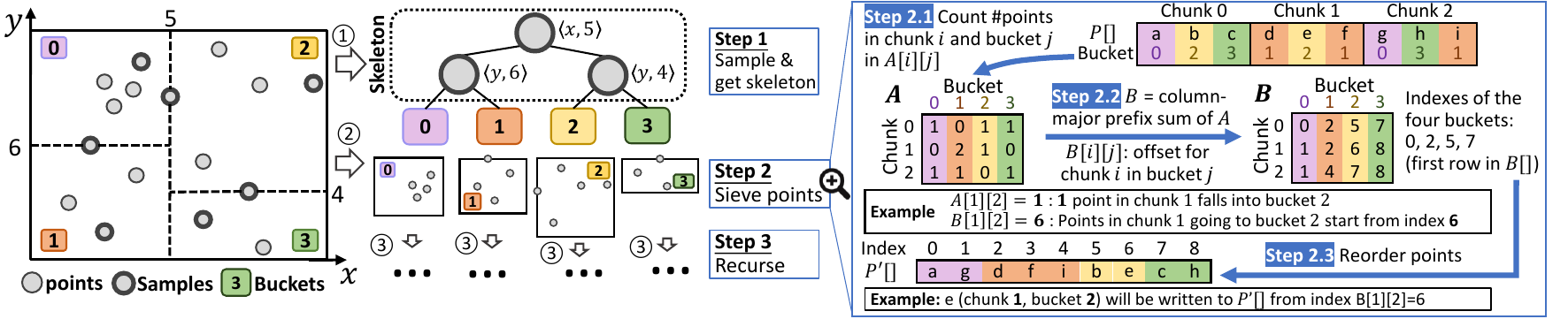}
	\vspace*{-1em}
	\caption{{\bf An illustration of our \kdtree{} construction algorithm}, with a detailed overview on the \emph{sieving step}.
		In this example, we first sample seven points and construct the tree skeleton using the samples, dividing the plane into four regions (buckets). Next, we sieve all points into the corresponding bucket. Concretely, we divide the points $P$ into chunks of size $l=3$. All chunks are processed in parallel. For each chunk, we count the number of points for every bucket in array $A$. We then compute the exclusive prefix sum of $A$ in column-major order to get the offset matrix $B$. 
			We then move all points from $P$ to $P'$ so that points within each bucket are contiguous.
			Finally, we recursively construct each subtree (a bucket in $P'$) in parallel.
		} 
	\label{fig:constr}
\end{figure*}

\hide{
We commence by introducing some necessary notation. Let $P$ denote the set of $n$ points and given a $k$d-tree $T$ over $P$, define \bdita{tree skeleton} $\mathcal T_\skheight$ as a (partial) subtree rooted at $T$ with height $\skheight\in\N$. Specifically, the "leaf" nodes in $\matht_\skheight$ are referred to the \bdita{bucket node} or simply the \bdita{bucket}, i.e., the interior node in $T$ with depth $\skheight$ or the leaf node with depth $\leq\skheight$. A tree skeleton $\matht_\skheight$ contains at most $2^\skheight$ bucket nodes; each is assigned with a unique identifier within $[ 0,2^\skheight)$, called the \bdita{bucket id}.
Without further clarification we omit the constant subscript $\skheight$ when mentioning the tree skeleton $\matht_\skheight$.
}

\hide{\myparagraph{Construct multiple levels at once.} Previous tree construction algorithms construct one level of tree after each recursion, but this is not necessary since the skeleton of a tree has been fixed once the splitters are determined. The primary benefit of extending the tree depth for multiple levels is to avoid the intermediate partition in each level so that fewer iterations are required to reach the base case, resulting in better parallelism due to asymptotically less work and span. Besides, since there is no redundant write for partition within each level, the whole framework is more cache-friendly and I/O-efficient.
}

\subsection{Algorithms Description}\label{sec:constr-algo}
We present our algorithm in \cref{algo:constr} and an illustration in \cref{fig:constr}.
The algorithm $T=\build(P)$ builds a \kdtree{} $T$ on the input points in array $P$.
As mentioned, our main idea is to use samples to decide all splitters in $\skheight$ levels.
We define the \emph{skeleton} at $T$, denoted as $\matht$, as the substructure consisting of all splitters (and thus interior nodes) in the first $\skheight$ levels.
We use the samples to build the skeleton. In particular, we will uniformly take $2^{\skheight}\cdot\os$ samples from $P$, where $\sigma$ is the \emph{over sampling rate}.
In \cref{sec:constr-analysis} we will show how to choose the parameter $\os$ to achieve strong theoretical guarantees.
Let $S$ be the set of sample points. The skeleton will be the first $\skheight$ levels of the \kdtree{} on $S$.
As we will show in \cref{sec:constr-analysis},
we will keep $S$ small and fit in cache, so that the skeleton can be built by the plain parallel \kdtree{} construction algorithm at the beginning of \cref{sec:constr}.

The skeleton depicts the first $\skheight$ levels of the tree, and splits the space into $2^{\skheight}$ subspaces, corresponding to the external nodes (leaves) of the skeleton.
We call each such external node a \emp{bucket} in this skeleton.
We label all buckets from $0$ to $2^{\skheight}-1$.
The problem then boils down to sieving the points into the corresponding bucket, so that we can further deal with each bucket recursively in parallel.
We first note that the target bucket for each point can be easily looked up in $O(\skheight)$ cost by searching in the skeleton.
Sequentially, one can simply move all points one by one to their target bucket, by maintaining a pointer to the (current) last element in each bucket.
In parallel, the key challenge is to independently determine the ``offset'' of each point,
so the points can be moved to their target buckets
in parallel without introducing locks or data races.

We borrow the idea from the cache-efficient parallel sorting algorithm~\cite{blelloch2010low,dong2023high,dong2024parallel}, which also involves redistributing elements into $\omega(1)$ buckets.
Our goal is to reorder array $P$ and make all points belonging to the same bucket to be contiguous, so that the next recursion receives a consecutive input slice.
To do this, we divide the array into chunks of size $l$, and process them in parallel.
We first count the number of elements in chunk $i$ that fall into bucket $j$ in $A[i][j]$.
Note that there is no data race since we count all points sequentially within each bucket.
Then we compute the exclusive prefix sum of matrix $A$ in column-major and get the offset matrix $B$---
i.e., we consider storing matrix $A$ in column major in an array, and compute the exclusive prefix sum at each element. 
This can be done by a parallel cache-efficient matrix transpose~\cite{blelloch2010low} and a standard parallel prefix-sum~\cite{Blelloch89}.
As such, $B[i][j]$ implies the offset when writing a point in chunk $i$ that belongs to bucket~$j$.
We present an illustration for this process in \cref{fig:constr}.
Then we process all buckets again in parallel, and move each point to its final destination by using the offsets provided from matrix $B$ as the starting pointer.
There is still no data race here, since all points that ``share'' the same offset must be in the same chunk and will be processed sequentially.

After all points in the same bucket are placed consecutively, we recursively build \kdtree{s} for each bucket in parallel. The recursion stops
when the number of points is smaller than $2^{\skheight}\cdot \os$.
We then switch to the base case and use the plain parallel \kdtree{} construction to build the subtree.
We will later show that by setting proper values for $\skheight$ and $\os$, the base case fits into the cache and using the plain parallel construction will not incur extra memory accesses.

\hide{
	We construct a tree skeleton from samples.
	Note that to obtain a balanced tree, one should avoid sample splitters from input points directly, since one group of skewed samples immediately leads to an imbalanced points partition.
	We borrow the idea from~\cite{agarwal2016parallel} to construct the tree skeleton from an $\epsilon$-approximation samples, i.e., sample $\log n$ points from $P$, select the median in samples as the splitter for skeleton root, partition samples using this splitter, then recurse until the tree skeleton grows to height $\skheight$.
	There is a trade-off between the tree quality and the required work: a higher sample size leads to a more balanced partition in each round, while the cost to sample and build the tree skeleton increases as well.
}

\hide{
	\myparagraph{Leveraging semisort for points sieving.}
	Deriving an efficient points sieving scheme is not trivial in parallel since one has to keep the work-efficiency as well as dealing with the race and extra space usage.  We notice that if we assign every leaf (bucket) node in the current tree skeleton a key and record each input point by the key associated with the leaf node it should be sieved in, then distributing points into corresponding leaves equivalent to solving the \bdita{semisort problem} with the points records, i.e., reorder the points so that those with same keys are contiguous. Since it costs constant time to get the key for each point considering the constant height tree skeleton, any work-efficient parallel semi-sort algorithm directly yields a work-efficient points sieving algorithm.
}

\setlength{\algomargin}{0em}
\hide{
\begin{algorithm}[t]
	\fontsize{8pt}{8.5pt}\selectfont
	\caption{Build a $k$d-tree in parallel\label{buildtree}}
	\SetKwFor{parForEach}{parallel-foreach}{do}{endfor}
	\KwIn{Sequence of points $P$ with dimension $k$.}
	\KwOut{A $k$d-tree over $P$.}
	\SetKw{MIN}{min}
	\SetKw{MAX}{max}
	\SetKwProg{MyFunc}{Function}{}{end}
	\SetKwInOut{Note}{Note}
	\SetKwInOut{Parameter}{Parameter}
	\Parameter{$\skheight$: the maximum height of a tree skeleton.\\
		$\leafwrap$: the leave wrap of the $k$d-tree.
	}
	\SetKwFor{inParallel}{in parallel:}{}{}
	
	\DontPrintSemicolon
	\vspace{.5em}
	\tcp{Build a $k$d-tree on points $P$}
	\MyFunc{\upshape{\build{$(P)$}}}{
		$n\gets |P|$\tcp*[f]{Size of the input points}\\
		\If(\tcp*[f]{Wrap the points with a leaf node}){$n<\leafwrap$\label{checkleafbuild}}{
			Allocate a leaf node containing $P$ \label{allocateleaf}\\
			\Return The pointer to the leaf node \label{returnleaf}
		}
		$S\gets$ Uniformly sample $\log n$ points on $P$\label{sample}\\
		$\matht\gets \buildTreeSkeleton(S,0)$\label{startbuildskeleton}\\
		Sieve points to the corresponding bucket node in $\matht$, array slice $B[i]$ holds all points sieved to bucket node $i$ \tcp*[f]{Partition points}\label{distributing}\\
		\parForEach{bucket id $i$}{
			$T_i\gets\build(B[i])$\label{recursive}\tcp*[f]{Recursion on points $B[i]$}\\
			Replace the bucket node in $\matht$ whose id is $i$ with $T_i$ and redirect its parent's pointer to $T_i$\label{replacing}
		}
		\Return The root of $\matht$\label{returnroot}
	}
	\vspace{.5em}
	\tcp{Build a tree skeleton with height $\skheight$ using sample points $S$ sequentially}
	\MyFunc{\upshape{\buildTreeSkeleton}$(S,h)$\label{buildskeleton}}{
		\If(\tcp*[f]{reach the last layer of tree skeleton}){$h\geq \skheight$}{
			Allocate a bucket node and assign it a new identifier\label{assignidentifier}\\
			\Return The pointer to the bucket node
		}
		$d\gets$ Pick the split dimension \label{dim}\\
		$m\gets$ Median coordinate of $S$ on dimension $d$\label{median}\\
		$L\gets\{p\in S~|~p \mtext{ is strictly to the left of the splitter }\ip{d}{m}\}$ \label{partitionleft}\\
		$R\gets S\setminus L$\label{partitionright}\\ 
		$\lson\gets \buildTreeSkeleton(L,h+1)$\\
		$\rson\gets \buildTreeSkeleton(R,h+1)$\\
		Create an interior node, assign its left child to $\lson$, right child to $\rson$ and splitter to $\ip{d}{m}$ \\
		\Return The pointer to the interior node\label{buildskeletonend}
	}

\end{algorithm}
}

\begin{algorithm}[t]
	\fontsize{8pt}{8.5pt}\selectfont
	\caption{Parallel $k$d-tree construction\label{algo:constr}}
	\SetKwFor{parForEach}{parallel-foreach}{do}{endfor}
	\KwIn{A sequence of points $P$.}
	\KwOut{A $k$d-tree $T$ on points in $P$.}
	\SetKw{MIN}{min}
	\SetKw{MAX}{max}
	\SetKwProg{MyFunc}{Function}{}{end}
	\SetKwInOut{Note}{Note}
	\SetKwInOut{Parameter}{Parameter}
	\Parameter{$\skheight$: the height of a tree skeleton.
	}
	\SetKwFor{inParallel}{in parallel:}{}{}
	
	\DontPrintSemicolon
	\vspace{.5em}
	\MyFunc{\upshape{\build{$(P)$}}}{
        \tcp{Base case}
        \lIf{$|P|<2^{\skheight}\cdot \os$\label{checkleafbuild}}{
            Use the plain parallel construction and \Return{}
		}
		$S\gets$ Uniformly sample $2^{\skheight}\cdot \os$ points on $P$\label{sampleBuild} with replacement\\
        Build tree skeleton $\matht$ by constructing the first $\skheight$ levels of a \kdtree{} on $S$\\
        \tcp{Sieve each point to their corresponding bucket (external node) in $\matht$. 
        This is performed by reordering all points in $P$ to make points in the same bucket consecutive. 
        }
        $R[]\gets \textsc{Sieve}(P, \matht)$\tcp*[f]{$R[i]$: the sequence slice for all points in bucket $i$}\label{line:obtainSlices}\\
		\parForEach{external node $i$}{
			$t\gets \build(R[i])$\label{recursive}\tcp*[f]{Recursively build a \kdtree{} on $R[i]$}\\
            Replace the external node with $t$\\
		}
		\Return The root of $\matht$\label{returnroot}
	}
	\vspace{.5em}
	\tcp{Sieve points in $P$ to the buckets (external nodes) in $\matht$}
	\MyFunc{\upshape{\textsc{Sieve}}$(P,\matht)$\label{line:sieve}}{
    (Conceptually) divide $P$ evenly into chunks of size $l$\\
    \parForEach{chunk $i$}{
    \For{point $p$ in chunk $i$\label{line:seq-loop-1}} {
        $j\gets$ the bucket id for $p$ by looking up $p$ in $\matht$\label{line:loopup-1}\\
        $A[i][j]\gets A[i][j]+1$
    }
    }
    Get the column-major prefix sum of $A[i][j]$ as matrix $B$\label{line:transpose}\\
    \parForEach{bucket $j$}{
      Let $s_j\gets B[0][j]$ be the offset of bucket $j$
    }
    \parForEach{chunk $i$\label{line:beginDistribution}}{
    \For{point $p$ in chunk $i$\label{line:seq-loop-2}} {
        $j\gets$ the bucket id for $p$ by looking up $p$ in $\matht$\label{line:loopup-2}\\
        $P'[B[i][j]]\gets p$\\
        $B[i][j]\gets B[i][j]+1$\label{line:endDistribution}\\
    }
    }
    Copy $P'$ to $P$\label{line:copyback}\\
    \parForEach{bucket $j$}{
      $R[j]\gets$ the slice $P[s_j..s_{j+1}-1]$\tcp*[f]{for the last bucket, $s_{j+1}=|P|$}
    }
    \Return{$R[]$}
}
\end{algorithm}

		\hide{
        \If(\tcp*[f]{Wrap the points with a leaf node}){$|P|<\leafwrap$\label{checkleafbuild}}{
            \Return{A leaf node containing $P$}
		}
        }

\hide{
\cref{buildtree} illustrates our parallel tree construction algorithm in detail. We receive the input points $P$ where every point has dimension $k$ and checks whether it is possible to put the points within a leaf node (\cref{checkleafbuild}) at first. The algorithm continues if the size of $P$ is larger than the leave wrap $\tau$ , otherwise it allocates a leaf node containing $P$ and return this node (\cref{allocateleaf}~-~\cref{returnleaf}). The points partition starts by sampling $\log n$ candidates from $P$ (\cref{sample}) and then construct a tree skeleton $\matht$ based on the sample $S$ (\cref{startbuildskeleton}) sequentially. The construction of $\matht$ (\cref{buildskeleton}~-~\cref{buildskeletonend}) is almost the same as the traditional algorithm used to build a $k$d-tree: determine the split dimension $d$ first (\cref{dim}), find the median coordinates on this dimension (\cref{median}), then partition the points accordingly (\cref{partitionleft}~-~\cref{partitionright}). The difference is that when the skeleton height reaches $\skheight$, it allocates a bucket node instead. Each bucket node receives an identifier, i.e., an unique integer within range $[0,2^\skheight)$, in the DFS order (\cref{assignidentifier}). Leveraging the tree skeleton $\matht$ and the distributing technique introduced in~\cite{dong2023high}, one can sieve the points into contiguous array slices ordered by the bucket id efficiently. Namely, the points that are sieved into the bucket $i$ are collected within an array slice $B[i]$ (\cref{distributing}) and for all identifier $i\in[1,2^\skheight-1)$, $B[i-1], B[i], B[i+1]$ are contiguous slices of one single array. Once the partition finishes, the position of every point has been fixed, so points in different slice are independent, we then launch $2^\skheight$ threads to construct the $k$d-tree over points in every slice in parallel (\cref{recursive}). The recursion returns a tree node $T_i$ containing points sieved to the bucket $i$, therefore replacing the $i$-th bucket node with $T_i$ yields a valid (sub-)$k$d-tree $\matht$ (\cref{replacing}). In \cref{returnroot}, the tree skeleton root $\matht$ is returned to mark the end of the recursion.
}

\subsection{Theoretical Analysis}\label{sec:constr-analysis}

We now formally analyze our construction algorithm and show its theoretical efficiency.
We start with a useful lemma about sampling.
Similar results about sampling have also been shown previously \cite{RR89,agarwal2016parallel,blelloch2018geometry}.
We put it here for completeness.

\begin{lemma}\label{lem:sampling}
	For a \ourkdtree $T$ with size $n'$,
	for any $\epsilon<1$,
	setting $\os=(6c\log n)/\epsilon^2$ guarantees that the size of a child subtree is within the range of $(1/2\pm \epsilon/4)\cdot n'$ with probability at least $1-2/n^c$.
\end{lemma}
Here we need to distinguish a subtree size $n'$ and the overall tree size $n$ for a stronger high probability guarantee, so we can apply union bounds in the later analysis.
\medskip
\begin{proof}
	\cref{algo:constr} guarantees that each leaf in a skeleton has at least $\os$ sampled points.
	Therefore, every time we find a splitter,
	it is the median of at least $2\os$ sampled points.
	Let $s$ ($\ge 2\os$) be the number of samples for this \ourkdtree $T$,
	and $\Lambda\subset T$ contains the smallest $(1/2 - \epsilon/4)\cdot n'$ points in the cutting dimension.
	We want to show that the chance we have more than $s/2$ samples in $\Lambda$  (i.e., the left side has fewer than $(1/2 - \epsilon/4)\cdot n'$ points) is small.
	Since all samples are picked randomly, we denote indicator variable $X_i$, where $X_i=1$ if the $i$-th sample is in $\Lambda$ and 0 otherwise.
	Let $X=\sum X_i$ for $i=1..|\Lambda|$, and $\mu = \mathbb E[X] = (1/2-\epsilon/4)s$.
	Let $\delta = \epsilon/(2-\epsilon)$. Then $(1+\delta)\mu=(1+\frac{\epsilon}{2-\epsilon})(\frac{1}{2}-\frac{\epsilon}{4})s=\frac{s}{2}$.
	Using the form of Chernoff Bound $\Pr[X\geq (1+\delta)\mu]\leq \exp(-{\delta^2\mu}/(2+\delta))$,
	we have:
	\begin{align*}
		\Pr\left[X\geq \frac{s}{2}\right] & \leq \exp\left(-\frac{\delta^2\mu}{2+\delta}\right) = \exp\left(-\frac{\epsilon^2 s}{16-4\epsilon}\right)                    & \hfill\!\!\!\!\!\!\!\!\!\!\!(\text{plug in }s\ge 2\os) \\
		                                  & \le\exp\left(-\frac{12c\log n}{16-4\epsilon}\right)=\frac{1}{n^{c\cdot{\frac{12\log_2 e}{16-4\epsilon}}}} \le \frac{1}{n^c}.
	\end{align*}
	The right subtree has the same low probability of being unbalanced, so taking the union bound gives the state bound.
\end{proof}

\hide{
	\begin{lemma}[Sampling]\label{lem:sampling}
		For a \ourkdtree $T$ with size $n'$, setting $\os=\Omega((\log n)/\epsilon^2)$ guarantee a child subtree within $(1/2\pm \epsilon/4)\cdot n'$ nodes for any $\epsilon<1$ \whp{} in $n$.
	\end{lemma}
	Here we need to distinguish a subtree size $n'$ and the overall tree size $n$ for a stronger high probability guarantee, so we can apply union bounds in the later analysis.
	\medskip
	\begin{proof}
		We show that by setting $\os=(12c\log n)/\epsilon^2$, the probability that the child subtree sizes drop below $(1/2\pm \epsilon/4)\cdot n'$ is no more than $1/n^c$.
		This is equivalent to the statement in the lemma.
		We first consider the left subtree $T.lc$, and the other case is symmetric.

		In \cref{algo:constr}, every time we find the pivot, it is the median of at least $\os$ sampled points.
		Here let $s$ ($\ge \os$) be the sample size of this \ourkdtree (the current subtree), and $S'$ contains the smallest $(1/2\pm \epsilon/4)\cdot n'$ points in the cutting dimension.
		We want to show that the chance we have more than $s/2$ samples in $S'$ is small.
		Since all samples are picked randomly, we denote indicator variable $X_i=1, i\in S'$ if $i$-th sample is in $S'$ and 0 otherwise.
		Let $X=\sum_{i\in S'}X_i$ and $\mu = \mathbb E[X] = (1/2-\epsilon/4)s$.
		Let $\delta = \epsilon/(2-\epsilon)$, and we have $(1+\delta)\mu=(1+\frac{\epsilon}{2-\epsilon})(\frac{1}{2}-\frac{\epsilon}{4})s=\frac{s}{2}$.
		Using the form of Chernoff Bound $P[X\geq (1+\delta)\mu]\leq \exp(-{\delta^2\mu}/(2+\delta))$,
		we have:
		\begin{align*}
			\mathbb P[X\geq \frac{s}{2}] & \leq \exp\left(-\frac{\delta^2\mu}{2+\delta}\right) = \exp\left(-\frac{\epsilon^2 s}{16-4\epsilon}\right)                  \\
			                             & \le\exp\left(-\frac{12c\log n}{16-4\epsilon}\right)=\frac{1}{n^{c\cdot{\frac{12\log e}{16-4\epsilon}}}} \le \frac{1}{n^c}.
		\end{align*}
		The right subtree has the same low probability to be unbalanced, so taking the union bound gives the stated bound.
		\hide{
			We can partition the input points $P$ into three groups:
			\begin{itemize}
				\item $P_L = \{x\in P: rank(x)\leq n/2-\epsilon/4\}$
				\item $P_M = \{x\in P: n/2-\epsilon/4 < rank(x) < n/2+\epsilon/4\}$
				\item $P_R = \{x\in P: rank(x)\geq n/2+\epsilon/4\}$
			\end{itemize}
			Let $s=|S|$ the sample size. If fewer than $s/2$ points from both $P_L$ and $P_R$ are presented in $S$, then the median from $S$ is the desired one. Denote indicator variable $X_i=1, i\in S$ if $i$-th sample is in $P_L$ and 0 otherwise. Let $X=\sum_{i\in S}X_i$ and $\mu = \mathbb E[X] = (1/2-\epsilon/4)s$. We have:
			\begin{align}
				\mathbb P[X\geq \frac{s}{2}] & = \mathbb P[X\geq (\frac{1}{2}-\frac{\epsilon}{4})s+\frac{\epsilon}{4}s]           \\
				                             & =\mathbb P[X\geq (1+\frac{\epsilon}{2-\epsilon})(\frac{1}{2}-\frac{\epsilon}{4})s] \\
				                             & =\mathbb P[X\geq (1+\delta)\mu]
			\end{align}
			where $\delta = \epsilon/(2-\epsilon)$. Using First Chernoff Bound, we have:
			\begin{align}
				\mathbb P[X\geq \frac{s}{2}] & \leq \exp(-\frac{\delta^2\mu}{2+\delta})        \\
				                             & = \exp(-\frac{\epsilon^2 s}{16-4\epsilon})      \\
				                             & =\exp(-\frac{2^\skheight}{16-4\epsilon} \log n) \\
				                             & = \frac{1}{n^c}
			\end{align}
			where $c=2^\skheight/(\ln2(16-4\epsilon))$. Similarly, there are less than $s/2$ points from $P_R$ in $S$ \whp. The proof follows by the union bound.}
	\end{proof}
}

\begin{lemma}[Tree height]\label{lem:tree-height}
	The total height of a \ourkdtree with size $n$ is $O(\log n)$ for $\os=\Omega(\log n)$, or $\log n+O(1)$ for $\os=\Omega(\log^3 n)$, both \whp.
\end{lemma}
\begin{proof}
	To prove the first part, we will use $\epsilon=1$ in \cref{lem:sampling}.
	In this case, for $\os=6c\log n$, one subtree can have at most 3/4 of the size of the parent with probability $1-1/n^c$, which means that
	the tree size shrinks by a quarter every level. This indicates that the tree heights is $O(\log n)$ \whp for any constant $c>0$.

	We now show the second part of this lemma. For leaf wrap $\leafwrap\ge 4$, the tree has height 1 for $n\leq 4$.
	We will show that using $\epsilon=4/\log n$  (i.e., $\os=O(\log^3 n)$), the tree height $h$ is $\log n+O(1)$.
	Similar to the above, here in the worst case, for a subtree of size~$n'$, the children's subtree size is at most $(1/2+\epsilon/4)\cdot n'=(1/2+1/\log n)\cdot n'$.
	Hence, the tree height satisfies $(1/2+1/\log n)^h=1/n$, so $h={-\log n}/{\log(1/2+1/\log n)}$.
	Here, let $\delta=h-\log n={-\log n}/{\log(1/2+1/\log n)}-\log n$. It solves to $\delta=O(1)$ for $n> 4$.
	Although complicated, the analysis primarily employs some algebraic methods. \ifconference{Due to the space limit, we put it in the supplementary material.}\iffullversion{Due to the space limit, we put it in \cref{app:treeHeightSupport}.}
	The high-level idea is to replace $t=\log n$, so $\delta=f(t)=-t/(\log(1/2-1/t))-t=t/(1+\log t-\log(t-2))-t$.
	We show that $f(t)$ is decreasing for $t\ge 2$ by proving $f'(t)<0$ for $t\ge 2$.
	Since we have several logarithmic functions in the denominator, we computed the second and third derivatives and used a few algebraic techniques to remove them.
\end{proof}
\hide{
	\begin{proof}
		Plugging in $\os\ge\log n$ to \cref{lem:sampling} gives $\epsilon\le 1$ \whp, so one subtree can have at most 3/4 of the size of the parent---the tree size shrinks by a quarter every level.
		This indicates the tree heights to be $O(\log_2 n)$ \whp.

		We now show the second part of this lemma, which is more mathematically interesting. W.l.o.g let leaf wrap $\leafwrap=4$, so that the tree has height 1 for $n\leq 4$.
		By plugging in $\epsilon=O(1)/\log n$, the tree height $h$ can be solved as $\log n+O(1)$.
		Similar to the above, here in the worst case, the children's subtree size shrinks by at least $1/2+\epsilon/4$.
		For simplicity, we let $\epsilon=4/\log n$ so $1/2+\epsilon/4=1/2+1/\log n$.
		Hence, the tree height satisfies $(1/2+1/\log n)^h=1/n$, so $h={-\log n}/{\log(1/2+1/\log n)}$.
		We will show that $\delta=h-\log n={-\log n}/{\log(1/2+1/\log n)}-\log n=O(1)$ for $n> 4$.
		\ifconference{
			This analysis is quite complicated---due to the space limit, we put it in the full version of the paper~\cite{pkdPaperFull}.
		}
		\iffullversion{
			This analysis is quite complicated---due to the space limit, we put it in the \cref{app:treeHeightSupport}.	}
		The high-level idea is to replace $t=\log n$, so $\delta=f(t)=-t/(\log(1/2-1/t))-t=t/(1+\log t-\log(t-2))-t$.
		We show that $f(t)$ is decreasing for $t\ge 2$ by proving $f'(t)<0$ for $t\ge 2$.
		Since we have multiple logarithmic functions in the denominator, we can compute the second and third derivatives and a few algebraic techniques to remove them.
	\end{proof}
}

Later, we will experimentally show that maintaining strong balancing criteria (tree height of $\log_2 n+O(1)$) is not necessary for most \kdtree{}'s use cases.
Hence, in the rest of the analysis, we will use $\os=\Theta(\log n)$ and assume the tree height as $O(\log n)$.

With these lemmas, we now show that \cref{algo:constr} is theoretically efficient in work, span, and cache complexity, if we plug in the appropriate parameters.
Recall that $M$ is the small memory size.
We will set 1) skeleton height $\levels=\epsilon\log M$ for some constant $\epsilon<1/2$;
and 2) chunk size $l=2^\levels$, so array $A$ and $B$ have size $O(2^{\levels}\times |P|/l)=O(|P|)$,
and operations on $A$ and $B$ will have $O(1)$ cost per input point on average.
We use $O(\mb{Sort}(n))=O((n/B)\log_M n)$ to refer to the best-known cache complexity of sorting $n$ keys, which is also a lower bound
for \kdtree{} construction---consider that the input points are in one dimension, then building a \kdtree{} is equivalent to sorting all points by their coordinates.
We also assume $M=\Omega(\polylog(n))$, which is true for realistic settings.

\begin{theorem}[Construction cost]\label{thm:constr}
	With the parameters specified above,
	\cref{algo:constr} constructs a \ourkdtree of size $n$ in optimal $O(n\log n)$ work and $O(\mb{Sort}(n))=O((n/B)\log_M n)$ cache complexity, and has $O({M}^{\epsilon}\log_M n)$ span for constant $0<\epsilon<1/2$, all with high probability.
	Here $M$ is the small-memory size and $B$ is the block size.
\end{theorem}

\begin{proof}
	We start with the work bound.
	Although the entire algorithm has several steps,
	each input point is operated for $O(1)$ times in each recursive level, except for \cref{line:loopup-1,line:loopup-2}.
	For these two lines, looking up the bucket id has $O(\levels)$ work.
	Since the total recursive depth of \cref{algo:constr} is $O(\log n)/\levels$ \whp{},
	the work is $O(\levels\cdot(\log n)/\levels)=O(\log n)$ \whp{} per input point,
	leading to total $O(n\log n)$ work \whp{}.

	\hide{
		We now analyze the span of \cref{algo:constr}.
		The algorithm starts with sampling $2^\levels\cdot\os$ points and build a tree skeleton with $\levels$ levels.
		Using the standard \kdtree construction algorithm (mentioned at the beginning of \cref{sec:constr}) gives $O(\levels \log n)$ span,
		which is dominated by the $O(2^{\levels})=O(M^{\epsilon}\log n)$ span for sampling.

		We now focus on the span of \textsc{Sieve}().
		There are two loop bodies, \cref{line:seq-loop-1,line:seq-loop-2}, that are executed sequentially.
		The span for them is hence the same as the work analyzed above, which is $O(l\levels)$.
		The column-major prefix sum on \cref{line:transpose} can be computed in $O(\log n)$ span~\cite{blelloch2010low}, and all other parts also have $O(\log n)$ span.
		The total span for one level of \textsc{Sieve}() is therefore $O(l\levels+\log n)$.
		Since \cref{algo:constr} have $O(\log n)/\levels$ recursive levels \whp, the span for \cref{algo:constr} is $O(l\levels+\log n)\cdot O(\log n)/\levels=O(2^\levels\log n+\log^2n/\levels)=O(M^{\epsilon}\log n)$ \whp, assuming $M^{\epsilon}>\log n$.
	}

	We now analyze the span of \cref{algo:constr}.
	The algorithm starts with sampling $2^\levels\cdot\os$ points and building a tree skeleton with $\levels$ levels.
	Taking samples and building the skeleton on them can be trivially parallelized in $O(\lambda\log n)$ span (using the plain algorithm at the beginning of \cref{sec:constr}).
	In the sieving step, each chunk has $l=2^{\levels}$ elements that are processed sequentially, and all chunks are processed in parallel.
	This gives $O(2^{\levels})$ span.
	The column-major prefix sum on \cref{line:transpose} can be computed in $O(\log n)$ span~\cite{blelloch2010low}, and all other parts also have $O(\log n)$ span.
	The total span for one level of recursion is therefore $O(l+\log^2 n)=O(M^{\epsilon})$, assuming $M=\Omega(\mathit{polylog(n)})$.
	Since \cref{algo:constr} have $O(\log n)/\levels$ recursive levels \whp{}, the span for \cref{algo:constr} is $O(M^{\epsilon}\log_M n)$ \whp.

	We finally analyze the cache complexity.
	Based on the parameter choosing, the samples fully fit in the cache.
	In each sieving step, since $l=2^\levels=M^{\epsilon}\le \sqrt{M}$, the array $A[i][\cdot]$ and $B[i][\cdot]$ fits in cache, so the loop bodies on \cref{line:seq-loop-1,line:seq-loop-2} will access the input points in serial, incurring $O(n/B)$ block transfers.
	All other parts cost $O(n/B)$ block transfers, including the column-major prefix sum on \cref{line:transpose}~\cite{blelloch2010low}.
	Hence, the total cache complexity for \cref{algo:constr} is $O(n/B)$ per recursive level, multiplied by $O((\log n)/\levels)=O(\log n/\log M)=O(\log_M n)$ levels \whp, which is $O(n/B\cdot \log_M n)$.
\end{proof}

The work and cache bounds in \cref{thm:constr} are the same as sorting (modulo randomization)~\cite{aggarwal1988input} and hence optimal.
The span bound can also be optimized to $O(\log^2 n)$ \whp{}, using a similar approach in~\cite{blelloch2010low}, with the details given in the proof of the theorem below.
\hide{We still use the previous version in the implementation, because in practice, $O({M}^{\epsilon}\log_M n)$ span bound is good enough to achieve sufficient parallelism,
	and the solution with stronger span bound may incur a larger constant that slows down the overall performance. }

\begin{theorem}[Improved span]\label{thm:constr-polylog}
	A \ourkdtree of size $n$ can be built in optimal $O(n\log n)$ work and $O(\mb{Sort}(n))=O((n/B)\log_M n)$ cache complexity, and has $O(\log^2 n)$ span, all with high probability.
\end{theorem}

\begin{proof}
	The $O(2^\levels)=O({M}^{\epsilon})$ span per recursive level is caused by the two sequential loops on \cref{line:seq-loop-1,line:seq-loop-2}.
	These two loops can be parallelized by a sorting-then-merging process as in~\cite{blelloch2010low}.
	The high-level idea is to first sort (instead of just count) the entries in the loop on \cref{line:seq-loop-1} based on the leaf labels.
	Once sorted, we can easily get the count of the points in each leaf.
	Then on \cref{line:seq-loop-2}, once the array is sorted, points can be distributed in parallel.
	We refer the readers to~\cite{blelloch2010low} for more details.
	The span bound for each level is $O(\log n)$~\cite{blelloch2010low,blelloch2020optimal} for the sieving step.
	For the rest of the part, the span is $O(\lambda\log n)$ caused by skeleton construction.
	Altogether, the span per level is $O(\lambda\log n)$, and there are $O(\log n/\lambda)$ recursion levels \whp.
	Therefore, the total span is $O(\log^2 n)$ \whp.
\end{proof}

In practice we still use the previous version because $O({M}^{\epsilon}\log_M n)$ span can enable sufficient parallelism,
and the additional sorting to get the improved span may lead to performance overhead.

\hide{
	\begin{lemma}
		Each bucket node holds $O(\frac{n}{2^\skheight})$ points with w.h.p in each round, where $n$ is the input size.\label{bucketSizeLemma}
	\end{lemma}
	\begin{proof}
		\hide{	Samples $S$ is an $\epsilon$-approximation w.h.p by construction. Given a bucket node with id $i$, we have:
			\begin{align}
				         & \left|\frac{|B_i|}{|P|}-\frac{|B_i\cap S|}{|S|}\right|\leq \epsilon \\
				\implies & |B_i|\leq (\frac{|S|/2^\skheight}{|S|}+\epsilon)|P|                 \\
				\implies & |B_i|\leq (\frac{1}{2^\skheight}+\epsilon)|P|                       \\
				\implies & |B_i|= O\left(\frac{n}{2^\skheight}\right)
			\end{align}
			as desired.}
	\end{proof}

	\begin{theorem}
		\cref{buildtree} runs in $O(n\log n)$ expected work and $O(\log^2 n\log\log n)$ span w.h.p, where $n$ is the input size.\label{theoremBuild}
	\end{theorem}
	\begin{proof}
		By \cref{bucketSizeLemma}, the recursion depth is $\Theta(\log_\skheight n)$, which yields a tree with height $\log_2 n +O(1)$, both w.h.p. In each round, construction of a tree skeleton incurs $O(|S|\log |S|)=O(\log n\log\log n)$ work and span. The following points sieving can be finished in $O(n)$ expected work and $O(\log n)$ span w.h.p using semisort algorithm in \cite{blelloch2020optimal}. In total $O(n)$ work and $O(\log n\log\log n)$ span in each round and the bounds follow then.
	\end{proof}}

\section{Parallel Algorithms for Batch Updates \label{sec:update}}

In this section, we present our parallel update algorithms for \ourtree{s}.
Here we consider the batch-parallel setting that inserts or deletes a batch of points $P$ to the current \ourtree{} $T$.
\ourtree{s} do not require the tree to be perfectly balanced as in existing parallel implementations~\cite{cgal51,jo2017progressive,shevtsov2007highly}.
Our key idea 
is to make the tree \emph{weight-balanced}
and to \emph{partially reconstruct} the tree upon imbalance.

\cref{fig:batchInsertDiagram} illustrates the high-level idea.
We allow the sizes of the two subtrees to be off by at most a factor of $\balpara$,
i.e., the size of a subtree can range from $(0.5-\balpara)$ to $(0.5+\balpara)$ times the size of its parent.
Such a relaxation allows most updates to be performed lazily, until at least a significant fraction of a subtree has been modified.
If such a case happens, we rebuild the unbalanced subtree using \cref{algo:constr}.
The rebuild cost can be amortized to the updated points in all batches.
This idea of lazy updates with reconstruction has been studied sequentially on various trees for point updates~\cite{galperin1993scapegoat,overmars1983design}.
The key challenge here is to adapt this idea to the batch-parallel setting while maintaining theoretical and practical efficiency.
Theoretically, we show efficient work and cache bounds, and high parallelism for our new batch-update algorithm.
In practice, we conduct an in-depth performance study with the relaxation of balancing criteria.
In \cref{sec:balancingParameterRevisted}, we show that the query performance of \ourtree{} remains fairly stable for $\balpara\le 0.4$ (the two subtree sizes differ by up to 9$\times$).
The weight-balanced feature of \ourtree{} allows it to significantly outperform all existing counterparts.


\subsection{Batch Insertion\label{sec:batchInsert}}

\begin{figure}[t]
	\includegraphics[width=0.45\textwidth]{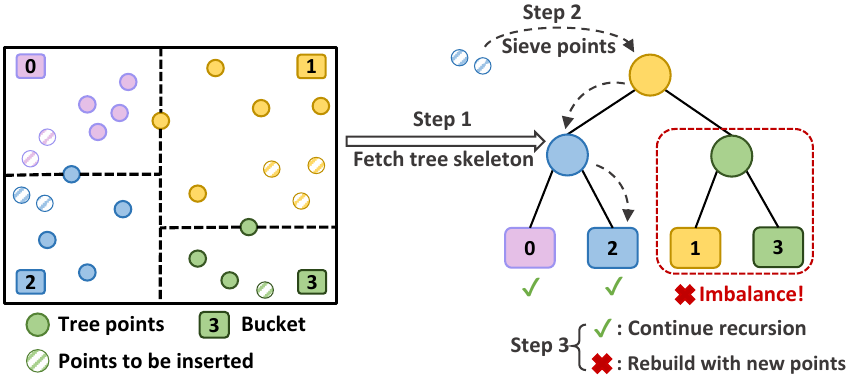}
	\caption{\textbf{Illustation of our batch insertion to a \kdtree{}}.
		Our algorithm first fetches the tree skeleton from the \kdtree{}, sieves the points into the corresponding bucket as in \cref{algo:constr}, then processes each buckets in parallel, and finally rebuilds the subtrees that become imbalance after insertion.
	}
	\label{fig:batchInsertDiagram}
\end{figure}

\hide{
	\begin{algorithm}[t]
		\fontsize{8pt}{8.5pt}\selectfont
		\caption{Batch insertion\label{batchinsert}}
		\SetKwFor{parForEach}{parallel-foreach}{do}{endfor}
		\SetKwProg{MyFunc}{Function}{}{end}
		\SetKwInOut{Note}{Note}
		\SetKwInOut{Parameter}{Parameter}
		\KwIn{Points $P$ with dimension $k$, a $k$d-tree $T$.}
		\KwOut{A $k$d-tree with $P$ inserted.}
		\Parameter{$\skheight$: the maximum height of a fetched tree skeleton.\\
			$\leafwrap$: the leave wrap of $T$.
		}
		\DontPrintSemicolon

		\vspace*{.5em}
		\tcp{Insert points $P$ into $k$d-tree $T$}
		\MyFunc{\upshape{\batchinsert{$(T,P)$}}}{
			$n\gets |P|$\tcp*[f]{Size of the input points}\\
			\If{$n=0$\label{emptyP}}{\Return $T$}
			\If{$T$ is leaf}{
				$S\gets T.p\cup P$\\
				\eIf(\tcp*[f]{There is enough space in leaf}){$|S|\leq \leafwrap$}{
					$T.p\gets S$\label{checkleaf}\tcp*[f]{Put nodes within the leaf directly}\\
					\Return $T$
				}{
					$\deleteTree(T)$\label{deleteatleaf}\tcp*[f]{Delete the tree before rebuild}\\
					\Return $\build(S)$\label{rebuildleaf}\tcp*[f]{Rebuild the tree}
				}
			}
			$\matht\gets\fetchTreeSkeleton(T,0)$\label{fetchinsert}\\
			Sieve points to the corresponding bucket node in $\matht$, array slice $B[i]$ holds all points sieved to bucket node $i$. \label{partitioninsert}\\
			Tag the roots of (sub)trees in $\matht$ that become imbalance after insertion. If both a node and its ancestor are imbalance, tag the ancestor only.\label{tagimba} \\
			Collect all tagged imbalance nodes, and all bucket nodes that are not in an imbalance (sub)tree, within set $U$.\\
			\parForEach{tree node $\varmatht\in U$\label{parlaykeyvalue}}{
				$S\gets$ points sieved into $\varmatht$\\
				\eIf(\tcp*[f]{Stop insertion}){$\varmatht$ is taged to be imbalance}{
					$T'\gets \build(\{\varmatht\}\cup S)$\tcp*[f]{Rebuild the tree instead}\label{needsrebuild}\\
					$\deleteTree(\varmatht)$
				}(\tcp*[f]{No need to rebuild, continue insertion})
				{$T'\gets\batchinsert(\varmatht,S)$\label{needscontinue}}
				Replace the node $\varmatht$ in $\matht$ with $T'$ and redirect its parent's pointer to $T'$\label{replaceinsert}
			}
			\Return The root of $\matht$\label{endinsert}
		}

		\vspace{.5em}
		\tcp{Fetch the tree skeleton from $T$ with height no more than $\skheight$}
		\MyFunc{\upshape{\fetchTreeSkeleton}$(T,h)$\label{retriveskeleton}}{
			\If(\tcp*[f]{Reaches last layer of tree skeleton}){$T$ is leaf \OR $h\geq \skheight$\label{retrieveend}}{
				Cast $T$ to a bucket node and assign it a new identifier\label{casting}\\
				\Return The pointer to the bucket node
			}
			$\fetchTreeSkeleton(T_\lson,h+1)$\\
			$\fetchTreeSkeleton(T_\rson,h+1)$\\
			\Return $T$ \label{retriveskeletonend}\tcp*[f]{Use the tree node itself as a skeleton }
		}

	\end{algorithm}
}

\begin{algorithm}[t]
	\fontsize{8pt}{8.5pt}\selectfont
	\caption{Batch insertion\label{algo:insert}}
	\SetKwFor{parForEach}{parallel-foreach}{do}{endfor}
	\SetKwProg{MyFunc}{Function}{}{end}
	\SetKwInOut{Note}{Note}
	\SetKwInOut{Parameter}{Parameter}
	\SetKwFor{pardo}{In Parallel:}{}{}
	\KwIn{A sequence of points $P$ and a \kdtree{} $T$.}
	\KwOut{A $k$d-tree with $P$ inserted.}
	\Parameter{$\skheight$: the maximum height of a fetched tree skeleton.\\
		$\leafwrap$: the leave wrap of $T$.
	}
	\DontPrintSemicolon

	\vspace*{.5em}
	\tcp{Insert points $P$ into $k$d-tree $T$}
	\MyFunc{\upshape{\batchinsert{$(T,P)$}}}{
	\lIf{$P=\emptyset$\label{emptyP}}{\Return $T$}
	\lIf(\tcp*[f]{Insert into a leaf}){$T$ is leaf\label{line:base2}}{
		\Return $\build( T \cup P)$\label{rebuildleaf}
	}
	$\matht\gets$ the skeleton at $T$\label{fetchinsert}\\
	Sieve points to the corresponding bucket in $\matht$ using \textsc{Sieve}$(P,\matht)$ from \cref{algo:constr}. Let $R[i]$ be the sequence of all points sieved to bucket $i$. \label{partitioninsert}\\
	$t\gets$ The root of skeleton $\matht$\\
	\Return \textsc{InsertToSkeleton}$(t,R[0..2^{\skheight}))$
	}
	\vspace{.5em}
	\tcp{Insert the buckets $R[l]$ to $R[r-1]$ to a node $t$ in the skeleton}
	\MyFunc{\upshape\textsc{InsertToSkeleton}$(t,R[l..r))$\label{retriveskeleton}}{
	\If{$t$ is an external node in the skeleton\label{line:rebalancebase}}{
		$x\gets$ the subtree at $t$\\
		\Return \batchinsert{$(x, R[l])$}
	} \Else{
		\If{after insertion, the two subtrees at $t$ are weight-balanced\label{insertrebuild}} {
			$m\gets$ number of buckets in $t.\lc$\\
			\pardo{}{
				$t.\lc \gets \textsc{InsertToSkeleton}(t.\lc,R[l,m))$\\
				$t.\rc \gets \textsc{InsertToSkeleton}(t.\rc,R[m,r))$\\
			}
			\Return $t$
		} \lElse(\tcp*[f]{Rebuild the subtree}){
			\Return \build{$\left(t\cup \left(\bigcup_{i=l}^{r-1}R[i]\right)\right)$}\label{line:update-build-tree}
		}
	}
	}

\end{algorithm}

We show our insertion algorithm in \cref{algo:insert}, which takes as input a \kdtree{} $T$ and a set of points $P$,
and inserts $P$ to $T$.
There are two base cases: 1) if $P=\emptyset$ no insertion is needed (\cref{emptyP}), and 2)
if $T$ is a leaf, the algorithm will construct a tree based on $P\cup T$ (\cref{line:base2}).

Otherwise, we will first grab the skeleton $\matht$ at $T$ at \cref{fetchinsert}.
Here we will also apply the \emph{sieving algorithm} in construction to sieve all points in $P$ based on the skeleton $\matht$ (\cref{partitioninsert}).
Based on the partition of the buckets $R[]$, we will apply the insertions and rebalance the tree in function \textsc{InsertToSkeleton}.
This algorithm not only processes the skeleton top down to perform the insertions of each bucket to the corresponding subtrees, but also identifies the unbalanced subtrees to reconstruct them.
In particular, with the set of points in each bucket and the original subtree size, we can compute the sizes of each subtree in $\matht$ after insertion, and thus determine whether any of these subtrees are unbalanced.
If we encounter the node $t$ in $\matht$ that will become unbalanced after insertion (the else-condition at \cref{line:update-build-tree}),
we will directly reconstruct the subtree using all original points in $t$ and the points in $P$ that belong to this subspace.
A reconstruction can be performed by flattening all points in the current subtree with the points to be inserted, and applying the construction algorithm to create a (almost) perfectly balanced tree.
Note that in our case, it is not perfectly balanced due to our sampling-based construction algorithm, but in \cref{sec:update-analysis} we will show our update algorithms are still theoretically efficient.
If a reconstruction is triggered at subtree $t$, we do not need to further process the subtrees of $t$ in this case.

\hide{
	We process the skeleton top down. If an external node in the skeleton is reached, then we just need to recursively insert corresponding bucket in $R$ (i.e., the set of points to be inserted in this subtree) into this subtree (\cref{line:rebalancebase}).
	Otherwise, with the set of points in each bucket and the original subtree size, we can compute the sizes of each subtree in $\matht$ after insertion,
	and thus determine whether any of these subtrees are unbalanced.
	If we encounter node $t$ in $\matht$ that will become unbalanced after insertion,
	we will directly reconstruct the subtree using all original points in $t$ and the points in $P$ that belong to this subspace.
	A reconstruction can be performed by flattening all points in the current subtree with the points to be inserted, and applying the construction algorithm to create a (almost) perfectly balanced tree.
	Note that in our case it cannot be perfectly balanced due to our sampling-based construction algorithm, but in \cref{sec:update-analysis} we will show our update algorithms are still theoretically efficient.
	If a reconstruction is triggered at subtree $t$, we do not need to further process the subtrees of $t$ in this case.
	If a subtree remains balanced after insertion, we recursively deal with the subtrees and their corresponding buckets (). }


\subsection{Batch Deletion}
Given a set of points $P$ and a \kdtree $T$, the batch deletion algorithm removes $P$ from $T$.
Compared with batch insertion, the challenge of batch deletion is the additional step of handling points that are not in the tree, i.e., $P\setminus T \neq \emptyset$. Due to these absent points, we can no longer identify the unbalanced subtrees before we traverse into these subtrees and mark all the points to be deleted.

Therefore, our algorithm for batch deletion goes in two rounds.
In the first round, all points in $P$ are sieved into the corresponding leaves in $T$ using the sieving algorithm from \cref{algo:constr}.
By doing this, we identify all points in $P$ that are not in $T$, and discard them.
After this, we know the exact size of each subtree after deletion, and we then identify the unbalanced ones after deletion.
This process is similar to the batch insertion algorithm, and thus the asymptotic cost also remains the same.

\hide{
	\myparagraph{Dealing with partially overlapped $P$ and $T$.}
	In practice, $P$ and $T$ may partially overlap, i.e., inserting points already in $T$ and removing points not in $T$.
	In this case, directly using $|T|+|P|$ or $|T|-|P|$ to estimate the tree size after updating cannot accurately determine the imbalance of the tree.
	There are several ways to get around this.
	The simplest one is to maintain a set for all points in $T$, and always first check if $p\in P$ is already in $T$ (for insertion) or not in $T$ (for deletion) before applying the updates.
	We can use parallel hash tables~\cite{shun2014phase} or BSTs~\cite{blelloch2016just,sun2018pam} to maintain the set without affecting the asymptotic bounds.
	Our implementation uses the second approach by traversing the tree in two rounds.
	In the first round, we record the actual size for each subtree after insertion/deleltion. Then in the second round we perform the actual update, and rebalance based on the new sizes computed.
}

\hide{
	\cref{batchinsert} inserts a set of points $P$ into a \kdtree{} $T$. It first skips the trivial case that $P=\emptyset$ at \cref{emptyP}. Then if $T$ is a leaf, the algorithm checks whether there are enough space to put $P$ within $T$ (\cref{checkleaf}), otherwise, after the deletion of $T$ (\cref{deleteatleaf}), it rebuilds $T$ using stored points $T.p$ and the batch $P$ (\cref{rebuildleaf}). The points partition starts when $T$ is an interior node. We first retrieve the tree skeleton $\matht$ with height at most $\skheight$ from $T$ recursively (\cref{retriveskeleton}-\cref{retriveskeletonend}). When it encounters a leaf node or the skeleton height reaches to $\skheight$ (\cref{retrieveend}), the retrieve ends with assigning the node an identifier and casting it to a bucket node (\cref{casting}). Once the tree skeleton $\matht$ has been retrieved (\cref{fetchinsert}), we can sieve $P$ into the corresponding bucket (\cref{partitioninsert}) using same technique as \cref{buildtree}. Some subtrees in $\matht$ would become imbalance after insertion, to avoid rebuilding a tree (and its subtrees) multiple times, we only tag those imbalance nodes with the smallest depth in $\matht$ (~\cref{tagimba}).
	Then for each tagged imbalance nodes and all bucket nodes that are not in an imbalance (sub)-tree (\cref{parlaykeyvalue}), we launch a threads in parallel to either rebuild the tree together with the collected points $S$ if $\varmatht$ is an imbalance node (\cref{needsrebuild}), or continue recursion with sieved points $S$ in that bucket node $\varmatht$ (\cref{needscontinue}). After which we replace the node $\varmatht$ with newly returned node $T_i$ (\cref{replaceinsert}). The recursion terminates by returning the root of the skeleton $\matht$ (\cref{endinsert}).
}


\hide{
	\subsection{Batch Deletion}
	\setlength{\algomargin}{0em}
\begin{algorithm}[t]
	\fontsize{8pt}{8.5pt}\selectfont
	\caption{Batch deletion\label{batchdelete}}
	\SetKwFor{parForEach}{parallel-foreach}{do}{endfor}
	\KwIn{Points $P$ with dimension $k$, a $k$d-tree $T$.}
	\KwOut{A $k$d-tree with $P$ deleted.}
	\SetKw{MIN}{min}
	\SetKw{MAX}{max}
	\SetKwProg{MyFunc}{Function}{}{end}
	\SetKwInOut{Note}{Note}
	\SetKwFor{inParallel}{in parallel:}{}{}
	\SetKwInOut{Parameter}{Parameter}
	\Parameter{$\skheight$: the maximum height of a fetched tree skeleton.\\
		$\leafwrap$: the leave wrap of $T$.
	}
	\DontPrintSemicolon
	\vspace{.5em}
	\tcp{Delete points $P$ from a $k$d-tree $T$. Boolean flag $t$ takes value $1$ if none of $T$'s ancestor needs rebuild after deletion and 0 otherwise. }
	\MyFunc{\upshape{\batchdelete$(P,T,t)$}}{
		$n\gets |P|$\tcp*[f]{Size of the input batch}\\
		\If{$n=0$}{\Return $T$}
		\If{$T$ is leaf}{
			$T.p\gets T.p\setminus P$\tcp*[f]{Remove the points from the leaf}\label{removeFromLeave}\\
			\Return $T$	
		}
		$\matht \gets \fetchTreeSkeleton(T,0)$\label{fetchTreeSkeleton}\\
		Sieve points to the corresponding bucket node in $\matht$, array slice $B[i]$ holds all points sieved to bucket node $i$.\label{sieveDelete}\\
		Put a tomb to the (sub)tree in $\matht$ if it becomes imbalance after deletion. If both a node and its ancestor are imbalance, put the tomb to its ancestor only\label{puttomb}\\
		\tcp{Delete the points from subtrees of $T$ first}
		\parForEach{bucket id $i$}{
			$\matht_i\gets$ bucket node with identifier $i$\\
			\eIf(\tcp*[f]{Any subtree containg $\matht_i$ is balance})
			{$\matht_i$ has tomb}
			{$T_i=\batchdelete(B[i],\matht_i,1)$\label{withTomb}}
			(\tcp*[f]{Some subtrees containg $\matht_i$ will be rebuilt})
			{$T_i=\batchdelete(B[i],\matht_i,0)$\label{noTomb}}
			Replace the node $\matht_i$ in $\matht$ with $T_i$ and redirect its parent's pointer to $T_i$ 
		}
		\tcp{Rebuild imbalance nodes then}
		$U\gets$ Set of nodes in $\matht$ that hold a tomb\\
		\parForEach{$\varmatht\in U$\label{beginRebuildDelete}}{
			 $\varmatht'\gets\build(\{q\})$\label{rebuildINDelete}\\
			 Replace node $\varmatht$ in $\matht$ with $\varmatht'$\\
			$\deleteTree(\varmatht)$\label{endRebuildDelete}
		}
		\Return The root of $\matht$ $\label{returnSkDeletion}$
	}
	\end{algorithm}

	As the converse version of batch insertion, the batch deletion \cref{batchdelete} deletes $k$-dimensional points $P$ from a $k$d-tree $T$ in a divide-and-conquer fashion. Note that unlike the batch insertion that rebuilds the (sub-)tree once identifying those imbalance nodes, one can only rebuild the (sub-)tree after removing all necessary points from its leaves. To avoid rebuild $T$ (and its subtrees) multiple times, we use an boolean flag $t$ that takes value $1$ if no (sub-)tree containing $T$ has been rebuilt and 0 otherwise. This is equivalent to say that $t$ equals to 1 if and only if no $T$'s ancestor has been rebuilt.

	The algorithm first removes $P$ from where they are stored if $T$ is a leaf (\cref{removeFromLeave}), otherwise, the partition starts by fetching the tree skeleton $\matht$ (\cref{fetchTreeSkeleton}) and then sieves $P$ to the corresponding bucket node in $\matht$ (\cref{sieveDelete}).
	Every node in $\matht$ receives a tomb if no ancestor of $\matht$ would be rebuilt
	meanwhile itself becomes imbalance after deletion (\cref{puttomb}).
	Now for every bucket node, one can continue deletion using sieved points $B[i]$ together with their imitated stamp $t$ (\cref{withTomb} and \cref{noTomb}). After the recursion ends, it rebuilds all imbalance (sub-)tree in parallel (\cref{beginRebuildDelete}-\cref{endRebuildDelete}) and returns the updated skeleton root $\matht$ (\cref{returnSkDeletion}).
}

\subsection{Theoretical Analysis}\label{sec:update-analysis}

We now show that our conceptually simple batch update algorithms also have good theoretical guarantees.
Since the update algorithms use the construction algorithm as a subroutine, we need to accordingly set up the parameters for both algorithms.
In particular, we select $\os=(6c\log n)/\balpara^2$ for some constant $c>0$ to ensure a low amortized cost in \cref{thm:update}.
Here we assume $\balpara$ is a constant and $\os=\Theta(\log n)$.

\begin{theorem}[Updates]\label{thm:update}
	Using $\os=(6c\log n)/\balpara^2$, the update (insertions or deletions) of a batch of size $m=O(n)$ on a \ourkdtree of size $n$ has $O(\log^2 n)$ span \whp;
	the amortized work and cache complexity per element in the batch is $O(\log^2 n)$ and $O(\log(n/m)+(\log n\log_M n)/B)$ \whp, respectively.
\end{theorem}

For simplicity, \cref{thm:update} assumes the batch size $m=O(n)$.
If $m=\omega(n)$, we just need to replace the term $n$ by $m+n$ in the bounds for batch insertions (no change needed for batch deletion).

Due to the space limit, we defer the full proof in \ifconference{the supplementary material}\iffullversion{\cref{app:batchUpdate}}.
The overall structure of this analysis is similar to \cref{thm:constr}, with the additional information that traversing $m$ leaves in a binary tree of size $n$ touches $O(m\log (n/m))$ tree nodes~\cite{blelloch2016just}.
\hide{
	\begin{proof}
		The span bound is asserted since the points sieving and the tree rebuilding all have $O(\log n\log_Mn)$ span \whp (\cref{thm:constr-polylog}), and other parts have $\levels = O(\log n)$ span. Meanwhile, note that the tree rebuilding can only be triggered once on any path, thus the bound follows.
		For work bound, once we rebuild a (sub)-tree with size $n'$ using $O(n'\log n')$ work, the new tree contains $(1/2\pm \balpara/4)n'$ points \whp (\cref{lem:sampling}), which needs to insert another $3\balpara n'4$ points to make it imbalanced again and trigger a new rebuilding. The amortized cost per point in this (sub)-tree is thus $O(\log n')$. Since the tree has height $O(\log n)$ \whp, in total $O(\log^2 n)$ amortized work needed per point.
		As for I/O bound, similarly, the amortized cost for each point in rebuilding is $O((1/B)\log_M n')$ \whp (\cref{thm:constr}), overall $(\log n\cdot (1/B)\log_M n')$ \whp per point, as the tree has height $O(\log n)$ \whp. Besides, it incurs $O(\log(n/m))$ amortized I/O per point to traverse the tree to find the leaf, and the I/O bound follows then.
		\ifconference{
			We refer the reviewers to our full paper\cite{xx} for more detail.
		}
		\iffullversion{
			We refer the readers to \cref{app:batchUpdate} for more details.
		}
		\hide{
			We will start with the span bound.
			According to \cref{thm:constr-polylog}, the sieve process and trees rebuilding (rebalancing) all have $O(\log n\log_M n)$ span \whp.
			Note that the tree rebuilding (\cref{line:update-build-tree}) can only be triggered once on any tree path.
			The span for other parts is $O(\log n)$---the \textsc{InsertToSkeleton} function can be recursively call for $\levels = O(\log n)$ levels each with constant cost.
			In total, the span is the same as the construction algorithm, since in the extreme case, the entire tree can be rebuilt.

			Then we show the work bound.
			The cost to traverse the \ourtree and find the corresponding leaves to update is $O(\log n)$ per point \ziyang{\whp?}, proportional to the tree height.
			Once a rebuild is triggered (on \cref{line:update-build-tree}), the cost is $O(n'\log n')$ where $n'$ is the subtree size.
			After that (or the original construction), each subtree will contains $(1/2\pm \sqrt{(12c\log n)/\os}/4)n'=(1/2\pm \balpara/4)n'$ points \whp (\cref{lem:sampling}).
			We need to insert at least another $3\balpara n'/4$ points for this subtree to be sufficiently imbalance that triggers the next rebuilding of this subtree.
			The amortized cost per point in this subtree is hence $O(\log n'/\balpara) = O(\log n')$ on this tree node assuming $\balpara$ is a constant.
			Note that \ourtree has the tree height of $O(\log n)$, so overall amortized work per inserted/deleted point is $O(\log^2 n)$.

			We can analyze the I/O bound similarly.
			We first show the rebuilding cost.
			For a subtree of size $n'$, the cost is $O((n'/B)\log_M n')$ (\cref{thm:constr}).
			The amortized cost per updated point, using the same analysis above, is $O((1/B)\log_M n')$.
			Again since the tree height is $O(\log n)$, the overall amortized work per inserted/deleted point is $O((\log n\log_M n)/B)$.
			Then, we consider the cost to traverse the tree and find the corresponding leaves to update.
			Finding $m$ leaves in a tree of size $n$ will touch $O(m\log(n/m))$ tree nodes~\cite{blelloch2016just}, so the amortized I/O per point is $O(\log(n/m))$.
			Putting both cost together gives the stated I/O bound.
		}
	\end{proof}
}
Again in practice, we use the sieving approach in \cref{algo:constr}, which leads to $O(M^{\epsilon}\log n)$ span and supports sufficient parallelism.

The update cost bound for \ourtree is higher than using the logarithmic method---e.g., the work per point is $O(\log^2 n)$ instead of $O(\log n)$.
However, we note that the bound for \ourtree is not tight.
Unless in the adversarial case, the update cost per point is more likely to be $O(\log n)$ when subtree rebuild is less frequent.
In \cref{sec:exp}, we will experimentally show that the update is faster than the logarithmic method practically.
Meanwhile, since \ourtree only keeps a single tree rather than $O(\log n)$ trees, the query performance on \ourtree is significantly better.

In addition, \ourtree can support stronger balancing criterion for by setting $\balpara=o(1)$.
In this case, the amortized work and cache complexity per point will increase to $O((\log^2 n)/\balpara)$ and $O(\log(n/m)+(\log n\log_M n)/B\balpara)$ \whp, respectively.
For example, we can enable $\log n+O(1)$ tree height by using $\balpara=O(1)/\log n$.
However as mentioned, our experimental results show that using tree height as $O(\log n)$ is good enough to give overall good performance for both updates and queries in practice.

\hide{
	\begin{lemma}
		\ziyang{in prelim?}\\
		Given a perfect balanced (sub-)\kdtree rooted at $T$, it needs at least $\Omega(\log n)$ times update to make $T$ imbalance, where $n$ is the tree size when it loss balance.
	\end{lemma}

	\begin{theorem}
		Given a \kdtree $T$ that was created using batch insertions/deletions, then every batch update $B$ to $T$ costs $O(m\log^2(n+m))$ amortized work and $O()$ span, where $m$ is the batch size and $n$ is the tree size before update.
	\end{theorem}
	\begin{proof}
		If $T=\emptyset$, then batch insertion equivalents to construct a new tree and deletion returns in $O(1)$ time. So in the following context we consider only the case that $T$ is not empty.\\
		We first compute the work. The cost to construct a tree skeleton is a constant since we assume its height $\skheight$ is a constant. It needs $O(\log_\skheight n)=O(\log_2 n)$ rounds to sieve all points to the corresponding leaf nodes, using $O(m)$ expected work in each round. Therefore, in total $O(m\log n)$ work required to finish the points sieving. Let $\varmatht$ denote an imbalance (sub-)tree in $T$ after updates and assume initially $|\varmatht|=n'$ and in total $m'\leq m$ nodes sieved into this (sub-)tree. By \cref{theoremBuild}, it needs $O((n'+m')\log (n'+m'))$ work to rebuild $\varmatht$. These costs we charge to the $m'$ updates within $\varmatht$, each update can only be charged on the skeleton nodes towards the leaf node and at most once for every such node. In this case, each update can be charged at most $O(\log n')$ times, each time charges $O(\log n')$ work. In this case, the amortized work for each batch update is $O(m\log^2 (m+n))$.
	\end{proof}
}

\section{Implementation Details}
\label{sec:impl}



\myparagraph{Avoid the Extra Copies.}
For simplicity, in \cref{algo:constr}, we assume copying the array of points in $P'$ back to $P$ (\cref{line:copyback}) after distributing the points.
In practice, this copy can be avoided by swapping $P$ and $P'$ in each recursive call.
This significantly saves unnecessary memory accesses in the algorithm.

\myparagraph{Parameter Choosing.}
Our theoretical analysis in \cref{sec:constr-analysis} suggests $\levels=\epsilon\log M$ for some constant $\epsilon<1/2$.
In practice, we observed that using $\skheight=4$ to 10 generally gives good performance. We use $\skheight=6$ for \ourtree{s} in our experiments.
We use $\leafwrap=32$ for the leaf warp size, and over sampling rate $\os=32$.
We set the balancing parameter $\alpha=0.3$, and further explain our choice in \cref{sec:balancingParameterRevisted}.

\myparagraph{Reduce the Memory Usage.} A key effort in implementing \ourtree{} is to minimizing the memory usage.
Reducing the memory footprint is crucial in at least two aspects.
First, it allows the \ourtree to handle larger inputs.
Second, a smaller memory footprint generally means better cache utilization, leading to better performance.

There are a few approaches in the design of the \ourtree to reduce memory usage.
The first is the leaf wrapping as mentioned, which creates a flat tree leaf when the subtree size drops below a certain threshold (32 in \ourtree{s}).
We also contract the leaves when all points are duplicates, and we refer the audience to \iffullversion{\cref{app:heavyleaf}}\ifconference{the supplementary material} for details.
Second, we try to keep each interior node as small as possible to fit more tree nodes in the cache.
The only additional information we keep for each tree node is the subtree size, which is needed in our weight-balance scheme and is used in range count queries.
Namely, unlike ParGeo and \cgal, the \ourtree does not store the \emph{bounding box} of each tree node,
which is the smallest box containing all points in this subtree.
This box can be used in queries to prune the subtree: when the query does not overlap with the box, the entire subtree can be skipped.
In \ourtree{}, instead of storing the bounding box,
the query will compute the subspace of each subtree on-the-fly:
the function will pass the subspace of the current tree node to recursive calls at its children,
so the subspaces for the children can be further computed with the splitter.
This is not as tight as the bounding box,
but
in our experiments, we observed that avoiding explicitly storing bounding boxes gives better overall performance for \ourtree{}.

\myparagraph{Queries.} Since the \ourtree{} is a single \kdtree{}, we can use all standard \kdtree{} query algorithms on \ourtree{s}.

In our \knn{} query, we use the standard depth-first search algorithm.
When searching the query point $q$ in a non-leaf subtree $T$, if $q$ is to the left of the splitter of $T$, it will visit the $T.lc$ first, and vice versa.
After the recursion returns, we prune the visit to the other child of $T$ by the distance between $q$ and the splitter in $T$.
If $T$ is a leaf, we traverse all points stored in $T$ and add them to the candidate container.
The range query is to traverse the tree recursively, checking if the subspaces associated with nodes fall within the query box and pruning branches that do not intersect the query region.
The only strong query bounds we know of for the standard \kdtree{} are for orthogonal range queries and range counts.
A range count on $D$ dimensions takes $O(2^{h(D-1)/D})$ work on a \kdtree of height $h$~\cite{agarwal2003cache,blelloch2018geometry}, which is $O(n^{(D-1)/D})$ if the tree height is $\log n+O(1)$.
The bound for a range query has an additive term $k$ where $k$ is the output size.
We can set the parameters of \ourtree{} accordingly as in \cref{lem:tree-height} to achieve this bound in theory, although later in \cref{sec:balancingParameterRevisted} we show that the query performance does not degenerate by a slightly larger tree height.
While no strong bounds are known for \knn{} queries on \kdtree{} on general distributions,
previous work has shown that \kdtree{} is highly practical for such queries, and it is the main use case for \kdtree{s} in the real world.

\hide{
	\myparagraph{Maintaining Candidates in $k$-NN Queries.}
	During \knn{} search, updating the candidates is a write-extensive process that takes a large portion of the search time.
	We use a standard approach as in CGAL, which uses a priority queue based on a max-heap of size $k$.
	When the queue is full, elements are inserted only if it is smaller than the root of the heap.
}

\myparagraph{Parallel Granularity Control.}
As standard parallel granularity control,
for tree construction and batch update, when the input size is smaller than 1024, we will continue the process using the standard sequential algorithm.

\hide{
	During $k$-NN search, updating the candidates is a write-extensive process that takes large portion of time within the search. \zdtree uses a swap-vector for small $k$ and the priority queue in \texttt{C++~STL} for large ones. ParGeo maintains an array with size $2\cdot k$ for candidates and deletes those that are larger than the median when the size reaches the capacity. A better alternative for above methods in general cases is the one used in \cgal, which is called \bdita{bounded (max-)heap}. It implements a priority queue using an array with size $k$. When the queue is full, elements are inserted if and only if it is smaller than the root of the heap. In our implementation, the bounded heap is slighly faster than the priority queue for $k=100$ and about 3$\times$ faster than the double-array used in ParGeo.
	For $k=10$, it keeps competitive with the swap-vector in \zdtree{}.
}

\hide{
	\begin{itemize}
		\item \zdtree uses a swap-vector for small $k$ and the priority queue in \texttt{C++ STL} for large $k$. It is reported in \cite{blelloch2022parallel} that the overhead of swap-vector for $k<40$ was significantly less than the priority queue.
		\item \cgal uses the \bdita{bounded (max)-heap} which simulates the priority queue in \texttt{STL} using a fixed size of array. When the size of heap is not full, inserting element is identical as insertion in the priority queue; otherwise, the element is inserted if and only if it is smaller than the root of the heap.
		\item ParGeo maintains a double-array with size $2 k$. Similarly to the bounded heap, when the array is not full, insertion is performed by appending the element to the end of the array. Once the elements reaches the capacity of array, it picks the median of all candidates, partition points based on the median and discards the elements that are larger than the median.
	\end{itemize}
	We compare above methods using our implementation. The experimental results show that the bounded heap provides the best overall performance. For $k=100$, it is slightly faster than the priority queue and about 3x faster than the double-array used in ParGeo. For $k=10$, it keeps competitive with the swap-vector in \zdtree{}. In this case, we use the bounded heap in our implementation.

	\myparagraph{Parallel flatten in range query}
	During range query, when the query rectangle fully contains the bounding box of a tree node (in our implementation, it is the one passed during recursion), we can flatten the points into the output array \textbf{in parallel} and therefore skip the entire tree. This approach saves the time for range query significantly due to the better parallelism.
	\cgal detect the inclusion as well but perform the flatten in serial. For each tree node, ParGeo stores all points in this (sub)-tree within a separate array and the flatten is simply to copy the array to the output.
	\zdtree does not implement the range query.
}

\hide{
	\section{Implementations}

	We have implemented a highly optimized parallel \kdtree{}, called \bdita{\ourlib}, that supports tree construction, batch updates, $k$-NN, range count, and range query, for arbitrary dimensions and all in parallel. Since there are various approaches to construct a \kdtree and perform quires, we shall discuss in this section the major techniques applied in our implementation that attribute to our high performance.


	\subsection{Reduce I/O Cost and Memory Usage\label{sec:reduceIO}}
	The practical performance of parallel algorithm heavily relies on the I/O cost and memory usage. Fewer I/O operations lead to less memory access, which is the bottleneck in the implementation of many algorithms. Our implementation follows the ideas in \cref{sec:constr} and \cref{sec:update} to achieve I/O efficiency for both tree construction and batch update, i.e., we use tree skeleton to build the tree with $\skheight$ levels at once and eliminate the repetitive data movement within each level, which can also be utilized in point sieving to reorder the input and relocate all points in the same bucket to be contiguous. We refer readers to corresponding sections for more detail.

	Small memory footprint enables higher cache utilization and better coarsening, resulting faster program execution time. To reduce the memory usage in our \ourlib, we let each leaf node wrap 32 points in maximum. Each stored point contains only its dimension and excludes other unnecessary attribute such as ID. Additionally, we also remove the bounding box from the tree but instead passing it as an argument during the query, which we shall discuss more in \cref{sec:boundingBox}.

	\subsection{Tree Construction}
	\myparagraph{Avoid the extra copy}
	The implementation for tree construction mostly follows the routine illustrated in \cref{algo:constr}, with the exception that after we finish the point distribution (\cref{line:beginDistribution}-\cref{line:endDistribution}), instead of copying the local array $P'$ back to the input $P$ (\cref{line:copyback}), we swap the pointer to $P'$ and $P$ directly and use $P'$ as the input for next recursion. Further optimization is to create and pass $P'$ as a twin array for $P$ from beginning of the algorithm.
	This requires extra $\Theta(n)$ auxiliary space for $P'$ but save the write cost in copy and memory allocation in each recursion, which brings about 10\% improvement for total build time.

	\myparagraph{Tree skeleton}
	Recall that we use sample to build the tree skeleton during the tree construction (\cref{algo:constr}, \cref{sampleBuild}) and the sample size is set to be $2^\skheight\cdot \os$. There is trade-off between the sample size (equivalent to the skeleton height) and tree quality: a larger sample leads to a more balanced tree, while the cost for sampling and building the skeleton increases as well. We set skeleton height as $\skheight = 6$ based on the empirical experience. It works the best for data less than 100M in our machine, and for larger scale data, increasing $\skheight$ only leads to minor improvement. For better memory access, the tree skeleton is stored in an array using the binary heap layout.

	\myparagraph{Handling of Duplicates}
	Many current implementations of \kdtree do not well-handle the duplicated points in tree construction. For example, if all input points are same, \cgal would separate duplicated points one at a time. It leads to a very imbalanced tree with much larger tree height, where the performance for query is degenerated and the stack memory is consumed quickly for large scale data.
	For \zdtree, if all points are on one side during split, it switches to the next bit; if no bits are available, it stores all points into one leaf node.
	ParGeo includes no handling for duplicated points.  In our implementation, we address this problem by introducing the \bdita{heavy leaf} and detecting the duplicates on-the-fly. The heavy leaf is a leaf node where the stored points are identical. Heavy leaf stores only one copy of the duplicated points, which reduces the space usage and can be easily maintained for batch update.
	Specifically, after we obtain the sequence slice $R$ (\cref{algo:constr}, \cref{line:obtainSlices}), we check if all points are sieved into one slice according to the slice size, if so, we double-check whether the points can be split using serial partition. If the serial partition still fails, we conform that the points are duplicated, and then they are packed into one dummy leaf directly.

	\subsection{Query}
	\myparagraph{Handling of Bounding Box}
	Given set of points, bounding box is a rectangle that contains all points. Most query operations on \kdtree relies on bounding box to prune the search space. Many implementations, such as \zdtree, ParGeo and \cgal, store the bounding box within each tree node for easy access. However, as we mentioned earlier, this would add extra space usage to the whole tree structure, which is not negligible especially for the high-dimensional data.
	Actually, as a space-partition data structure, each node in \kdtree already represents a part of area in space that is divided by the splitters from the root to the node.
	Therefore, we can save the memory by removing the bounding box from tree nodes, but pass them as arguments during the tree construction and query.
	The new bounding boxes are computed dynamically by splitting the current bounding box using the splitter from the tree skeleton or the internal nodes. The newly constructed bounding box may not be tight, but it ensures the correctness. In fact, our experimental results indicate that passing bounding boxes from top preserves the good tree quality meanwhile incurs negligible impact on the $k$-NN query. However, for range query, especially the high-dimensional case, a false positive occurs if the loose bounding box of a tree node intersects with the query rectangle, while there is no points in this (sub)-tree are in the query rectangle.

	\myparagraph{Storing candidates in $k$-NN}
	There are various methods to store candidates during the $k$-NN search, including but not limited to:
	\begin{itemize}
		\item \zdtree uses a swap-vector for small $k$ and the priority queue in \texttt{C++ STL} for large $k$. It is reported in \cite{blelloch2022parallel} that the overhead of swap-vector for $k<40$ was significantly less than the priority queue.
		\item \cgal uses the \bdita{bounded (max)-heap} which simulates the priority queue in \texttt{STL} using a fixed size of array. When the size of heap is not full, inserting element is identical as insertion in the priority queue; otherwise, the element is inserted if and only if it is smaller than the root of the heap.
		\item ParGeo maintains a double-array with size $2\cdot k$. Similarly to the bounded heap, when the array is not full, insertion is performed by appending the element to the end of the array. Once the elements reaches the capacity of array, it picks the median of all candidates, partition points based on the median and discards the elements that are larger than the median.
	\end{itemize}

	We compare above methods using our implementation. The experimental results show that the bounded heap provides the best overall performance. For $k=100$, it is slightly faster than the priority queue and about 3x faster than the double-array used in ParGeo. For $k=10$, it keeps competitive with the swap-vector in \zdtree{}. In this case, we use the bounded heap in our implementation.

	\myparagraph{Parallel flatten in range query}
	During range query, when the query rectangle fully contains the bounding box of a tree node (in our implementation, it is the one passed during recursion), we can flatten the points into the output array \textbf{in parallel} and therefore skip the entire tree. This approach saves the time for range query significantly due to the better parallelism.
	\cgal detect the inclusion as well but perform the flatten in serial. For each tree node, ParGeo stores all points in this (sub)-tree within a separate array and the flatten is simply to copy the array to the output.
	\zdtree didn't implement the range query.

}

\begin{table*}[htbp]
	\small
	\centering
	\setlength\tabcolsep{4pt} 
	\renewcommand{\arraystretch}{1.0} 

	\begin{tabular}{c@{}c|cccc|cccc|cccc|cccc|cccc@{}}
		\toprule
		\multirow{2}[2]{*}{Bench.}                                                                      &          & \multicolumn{4}{c|}{Construction} & \multicolumn{4}{c|}{Batch Insertion (1\%)} & \multicolumn{4}{c|}{Batch Deletion (1\%)} & \multicolumn{4}{c|}{10-NN (1\%)} & \multicolumn{4}{c}{Range Report (10K)}                                                                                                                                                                                                                                                                                              \\
		                                                                                                & $D$      & 2                                 & 3                                          & 5                                         & 9                                & 2                                      & 3                & 5                & 9                & 2                & 3                & 5                & 9                & 2                & 3                & 5                & 9                & 2                & 3                & 5                & 9                \\
		\midrule
		\multicolumn{1}{c}{\multirow{4}[2]{*}{\begin{tabular}[c]{@{}c@{}}Uniform\\ 1000M\end{tabular}}} & Ours     & \underline{3.15}                  & \underline{3.65}                           & \underline{5.67}                          & \underline{9.66}                 & \underline{.104}                       & \underline{.107} & \underline{.123} & \underline{.152} & \underline{.121} & \underline{.134} & \underline{.171} & \underline{.232} & \underline{.381} & \underline{.822} & \underline{4.58} & \underline{108}  & \underline{.391} & \underline{.706} & \underline{2.31} & \underline{16.2} \\
		                                                                                                & Log-tree & 37.9                              & 45.4                                       & 58.0                                      & 92.7                             & 2.16                                   & 2.66             & 3.67             & 6.19             & .396             & .485             & 1.94             & 2.39             & 2.96             & 4.48             & 20.2             & 879              & 2.62             & 4.14             & 8.94             & 31.6             \\
		                                                                                                & BHL-tree & 31.7                              & 40.5                                       & 58.4                                      & 104                              & 31.4                                   & 40.3             & 57.1             & 103              & 30.9             & 39.3             & 68.7             & 114              & .487             & 1.02             & 7.38             & 448              & 2.06             & 2.94             & 6.53             & 23.2             \\
		                                                                                                & CGAL     & 1147                              & 1079                                       & 1217                                      & 1412                             & 1660                                   & 1815             & 1863             & 2145             & 41.2             & 41.3             & 45.0             & 40.2             & 1.04             & 2.30             & 12.5             & 189              & 311              & 282              & 249              & 184              \\
		\midrule
		\multicolumn{1}{c}{\multirow{4}[2]{*}{\begin{tabular}[c]{@{}c@{}}Varden\\ 1000M\end{tabular}}}  & Ours     & \underline{3.66}                  & \underline{4.78}                           & \underline{6.27}                          & \underline{11.2}                 & \underline{.055}                       & \underline{.107} & \underline{.157} & \underline{.350} & \underline{.049} & \underline{.112} & \underline{.127} & \underline{.237} & \underline{.172} & \underline{.210} & .336             & .433             & \underline{.382} & \underline{.745} & \underline{2.24} & 13.1             \\
		                                                                                                & Log-tree & 34.2                              & 41.8                                       & 57.8                                      & 92.6                             & 2.01                                   & 2.60             & 3.72             & 6.07             & 1.06             & 1.14             & 1.92             & 2.30             & 2.05             & 2.29             & 23.3             & 2225             & 2.63             & 4.25             & 7.95             & 14.1             \\
		                                                                                                & BHL-tree & 30.2                              & 39.2                                       & 58.7                                      & 104                              & 29.4                                   & 39.1             & 57.3             & 102              & 29.0             & 38.4             & 67.0             & 123              & .242             & .324             & .456             & .535             & 1.95             & 3.03             & 5.72             & \underline{9.96} \\
		                                                                                                & CGAL     & 429                               & 390                                        & 372                                       & 438                              & 849                                    & 700              & 582              & 599              & 13.0             & 9.53             & 23.1             & 3.90             & .511             & .217             & \underline{.318} & \underline{.392} & 296              & 283              & 253              & 278              \\
		\bottomrule
	\end{tabular}%

	\vspace{.2em}
	\caption{\textbf{Running time (in seconds) for \ourlib and other baselines. Lower is better.}
		$D$: dimensions. 
		Baselines are introduced in \cref{sec:exp}. The fastest runtime for each benchmark is underlined.
		Batch insertion is on 10M points from same distribution of the points in the tree, and batch deletion removes 10M points from the tree.
		``10-NN'': 10-nearest neighbor queries on 1\% (10M) of the points in the tree.
		``Range Report'': 10K orthogonal rectangle report queries with output size between $10^4$--$10^6$. 
		For queries, we run all of them in parallel and each query itself is run sequentially.
		\vspace{-0.5em}
	}
	\label{table:summary}
\end{table*}

\hide{
	\begin{table*}[htbp]
		\small
		\centering
		\setlength\tabcolsep{0.75pt}
		\renewcommand{\arraystretch}{0.9}

		\begin{tabular}{cc|ccccc|ccccc|cccc|cccc}
			\toprule
			                                                                                                          &      & \multicolumn{5}{c|}{Uniform-100M} & \multicolumn{5}{c|}{Varden-100M} & \multicolumn{4}{c|}{Uniform-1000M} & \multicolumn{4}{c}{Varden-1000M}                                                                                                                                                                                                                       \\
			Op.                                                                                                       & Dims & Ours                              & \zdtree                          & \logtree                           & \bhltree                         & \cgal & Ours             & \zdtree          & \logtree & \bhltree         & \cgal            & Ours             & \zdtree & \logtree & \bhltree & Ours             & \zdtree          & \logtree & \bhltree         \\
			\midrule
			\multirow{4}[2]{*}{\begin{sideways}\footnotesize Build\end{sideways}}                                     & 2    & \underline{.246}                  & .523                             & 3.92                               & 3.24                             & 98.8  & \underline{.255} & .501             & 3.71     & 2.96             & 40.6             & \underline{3.20} & 5.59    & 36.2     & 31.8     & \underline{3.65} & 5.31             & 34.4     & 29.8             \\
			                                                                                                          & 3    & \underline{.293}                  & .555                             & 4.44                               & 3.70                             & 94.4  & \underline{.316} & .530             & 4.42     & 3.66             & 33.6             & \underline{3.75} & 6.06    & 43.7     & 39.8     & \underline{4.77} & 6.89             & 43.1     & 38.8             \\
			                                                                                                          & 5    & \underline{.432}                  & n.a.                             & 5.90                               & 5.10                             & 103   & \underline{.480} & n.a.             & 5.89     & 5.11             & 38.3             & \underline{5.69} & n.a.    & 59.8     & 58.3     & \underline{6.46} & n.a.             & 58.7     & 58.4             \\
			                                                                                                          & 9    & \underline{.720}                  & n.a.                             & 9.19                               & 8.66                             & 120   & \underline{.821} & n.a.             & 8.99     & 8.64             & 42.8             & \underline{9.77} & n.a.    & 95.2     & 103      & \underline{11.2} & n.a.             & 92.2     & 103              \\
			\midrule
			\multicolumn{1}{c}{\multirow{4}[2]{*}{\begin{sideways}\begin{tabular}[c]{@{}c@{}}\footnotesize Batch insert\\\footnotesize(10\%)\end{tabular}\newline{}\end{sideways}}} & 2    & \underline{.044}                  & .145                             & 1.09                               & 3.30                             & 119   & \underline{.050} & .059             & .980     & 2.99             & 89.2             & \underline{.480} & 1.65    & 30.7     & 39.3     & \underline{.462} & 1.08             & 27.8     & 36.8             \\
			                                                                                                          & 3    & \underline{.050}                  & .148                             & 1.29                               & 3.81                             & 123   & .094             & \underline{.066} & 1.26     & 3.77             & 62.7             & \underline{.527} & 1.70    & 39.2     & 49.1     & \underline{.566} & .772             & 37.1     & 47.6             \\
			                                                                                                          & 5    & \underline{.063}                  & n.a.                             & 1.79                               & 5.57                             & 145   & \underline{.136} & n.a.             & 1.76     & 5.51             & 64.3             & \underline{.664} & n.a.    & 55.0     & 69.4     & \underline{1.46} & n.a.             & 54.9     & 69.3             \\
			                                                                                                          & 9    & \underline{.097}                  & n.a.                             & 2.93                               & 9.62                             & 170   & \underline{.221} & n.a.             & 2.84     & 9.58             & 68.0             & \underline{.985} & n.a.    & 161      & 122      & \underline{2.21} & n.a.             & 174      & 121              \\
			\midrule
			\multicolumn{1}{c}{\multirow{4}[2]{*}{\begin{sideways}\begin{tabular}[c]{@{}c@{}}\footnotesize Batch delete\\\footnotesize (10\%)\end{tabular}\end{sideways}}}           & 2    & \underline{.050}                  & .128                             & .240                               & 2.83                             & 28.5  & .076             & \underline{.056} & .310     & 2.66             & 5.88             & \underline{.530} & 1.50    & 1.50     & 32.5     & \underline{.660} & .992             & 1.96     & 30.0             \\
			                                                                                                          & 3    & \underline{.058}                  & .149                             & .370                               & 3.47                             & 35.4  & .097             & \underline{.042} & .410     & 3.42             & 4.00             & \underline{.617} & 1.72    & 4.20     & 41.8     & 1.05             & \underline{.567} & 4.92     & 41.7             \\
			                                                                                                          & 5    & \underline{.081}                  & n.a.                             & .510                               & 4.86                             & 29.6  & \underline{.143} & n.a.             & .550     & 4.87             & 3.30             & \underline{.870} & n.a.    & 2.50     & 74.3     & \underline{1.59} & n.a.             & 2.53     & 75.7             \\
			                                                                                                          & 9    & \underline{.130}                  & n.a.                             & .760                               & 8.45                             & 30.2  & \underline{.230} & n.a.             & .850     & 8.51             & 3.51             & \underline{1.27} & n.a.    & 3.40     & 134      & \underline{2.52} & n.a.             & 3.17     & 135              \\
			\midrule
			\multicolumn{1}{c}{\multirow{4}[2]{*}{\begin{sideways}\begin{tabular}[c]{@{}c@{}}\footnotesize 10-NN\\\footnotesize (all)\end{tabular}\end{sideways}}}           & 2    & \underline{1.05}                  & 1.68                             & 24.2                               & 3.93                             & 7.99  & \underline{1.04} & 1.62             & 17.9     & 2.25             & 3.68             & \underline{11.3} & 18.9    & s.f.     & s.f.     & \underline{11.4} & 18.4             & s.f.     & s.f.             \\
			                                                                                                          & 3    & \underline{2.12}                  & 3.09                             & 33.2                               & 8.33                             & 16.9  & 1.98             & 3.67             & 20.3     & 3.32             & \underline{1.95} & \underline{22.4} & 33.0    & s.f.     & s.f.     & \underline{22.0} & 30.6             & s.f.     & s.f.             \\
			                                                                                                          & 5    & \underline{10.5}                  & n.a.                             & t.o.                               & 59.6                             & 90.6  & \underline{3.31} & n.a.             & t.o.     & 5.68             & 3.89             & \underline{114}  & n.a.    & s.f.     & s.f.     & \underline{35.7} & n.a.             & s.f.     & s.f.             \\
			                                                                                                          & 9    & \underline{631}                   & n.a.                             & t.o.                               & 3040                             & 1181  & \underline{5.87} & n.a.             & t.o.     & 80.3             & 33.4             & \underline{8433} & n.a.    & s.f.     & s.f.     & \underline{56.5} & n.a.             & s.f.     & s.f.             \\
			\midrule
			\multicolumn{1}{c}{\multirow{4}[2]{*}{\begin{sideways}\begin{tabular}[c]{@{}c@{}}\footnotesize Range query\\\footnotesize(10K,100M]\end{tabular}\newline{}\end{sideways}}} & 2    & \underline{.145}                  & n.a.                             & .358                               & .220                             & 127   & \underline{.143} & n.a.             & .359     & .220             & 136              & \underline{.389} & n.a.    & 2.16     & 1.70     & \underline{.381} & n.a.             & 2.03     & 1.60             \\
			                                                                                                          & 3    & \underline{.274}                  & n.a.                             & .665                               & .414                             & 108   & \underline{.257} & n.a.             & .599     & .407             & 131              & \underline{.704} & n.a.    & 3.41     & 2.50     & \underline{.715} & n.a.             & 3.44     & 2.42             \\
			                                                                                                          & 5    & \underline{.841}                  & n.a.                             & 1.64                               & 1.08                             & 104   & .850             & n.a.             & 1.14     & \underline{.749} & 113              & \underline{2.17} & n.a.    & 7.45     & 5.08     & \underline{2.29} & n.a.             & 6.11     & 4.36             \\
			                                                                                                          & 9    & \underline{4.95}                  & n.a.                             & 7.63                               & 5.57                             & 79.0  & 2.95             & n.a.             & 2.01     & \underline{1.33} & 115              & \underline{12.8} & n.a.    & 29.3     & 20.9     & 12.8             & n.a.             & 10.8     & \underline{7.73} \\
			\bottomrule
		\end{tabular}%


		\vspace{.2em}
		\caption{\textbf{Running time (in seconds) for \ourlib and other baselines. Lower is better.}
			Baselines are introduced in \cref{sec:exp}. The fastest runtime for each benchmark is underlined.
			``n.a.'': not available. ``s.f.'': segmentation fault. ``t.o.'': time out (more than 3 hours).\vspace{-1em}
		}
		\label{table:summary}
	\end{table*}
}

\section{Experiments\label{sec:exp}}
We conducted extensive experiments to demonstrate the efficiency of the \ourlib{}.
For both synthetic (\cref{sec:exp:operation}) and real-world (\cref{sec:exp:real-world}) datasets, the \ourlib{} shows better performance than other baselines in construction, batch updates, and various queries.

We also provide in-depth studies to further understand the performance gains of \ourtree{s}.
\cref{sec:exp:cacheMemory} measures the number of cache misses in different algorithms. Our results show that the theoretical guarantee for \ourtree{s} (\cref{thm:constr,thm:update}) indeed allows for better cache-efficiency and leads to good performance in practice.
\cref{sec:tech:analysis} studies the two techniques in our tree construction algorithm, sampling and constructing multiple levels, and show that they lead to roughly 2$\times$ and 4$\times$ performance gains, respectively.
Since \ourtree{s} are weight-balanced, in \cref{sec:balancingParameterRevisted} we show how the balancing criterion affects the update and query performance, and explain how we choose the parameters in the \ourtree{}.
Finally, we show that the \ourtree{} has good parallel scalability in \cref{sec:exp:scalability}.

\myparagraph{Setup.}
We use a machine with 96 cores (192 hyperthreads) with four-way Intel Xeon Gold 6252 CPUs and 1.5 TB RAM.
Our implementation is in \texttt{C++} using ParlayLib~\cite{blelloch2020parlaylib} to support fork-join parallelism.
The reported numbers are the average of three runs after a warm-up run.
This approach ensures that all timed runs begin with a consistent cache configuration, resulting in more stable performance.
We use $\lambda=6$ as explained in \cref{sec:impl}.
We set $\alpha=0.3$, and discuss the choice of this parameter in \cref{sec:balancingParameterRevisted}.

\myparagraph{Baselines.}
We compare the \ourtree{} with three existing implementations, described as follows.
\begin{itemize}
	\item \textbf{\bhltree}~\cite{wang2022pargeo}. The plain implementation of a parallel \kdtree{} from \pargeo{} uses a single tree structure with the binary heap layout.
	      \bhltree{} uses the \emph{plain parallel \kdtree{}} construction algorithm described in \cref{sec:constr}, which is work-efficient but not cache-optimized. Its batch update simply rebuilds the whole tree.
	\item \textbf{\logtree}~\cite{wang2022pargeo}. The \kdtree{} implementation based on the logarithmic method from \pargeo{}. The \logtree{} keeps $O(\log n)$
	      static cache-oblivious \kdtree{s} (using the vEB layout)
	      with exponentially increasing sizes.
	      A batch update is performed by merging and rebuilding a subset of the trees. 
	      Queries have to be performed on all $O(\log n)$ trees and the results need to be combined.
	\item \textbf{\cgal}~\cite{cgal51}. The \kdtree{} in the computational geometry library \cgal{}. \cgal{} supports parallelism using the threading building blocks (TBB)~\cite{TBB}.
	      During construction, \cgal{} partitions the points sequentially, then builds two subtrees in parallel.
	      A batch insertion rebuilds the whole tree with the inserted points; a batch deletion removes the points one by one.  \hide{\yihan{mention this later.} There can be scalability issues for \cgal when using more than 36 threads, as reported in~\cite{blelloch2022parallel} and in our \cref{fig:scalability}. }
\end{itemize}

\myparagraph{Datasets.\label{sec:exp:datasets}}
We test both synthetic and real-world datasets.
We introduce the real-world datasets in \cref{sec:exp:real-world}.
For synthetic datasets, we use 64-bit integer coordinates with two distributions: \varden{} and \uniform{}.
\varden{} is a skewed distribution from~\cite{gan2017hardness}.
It generates points by a random walk with a low probability of restarting from a random place.
Therefore, it contains some very dense subareas that can be far from each other, which can be used as pressure tests for the quality of \kdtree{s} as well as the performance for frequent rebalancing.
\uniform{} draws points within a box uniformly at random.
For simplicity, we shorthand each dataset by the dimension, distribution and size.
For instance, ``3D-V-1000M'' stands for 1000 million points in 3 dimensions from \varden{}.




\subsection{Operations on Synthetic Datasets}\label{sec:exp:operation}
\myparagraph{Overall Performance.}
We summarize the performance for the \ourtree{} and other baselines in \cref{table:summary} with tree size of $10^9$, in dimensions $D\in\{2,3,5,9\}$.
\ifconference{We also include an experiment on a synthetic dataset with 12 dimensions for all baselines in the full paper.}\iffullversion{We also include an experiment on a synthetic dataset with 12 dimensions for all baselines, as detailed in \cref{app:synthetic}.}
Each batch for insertion or deletion contains $10^7$ (1\% of the tree size) points from the same distribution.
We test two query types: 1) $10^7$ of 10-NN queries, 2)
$10^4$ range report queries (output sizes $10^4$--$10^6$).
Queries are performed directly after construction.
Different queries run in parallel, and each query runs in serial.
No other baselines support range count queries, so we test this query on \ourtree{} separately in \cref{fig:rquery}.

For construction, the \ourtree{} is the fastest in all tests, which is
8.26--12.5$\times$ faster than \logtree{s}, 8.20--11.1$\times$ faster than \bhltree{s}, and 39.1--363$\times$ faster than \cgal{}.
The high performance is mainly from good \emph{cache-efficiency} and \emph{scalability}, which will be studied in more details in \cref{sec:exp:scalability,sec:exp:cacheMemory}, respectively.
\cgal{} has a known scalability issue~\cite{blelloch2022parallel} (also see the scalability curve in \cref{fig:scalability}), making it much slower than other implementations in construction.

\hide{
	It is worth noting that \zdtree{} is a quad/octree based on the space-filling curve,
	which handles multidimensional points as one-dimensional points.
	Therefore, the computation is simpler than the \kdtree-based structures in construction and update,
	but it cannot efficiently handle dimensions $D>3$.
	Indeed, it is slightly faster than \ourtree{} in some batch update operations in low dimensions.
	Interestingly, \ourtree{} is still faster than \zdtree{} in all construction tests and most batch updates. We believe
	this is also due to the I/O optimization in \ourtree{s}.
}

The \ourtree{} also outperforms all baselines in batch updates.
\bhltree{s} and \cgal{} always rebuild the entire tree after updates,
so \ourtree{s} are orders of magnitude faster than them, especially on small batches.
Compared to the fastest baseline \logtree{}, \ourtree{s} are 17.4--40.7$\times$ faster in insertions and 3.27--21.7$\times$ faster in deletions,
mainly due to two reasons:
1) \ourlib{s} sieve the updated points to the subtrees in a cache-efficient way,
and 2) 
both \ourlib{s} and \logtree{s} may reconstruct some (sub)trees in batch updates, and \ourlib{s} are faster in tree construction as discussed above.

For both \knn{} and range queries, the \ourtree{} is the fastest except for three cases, all in high-dimensional queries.
It is within 1.1$\times$ slower than \cgal{} in two cases,
and 1.32$\times$ slower than the \bhltree{} in one case.
As mentioned in \cref{sec:impl}, \ourtree{s} do not store the bounding boxes to optimize memory usage,
but compute the subspaces for each subtree on-the-fly in queries.
This approach trades off (slower) query performance for (faster) construction and updates,
which can be more pronounced in higher dimensions.
Even so, \ourtree{} is still the fastest in queries for 13 out of 16 instances.
Therefore, we choose not to maintain the bounding boxes in tree nodes
to achieve better performance for both updates and queries.

Another interesting finding is that almost all \kdtree{s} perform better on \varden{} than \uniform{} for \knn{} queries.
Since \varden{} datasets contain dense subareas, the neighbors are usually in these regions, resulting in more effective pruning than the uniform datasets.

In the following, we present an in-depth performance study for updates and queries.
We use synthetic datasets with size $10^9$ points in 3D as the benchmark for the rest of this section.

\myparagraph{Batch Updates.}
To further understand the performance of batch updates, we vary the batch sizes from $10^5$ to $10^9$, and show the results in
\cref{fig:batchUdpate}.
We omit smaller batch sizes because they can be completed quickly anyway and do not have high demand for parallelism.
We first construct a tree with $10^9$ points. Then a batch insertion inserts a batch from the same distribution into the tree;
the batch deletion removes a batch of points from the tree.

\ourlib{s} have the best performance for all instances on all distributions.
Both the \bhltree{} and \cgal{} fully rebuild the tree on insertions, showing a flat curve of running time with varying batch sizes, which are much slower than the \ourtree{} based on rebalancing.
A batch deletion in \cgal{} removes the points one by one, which performs well for small batches, but is very inefficient for large batches.
The \logtree{}'s performance sits in the middle.
It avoids fully rebuilding the tree by merging a subset of the trees for batch insertions.
One may notice that there are several jumps for the \logtree{} in batch insertions.
This is because when the batch size reaches certain threshold, a reconstruction for a large tree may be triggered, causing significant more time.
\ourlib{s} have the most stable and efficient performance across tests.

\begin{figure}[t]
	\includegraphics[width=0.48\textwidth]{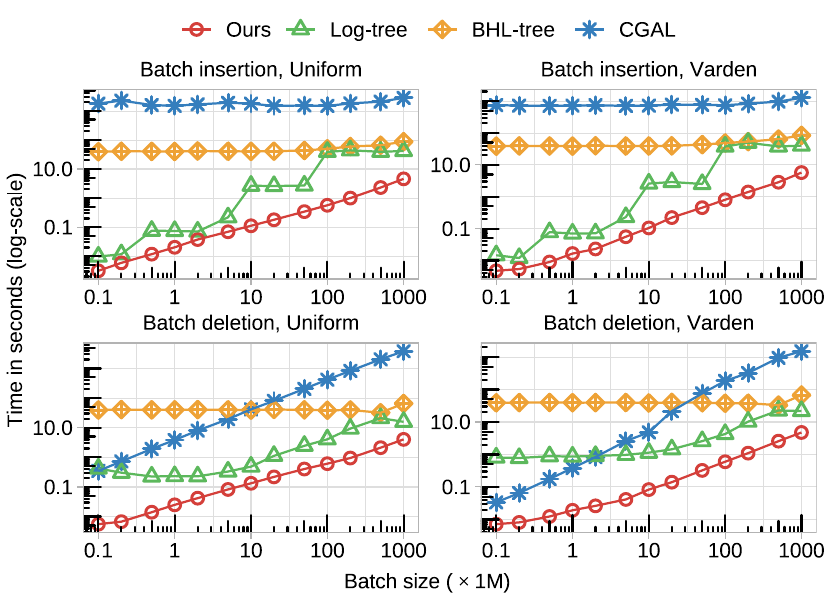}
	\caption{\textbf{Time required for batch update on points from \varden and \uniform on a tree with 1000M points in 3 dimensions. Lower is better.} The batch size is the number of points ($\times$ 1M) in the batch. The time is measured in seconds. Both axes are in log-scale. }
	\label{fig:batchUdpate}
\end{figure}

\myparagraph{$k$-NN Queries.}
We now study $k$-NN queries on \uniform and \varden with $k\in\{1, 10, 100\}$.
We consider $10^9$ input points in 3D and call \knn{} queries on the first $10^7$ input points in parallel.
Results are presented in \cref{fig:knn}.
\ifconference{We also measure out-of-distribution \knn{} queries in the full paper.}\iffullversion{We also measure the out-of-distribution \knn{} queries in \cref{app:ood}.}
The \ourtree{} is always among the fastest.
The performance of the \bhltree{} and \cgal{} is similar since they also keep a single tree.
The \ourtree{} is slightly faster due to not storing bounding boxes in the tree nodes (see \cref{sec:impl})---it saves the memory footprint at the cost of less efficient pruning.
Overall it gives some small advantages on $k$-NN query performance in low dimensions.
\logtree{s} have significantly worse performance for $k$-NN queries---5.85--11.7$\times$ slower on \uniform and 12.0--21.9$\times$ slower on \varden{} than \ourtree{s}.
This is because a query needs to search in all the $O(\log n)$ trees and merge the results.

\hide{
	Most other baselines are also competitive within \edit{1.33--21.9$\times$.}
	\edit{
		The high space usage in the \logtree{} and the cost of querying on all $O(\log n)$ \kdtree{s} significantly slows down its performance, which is 5.85--21.9$\times$ slower than the \ourlib{}.
		The difference between the \ourlib{} and other baselines is relatively small; this is reasonable as both the \bhltree{} and \cgal{} are standard \kdtree{s} as the \ourlib{}.
		We believe the low memory usage of the \ourlib{}, such as excluding the bounding box (see \cref{sec:impl}), differs it from others:
		reducing the memory footprint of each tree node allows more nodes to fit in the cache simultaneously, even though the absence of the bounding box reduces the pruning efficiency to some extent.
	}}

\begin{figure}[t]
	\includegraphics[width=0.48\textwidth]{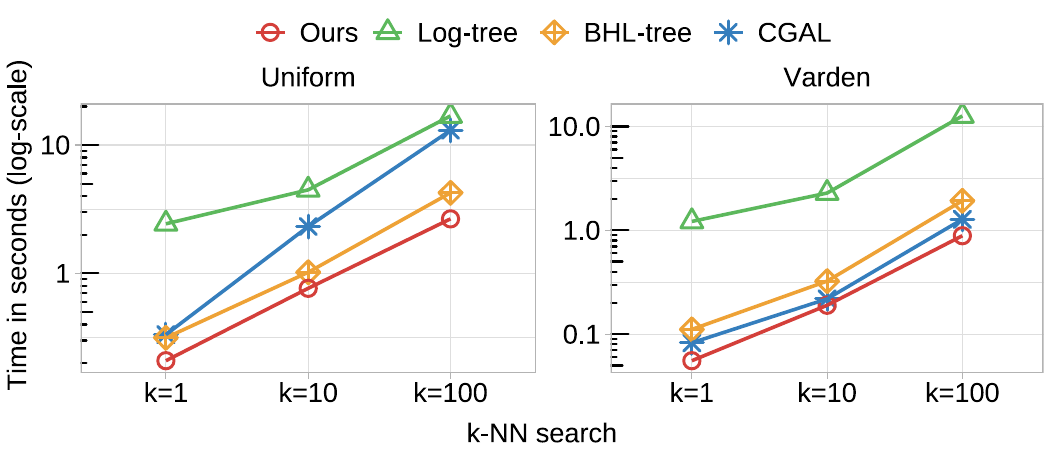}
	\caption{\textbf{Running time (in seconds) of $k$-NN queries for $k\in\{1,10,100\}$. Lower is better.}
		The dataset contains 1000M points in 3 dimensions. The test contains $k$-NN queries from $10^7$ points in the input.
		Plots are in log-log scale.
	}
	\label{fig:knn}
\end{figure}

\myparagraph{Range Queries.\label{sec:expRangeQuery}}
In \cref{table:summary}, we test range reports with relatively large output sizes
$10^4$--$10^6$ to make the queries more adversarial.
Since the performance of range report is proportional to the output size,
in this section, we conduct additional tests on range-count and range-report queries
with a variety of output sizes.
Note that although small query sizes are more frequently encountered in practice, range queries with large output sizes are also prevalent in various applications, including dynamic programming~\cite{gu2023parallel}, etc.
The \ourtree{} is the only one that supports range count.
We run all queries sequentially to measure the query time w.r.t. the output size in \cref{fig:rquery}.
The \ourtree{} (red circles in \cref{fig:rquery}) generally have the best performance across a wide range of output sizes from $1$ to $10^6$.
The \bhltree{} is competitive.
\logtree{} is slower than \ourtree{} and \bhltree{} since it has to query $O(\log n)$ trees.
Note that when a subtree is fully contained in the query box,
one can output all points in the subtree in parallel.
This has been incorporated into all baselines except for \cgal{},
which makes \cgal{} particularly slow with large output sizes.
When the query only requires the count of points in the range, the range count function on \ourtree{s} (blue rectangles in \cref{fig:rquery}) can be much faster than reporting all the points, especially with large output sizes.
This indicates the necessity of including range-count in the interface of a \kdtree{}.
\hide{
	\edit{
		We generate 100 queries for each type and answer the query one-by-one.
		Each query searches the tree in serial.
	}
	We also include the results for our range count queries for better comparison.
	Results are in \cref{fig:rquery}. Overall, the \ourtree{} has the best performance.
	\edit{
		Other baselines are competitive.
		It is worthy mention that range count queries are generally faster than range report queries, especially for those with large output sizes, due to the saved work of reporting points.
		This indicates the necessity of directly providing such an interface in the library rather than using the size of the range report as an alternative.
	}}

\begin{figure}[t]
	\includegraphics[width=0.48\textwidth]{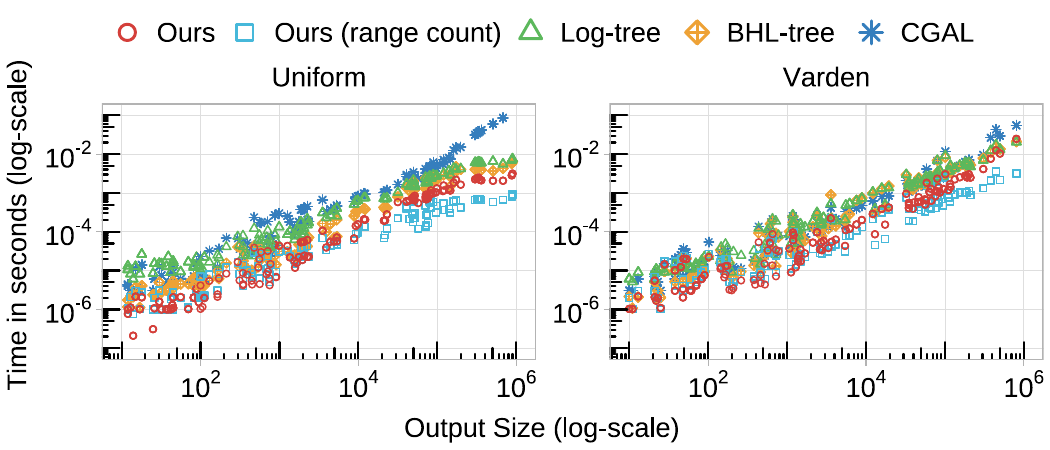}
	\caption{\textbf{Running time (in seconds) of range queries w.r.t.\ output sizes. Lower is better.}
		The dataset contains 1000M points in 3 dimensions.
		Plots are in log-log scale.}
	\label{fig:rquery}
\end{figure}

\hide{ \myparagraph{Range Count.\label{sec:exp:range-count}}
	A range count query is similar to a range query but only reports the number of points in the queried range.
	Range count is a crucial subroutine in many applications~\cite{schubert2017dbscan,agarwal1999geometric,wu1998range,patwary2012new,zhang2016privtree}.
	However, no baseline provides a range count interface, and users have to call a regular range query and report the output size.
	\ourtree{} answers range count queries efficiently by using the subtree sizes stored in each tree node:
	when a subtree is totally in the query range, we skip the subtree and add its size to the return value.
	\cref{table:range-count} presents the results.
	Compared with the running time on range queries, a range count query is faster than
	reporting all elements in the range.
	For example, for 3D-U-100M and 3D-V-100M, a range count query is up to 3.0$\times$ faster and up to 3.8$\times$ faster than a range query, respectively.
	\hide{
		Compared with range queries,
		range count queries can be answered much faster (about 14$\times$ faster per query).
	}

	\begin{table}[t]
\vspace{.2em}
	\centering

	\small
	\setlength\tabcolsep{2.5pt}

	\begin{tabular}{c|ccc|ccc}
		\toprule
		              & \multicolumn{3}{c|}{\uniform} & \multicolumn{3}{c}{\varden}                                                                      \\
		\textbf{Dims} & \textbf{Small}                & \textbf{Medium}             & \textbf{Large} & \textbf{Small} & \textbf{Medium} & \textbf{Large} \\
		\midrule
		2             & 0.001                         & 0.008                       & 0.014          & 0.001          & 0.008           & 0.013          \\
		3             & 0.004                         & 0.032                       & 0.093          & 0.001          & 0.026           & 0.069          \\
		5             & 0.009                         & 0.212                       & 0.698          & 0.028          & 0.180           & 0.627          \\
		9             & 0.272                         & 2.18                        & 5.12           & 0.010          & 0.792           & 2.95           \\
		\bottomrule
	\end{tabular}%

	\caption{\textbf{Orthogonal range count query on our \ourlib for three types of rectangle in parallel.} The benchmark contains 100M points with dimension 3. Each type includes 10K random chosen rectangles. }
	\label{table:range-count}%
\end{table}%
 }


\hide{

	\subsection{Extendability}\label{sec:exp-ext}
	Above evaluations focus on low-dimension space and synthetic datasets, we now present the experimental results for tree construction and $k$-NN search on high-dimensional synthetic datasets and real-world datasets. Still, our \ourlib shows the best performance among all benchmarks. Specifically, for high-dimensional space, \ourlib constructs the tree 13$\times$ faster than \bhltree and \logtree, meanwhile it answers the $k$-NN query about 1.5$\times$ faster than \cgal and 2.4$\times$ faster than others. For Real-World datasets, \ourlib is the only one that can constructs the tree and perform the $k$-NN search for all benchmarks

	\myparagraph{High Dimension Space.}
	We compare the time needed for tree construction and $10$-NN search between \ourlib and other baselines (except \zdtree due to its precision issue on high dimension) using synthetic datasets with dimensions in $\{2,3,5,7,9\}$. \cref{fig:highDim} demonstrates the results. The time for tree construction of \ourlib is much less than others and grows slowly as the increases of dimension, which verifies the I/O-efficiency and good parallelism in \cref{algo:constr}. Meanwhile, the \ourlib remains high efficiency in $k$-NN search, which is about twice faster than \cgal and \bhltree, and four times faster than \logtree. We attribute this to the heuristic search strategy employed and careful implementation as discussed in \cref{sec:impl}.

	\begin{figure}[t]
		\includegraphics[width=0.48\textwidth]{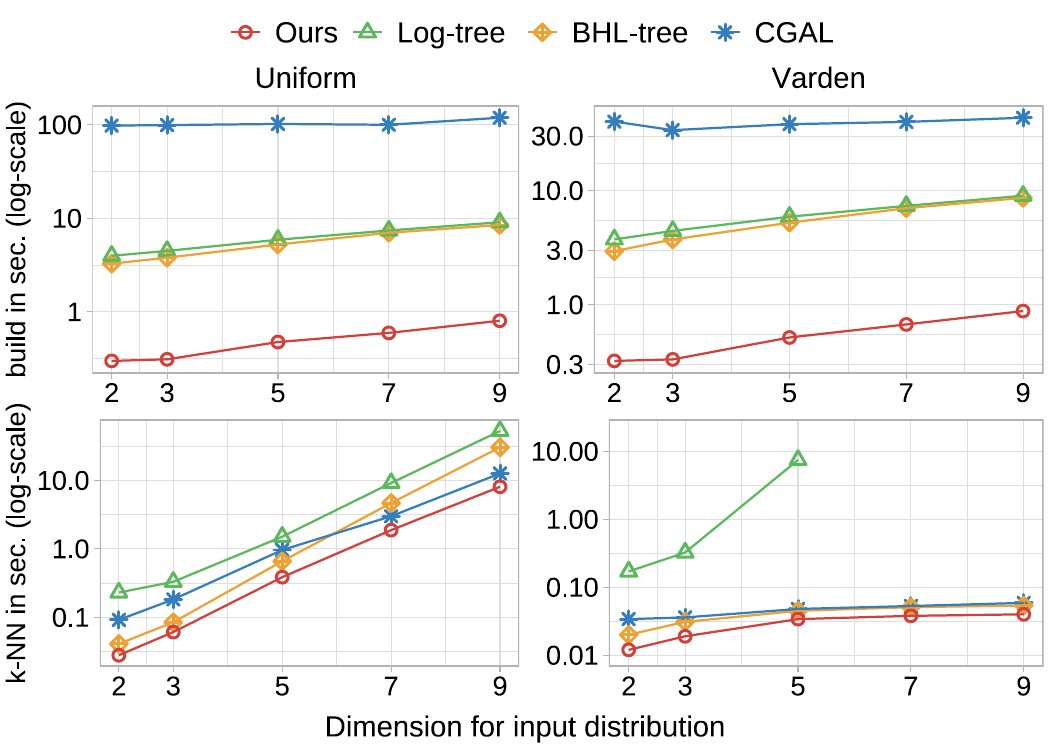}
		\vspace*{.05em}
		\caption{Tree construction and $k$-NN time for high dimension space. Each benchmark contains 100M points. The $k$ is set to $10$. Both time is measured in seconds and transformed to log-scale.}
		\label{fig:highDim}
	\end{figure}
}

\subsection{Real-World Datasets\label{sec:exp:real-world}}

\begin{table}[t]
	\centering

	\small
	\setlength\tabcolsep{4pt}
	\renewcommand{\arraystretch}{0.8}

	\begin{tabular}{c|cc|ccccc}
		\toprule
		                                                       & \textbf{Points}           & \textbf{Dims}          & \textbf{Op.} & \textbf{Ours}    & \textbf{\logtree} & \textbf{\bhltree} & \textbf{\cgal}   \\
		\midrule
		\multirow{4}[2]{*}{\begin{sideways}HT\end{sideways}}   & \multirow{4}[2]{*}{928K}  & \multirow{4}[2]{*}{10} & Build        & \underline{.008} & .678              & .061              & .472             \\
		                                                       &                           &                        & 1-NN         & \underline{.008} & 2.53              & .015              & .015             \\
		                                                       &                           &                        & 10-NN        & .043             & 2.81              & .059              & \underline{.020} \\
		                                                       &                           &                        & Range        & \underline{.095} & .651              & .478              & 1.38             \\
		\midrule
		\multirow{4}[2]{*}{\begin{sideways}HH\end{sideways}}   & \multirow{4}[2]{*}{2.04M} & \multirow{4}[2]{*}{7}  & Build        & \underline{.054} & .716              & .102              & t.o.             \\
		                                                       &                           &                        & 1-NN         & \underline{.058} & 1.26              & 1.60              & -                \\
		                                                       &                           &                        & 10-NN        & \underline{.229} & 2.60              & 3.19              & -                \\
		                                                       &                           &                        & Range        & \underline{.080} & .819              & .564              & -                \\
		\midrule
		\multirow{4}[2]{*}{\begin{sideways}CHEM\end{sideways}} & \multirow{4}[2]{*}{4.21M} & \multirow{4}[2]{*}{16} & Build        & \underline{.059} & 7.07              & .786              & 2.52             \\
		                                                       &                           &                        & 1-NN         & .042             & 16.1              & .123              & \underline{.035} \\
		                                                       &                           &                        & 10-NN        & 3.53             & 17.3              & \underline{3.32}  & 3.95             \\
		                                                       &                           &                        & Range        & \underline{.412} & 4.28              & 2.64              & 3.14             \\
		\midrule
		\multirow{4}[2]{*}{\begin{sideways}GL\end{sideways}}   & \multirow{4}[2]{*}{24M}   & \multirow{4}[2]{*}{3}  & Build        & \underline{.256} & 1.34              & .792              & s.f.             \\
		                                                       &                           &                        & 1-NN         & \underline{.274} & 3.74              & 1.31              & -                \\
		                                                       &                           &                        & 10-NN        & \underline{.775} & 14.4              & 9.37              & -                \\
		                                                       &                           &                        & Range        & \underline{.192} & 1.40              & 1.30              & -                \\
		\midrule
		\multirow{4}[2]{*}{\begin{sideways}CM\end{sideways}}   & \multirow{4}[2]{*}{321M}  & \multirow{4}[2]{*}{3}  & Build        & \underline{1.54} & 16.7              & 13.3              & 184              \\
		                                                       &                           &                        & 1-NN         & \underline{2.79} & 25.9              & 5.24              & 5.94             \\
		                                                       &                           &                        & 10-NN        & \underline{9.09} & s.f.              & s.f.              & 33.0             \\
		                                                       &                           &                        & Range        & \underline{.136} & 1.88              & 1.63              & 26.0             \\
		\midrule
		\multirow{4}[2]{*}{\begin{sideways}OSM\end{sideways}}  & \multirow{4}[2]{*}{1298M} & \multirow{4}[2]{*}{2}  & Build        & \underline{5.08} & 51.3              & 56.6              & 497              \\
		                                                       &                           &                        & 1-NN         & \underline{8.73} & 134               & 13.0              & 10.5             \\
		                                                       &                           &                        & 10-NN        & \underline{16.5} & 214               & 30.6              & 22.6             \\
		                                                       &                           &                        & Range        & \underline{.107} & 4.87              & 3.80              & 62.9             \\
		\bottomrule
	\end{tabular}%

	\caption{\textbf{Tree construction and $k$-NN time on read-world datasets for \ourlib and baselines. Lower is better.} The ``Points'' is the number of points in the datasets and ``Dim.'' is the dimension for the points. \knn{} queries are performed in parallel on all points in the dataset. 
		``Range'' is the time for $10^3$ range report queries with output size between $10^4$--$10^6$.
		The fastest runtime for each benchmark is underlined.
		``s.f.'': segmentation fault. ``t.o.'': time out (more than 3 hours).
	}
	\label{table:real-world}%
\end{table}%

We test our \ourtree{} and other baselines on real-world datasets, including
very large ones \texttt{COSMOS} (\texttt{CM})~\cite{scoville2007cosmic} and
the Northern American region of \texttt{OpenStreetMap} (\texttt{OSM})~\cite{haklay2008openstreetmap} with up to 1.298 billion points, high-dimensional datasets \texttt{HT}~\cite{huerta2016online}, \texttt{CHEM}~\cite{fonollosa2015reservoir,wang2021fast}, and \texttt{HouseHold} (\texttt{HH})~\cite{household} with up to 16 dimensions, and \texttt{GeoLife} (\texttt{GL})~\cite{geolife}, with highly duplicated points.
All coordinates are 64-bit real numbers.
Experiments on these real-world datasets show similar trends as in the synthetic datasets in \cref{sec:exp:operation}.


\myparagraph{Tree Construction, \knn{} and Range Report.}
We present the running time for tree construction, 1-NN and 10-NN queries, and range report queries for all input points in \cref{table:real-world}.
\cgal{} fails to build \texttt{HH} and \texttt{GL} due to the inability to handle heavy duplicates.
\bhltree{s} and \logtree{s} cannot process 10-NN queries on \texttt{CM} due to high memory usage.

The \ourlib{} shows the best performance in all but three instances of queries in high dimensions, which is
consistent with the results in \cref{table:summary}.
\cgal{} is faster in 10-NN on \texttt{HT} ($D=10$) and 1-NN on \texttt{CHEM} ($D=16$).
The \bhltree{} is also slightly faster than \ourtree{} in 10-NN on \texttt{CHEM}.
The \bhltree{} is the best baseline on these datasets,
but it is still 1.89--13.3$\times$ slower in construction and
1.37--27.6$\times$ slower in query than the \ourtree{},
except for 10-NN in \texttt{CHEM}, where it is 1.06$\times$ faster than the \ourlib{}.

\myparagraph{Dynamic Updates.}
The points from \texttt{OpenStreetMap} (\texttt{OSM})~\cite{haklay2008openstreetmap} are associated with time stamps, so we acquire the batch of updates per year (2014 to 2023) to compare the performance of batch updates.
We simulate a sliding-window setting.
We start with building a tree of the data in 2014.
In each year, we insert the corresponding batch, and delete the batch 5 years ago if applicable.
After the update of each year, we perform a 10-NN query on another $10^7$ points.
\cref{fig:osm-batch} presents the performance for all tested algorithms.

\begin{figure*}[t]
	\centering
	\vspace{-1em}
	\begin{minipage}[t]{0.65\textwidth}
		\includegraphics[width=1.0\textwidth]{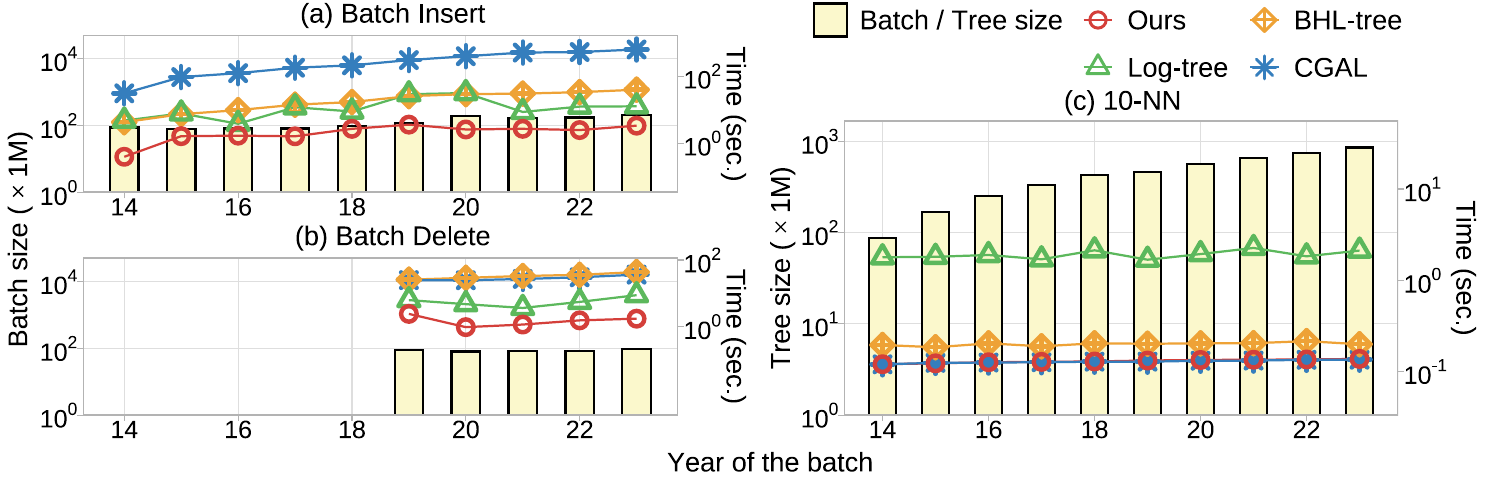}
		\caption{
			\textbf{Batch update using a sliding window spanning five years and 10-NN queries on \texttt{OSM}~\cite{haklay2008openstreetmap}.
				Lower is better.}
			The input is batched by years from 2014 to 2023.
			In each year, we insert the corresponding batch,
			and delete the batch from five years ago if applicable.
			After the update, we perform $10^7$ 10-NN queries in parallel.
			In (a) and (b), bars (left axis) are the batch size, and lines (right axis) depict the time for batch insertion and deletion in every year respectively.
			In (c), bars (left axis) are tree size, and
			lines (right axis) show the 10-NN query time.
			Vertical axes are in log scale.
		}
		\label{fig:osm-batch}
	\end{minipage}\hfill
	\begin{minipage}[t]{0.323\textwidth}
		\includegraphics[width=1.0\textwidth]{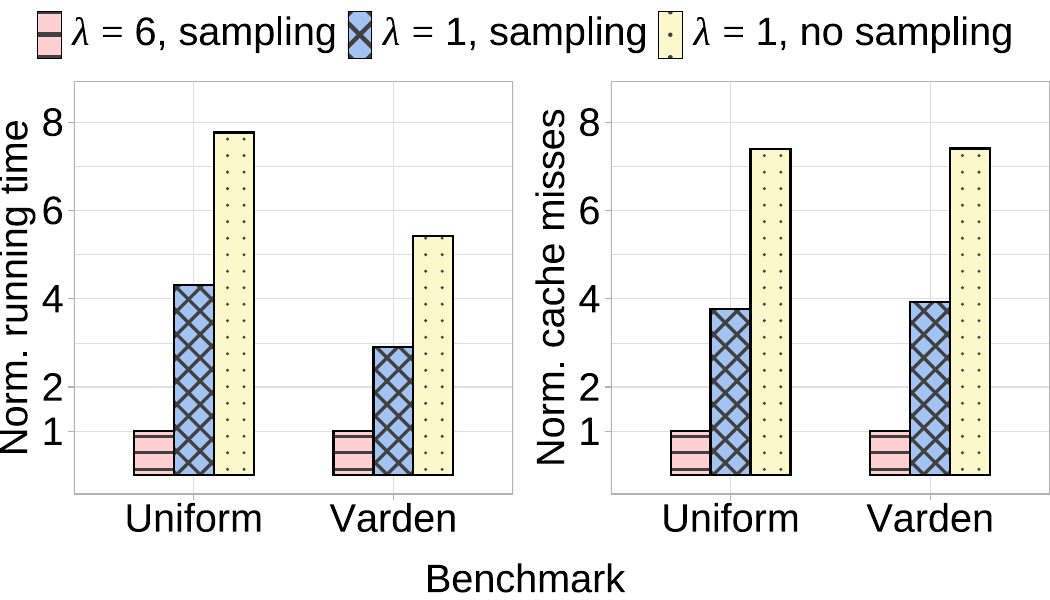}
		\caption{\textbf{The evaluation of the performance gain for techniques in tree construction. Lower is better.}
			The datasets have 1000M size in 3 dimensions. The y-axis normalized to the final version with $\skheight = 6$ and using sampling.}
		\label{fig:technique}
	\end{minipage}
	\vspace{-1em}
\end{figure*}

\captionsetup{labelfont=bf, textfont={}}

\ourtree{s} significantly outperform all baselines in all batch updates.
For insertions, \ourlib{s} is 1.82--9.55$\times$ faster than \logtree{s}, 3.81--12.0$\times$ than \bhltree{s} and 58.3--214$\times$ than \cgal{}.
For deletions, \ourtree{} outperforms the fastest baseline \logtree{} by 2.44--4.55$\times$.
For \knn{} queries, \ourlib{}, \bhltree, and \cgal{} have similar performance.
The \ourtree{} is slightly faster due to various optimizations (e.g., avoiding storing bounding boxes).
\logtree{} can be much slower than other implementations due to the use of $O(\log n)$ trees.
These conclusions are consistent with the results on synthetic datasets shown in \cref{sec:exp:operation}.

\subsection{In-Depth Performance Study \label{sec:exp:cacheMemory}}
\hide{
	Since the \ourtree{} and all baselines have $O(n\log n)$ work and all but \cgal{} have good parallel scalability (see \cref{sec:exp:scalability} for details),
	the main advantage of the \ourtree{} comes from better cache efficiency, as discussed in \cref{sec:constr,sec:update}.
}

\myparagraph{Cache Efficiency and Memory Usage for Tree Construction.}
One major effort in our work is to achieve low cache complexity for construction of \ourtree{s}.
In this section, we quantitatively verify this and show that the cache-efficiency indeed contributes to the high performance of \ourtree{s}.
\cref{fig:cacheMemory} shows the numbers of cache misses in tree construction for all implementations, along with the memory usage, which is also crucial for practical performance. 

\begin{figure}[t]
	\includegraphics[width=.42\textwidth]{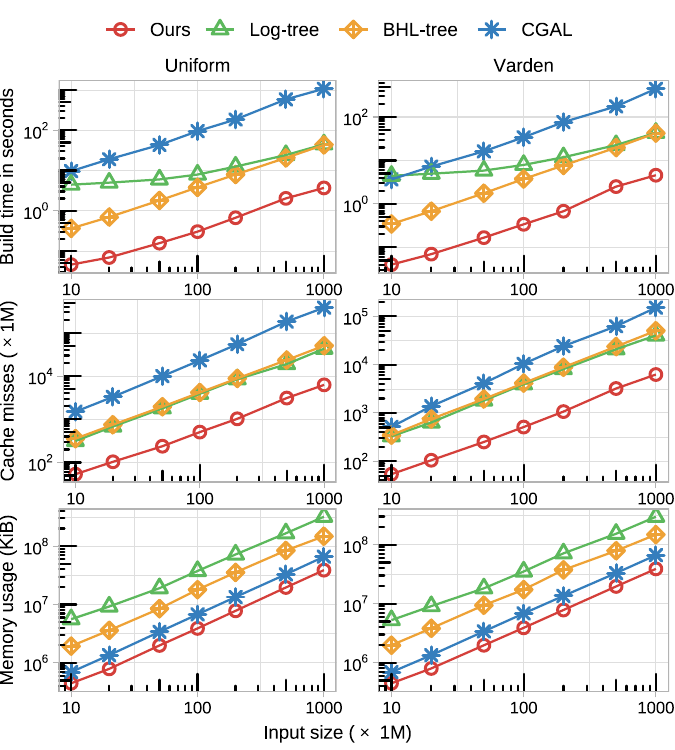}
	\caption{
		\textbf{Time, cache misses, and memory usage needed during the tree construction. Lower is better. }
		Plots are in log-log scale. All points are in 3 dimensions.
	}
	\label{fig:cacheMemory}
\end{figure}

The results verified that our performance gain is consistent with the reduction of the cache misses,
demonstrating the importance of cache-efficiency in parallel algorithms.
\ourtree{s} incur 6--12$\times$ less cache misses than \bhltree{s} or \logtree{s},
which is close to the speedup of \ourtree{} over \bhltree{} or \logtree{} in construction.
The reported cache misses in 1000M-3D-V construction indicates over 80\% of the memory bandwidth usage of the testing machine. 
We also verify this using the Intel\textregistered{} VTune profiler~\cite{IntelVTune2024}.
On our machine with a peak memory bandwidth of 443.78 GB/s, \ourtree{} has a 327--421 GB/s usage of bandwidth in many stages during construction and update.
This indicates that our construction and update algorithms are memory bottlenecked, and optimizing the cache complexity almost linearly contributes the performance.

\myparagraph{Study of the Impact of Bounding Boxes.}
As mentioned, \ourlib{s} avoid storing bounding boxes within tree nodes,
but compute them on-the-fly, which may be looser than the actual bounding boxes of the subtree.
Therefore, \ourtree{} reduces memory accesses and memory usage in construction and batch updates,
but may introduce additional computation and tree traversal during queries.
To study this trade-off, we measure the \ourlib{} and its variant storing bounding boxes (referred to as \oursbb{}) in terms of range report queries on real-world datasets.
\cref{table:perf-mix} shows the time, instruction per cycle (IPC), cache references (CRs) and cache-misses (CMs), as well as the average number of nodes visited per query.
\ifconference{We present the full results with more profiling results
	in the full paper.}\iffullversion{We refer the readers to \cref{app:profiling} for the full results.}
\oursbb{} allows for more effective prune, evidenced by the fewer number of nodes visited during search.
This benefit is more significant on high-dimensional datasets \texttt{HT}, \texttt{HH} and \texttt{CHEM}.
Due to visiting fewer nodes, on these three datasets, \oursbb{} also has lower cache references and cache misses, and is thus faster.
On low-dimensional datasets \texttt{CM} and \texttt{OSM}, the difference between the computed subspaces and bounding boxes is small.
In this case, \ourtree{} and \oursbb{} visited similar numbers of tree nodes.
Thus, \ourtree{} achieves lower time than \oursbb{} in queries due to fewer memory accesses.

\begin{table}[t]
	\centering

	\small
	\setlength\tabcolsep{3pt}
	\renewcommand{\arraystretch}{1.1}


	\begin{tabular}{cc|cccccc}
		\toprule
		                                                       & \textbf{Tree} & \textbf{Time(sec.)} & \textbf{\# Leaf}  & \textbf{\# Interior} & \textbf{IPC}     & \textbf{CRs(M)} & \textbf{CMs(M)} \\
		\midrule
		\multirow{2}[2]{*}{\begin{sideways}HT\end{sideways}}   & Pkd           & .587                & 2,675             & 3,957                & .326             & 926             & 508             \\
		                                                       & Pkd-bb        & \underline{.268}    & \underline{207}   & \underline{891}      & \underline{.506} & \underline{245} & \underline{134} \\
		\midrule
		\multirow{2}[2]{*}{\begin{sideways}HH\end{sideways}}   & Pkd           & .385                & 2,817             & 3,621                & .320             & 557             & 361             \\
		                                                       & Pkd-bb        & \underline{.192}    & \underline{615}   & \underline{1,135}    & \underline{.387} & \underline{242} & \underline{150} \\
		\midrule
		\multirow{2}[2]{*}{\begin{sideways}CHEM\end{sideways}} & Pkd           & 1.15                & 4,276             & 5,701                & \underline{.271} & 1,662           & 1,450           \\
		                                                       & Pkd-bb        & \underline{.837}    & \underline{1,330} & \underline{2,506}    & .235             & \underline{954} & \underline{820} \\
		\midrule
		\multirow{2}[2]{*}{\begin{sideways}GL\end{sideways}}   & Pkd           & .329                & 3,268             & 5,478                & \underline{.303} & 439             & 345             \\
		                                                       & Pkd-bb        & \underline{.291}    & \underline{1,285} & \underline{3,484}    & .215             & \underline{407} & \underline{317} \\
		\midrule
		\multirow{2}[2]{*}{\begin{sideways}CM\end{sideways}}   & Pkd           & \underline{.531}    & 2,456             & 3,939                & \underline{.186} & \underline{692} & \underline{649} \\
		                                                       & Pkd-bb        & .577                & \underline{2,195} & \underline{3,785}    & .165             & 752             & 703             \\
		\midrule
		\multirow{2}[2]{*}{\begin{sideways}OSM\end{sideways}}  & Pkd           & \underline{.326}    & 529               & 1,243                & \underline{.171} & \underline{460} & \underline{426} \\
		                                                       & Pkd-bb        & .363                & \underline{236}   & \underline{959}      & .148             & 473             & 441             \\
		\bottomrule
	\end{tabular}%

	\caption{
		\textbf{Performance comparison of the original \ourlib{} (Pkd) and a variant with bounding boxes (Pkd-bb) for range report queries on real-world datasets. The best performance is underlined.
		}
		The query contains $10^4$ range report queries with output size $10^4$--$10^6$.
		``Time'': Time for all queries in seconds, ``Leaf'': Average number of leaf nodes visited per query, ``Interior'': Average number of interior nodes visited per query,
		``IPC'': Instructions per cycle, ``CR'': Cache reference, ``CMs'': Cache misses.
	}

	\label{table:perf-mix}%
\end{table}%

\subsection{Technique Analysis for Tree Construction\label{sec:tech:analysis}}
In \cref{algo:constr}, we mainly employed two techniques to improve the cache-efficiency:
1) building $\skheight$ levels at a time to save total data movements,
and 2) using sampling to determine splitters to save memory accesses.
To test these two techniques, we measure the time and the cache misses in tree construction
using three versions of \ourtree{s} with different levels of optimizations:
1) the final version with sampling and $\skheight=6$ (red bars),
2) 
constructing one level at a time using sampling, i.e., $\skheight=1$ (blue bars)
and
3) constructing one level at a time without sampling, i.e., $\skheight=1$ and finding exact median in parallel (yellow bars).
All benchmarks have $10^9$ points in 3 dimensions.
Results are presented in \cref{fig:technique}.
By comparing the red and blue bars, we observe that building multiple levels reduces running time by 2.91--4.31$\times$ and cache misses by 3.8$\times$ for both distributions.
The difference between the blue and yellow bars indicates
that sampling improves the time by about 1.86$\times$ and reduces cache misses by about 1.9$\times$.
The improvement in running time is consistent with the reduction of cache misses.
This verifies that the high performance of our construction algorithm indeed comes from the better cache efficiency enabled by the two techniques.


\subsection{Balancing Parameter Revisited\label{sec:balancingParameterRevisted}}


\hide{
	In this section we wish to answer the question that, to which extent the balancing parameter $\alpha$ influences the tree construction and the following query. Intuitively, a small $\alpha$ keeps the tree to be more balanced after consecutive updates, which can lead to faster query but the cost for reconstruction increases as well. Therefore, we wish to find a suitable value for $\alpha$ that achieves virtue in both worlds. However, to the best of our knowledge, there is no literature studies this question systematically.
}
One of the core ideas of the \ourtree{} is to relax the strong balancing criterion to achieve much better performance in construction and updates.
This may increase the tree height and affect the query performance to some extent, and the trade-off is controlled by the parameter $\alpha$.
In this section, we systematically study the choice of~$\alpha$ and the corresponding impact on the performance.

To do this, we create adversarial inputs such that belated rebalancing may result in a large tree height.
We generate skewed batch sequences and insert them into an empty tree by 1000 batches.
After each update, we perform queries to see how the imbalance affects the query performance. 
We test $\alpha$ in a full range from $0.01$ (almost always rebalance on updates) to $0.50$ (no rebalancing; siblings can be arbitrarily different).
We generated multiple distributions and show two of the most representative ones in the paper:

\begin{itemize}
	\item \type{1}: one instance of the 3D-V-1000M dataset.
	\item \type{2}: concatenation of one instance from 3D-U-100M and another one from 3D-V-900M.
\end{itemize}

\begin{figure}[t]
	\includegraphics[width=.48\textwidth]{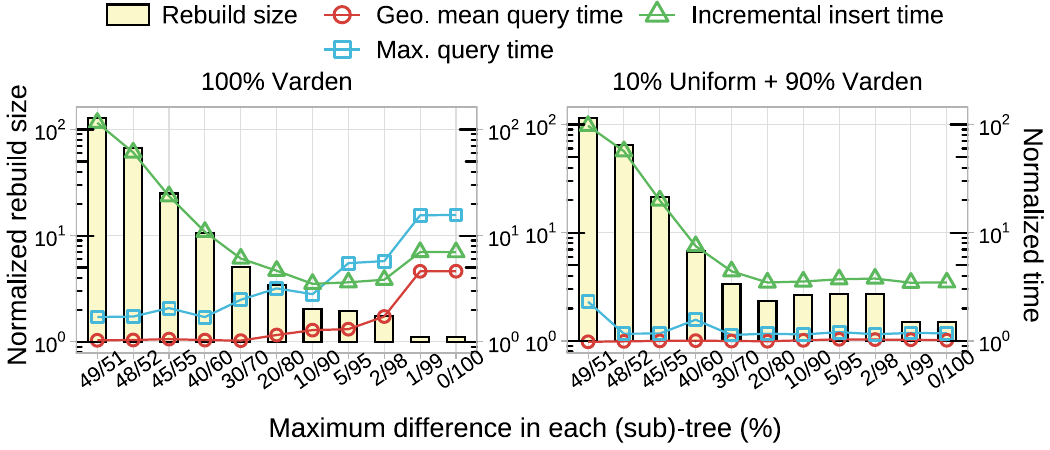}
	\caption{
		\textbf{Normalized rebuild size, update and query times with varying balancing parameter $\alpha$. Lower is better.}
		The dataset contains $10^9$ points in 3D, divided into 1000 batches, and incrementally inserted into an initially empty tree.
		The ``rebuild size'' (yellow bars, left axis) denotes the total number of tree nodes involved in reconstructions, 
		normalized to the final tree size.
		The total insertion time (green line, right axis) is normalized to that of building a tree directly from the input.
		After each of the 1000 batch insertions, we perform $1$-NN queries for 1M points from the \uniform{} distribution.
		The blue and red lines (right axis) show the geometric mean and maximum time among the 1000 queries respectively,
		normalized to the query time on a perfectly-balanced tree on the same set of points.
	}
	\label{fig:inbaRatio}
\end{figure}

\captionsetup{labelfont=bf, textfont={}}

We observe that \type{1} is adversarial since the \varden{} is generated by a random-walk plus random jump process.
By cutting the stream into 1000 batches, different dense areas (clusters) are added to the tree in-order, which will trigger frequent rebalancing for the \ourtree{s} (otherwise the tree quality can degenerate significantly).
\type{2}, as well as most of the other distributions are more resistant to large $\alpha$ values---the initial tree are generated on a \uniform{} distribution, providing a roughly even partition of the space---belated rebalancing does not affect the tree quality as much as \type{1}.

\cref{fig:inbaRatio} demonstrates the construction and $1$-NN time w.r.t. the balancing parameter $\alpha$.
The ``rebuild size'' (yellow bars) denotes the cumulative size of the subtrees that are reconstructed throughout all batch insertions.
This value is normalized to the final tree size, $10^9$.
The ``incremental update time'' is to construct a tree by inserting 1000 batches incrementally.
We normalize the construction time to that if we directly build a tree once using all the points.
After each batch insertion, we perform $1$-NN queries for a batch of another 1M points generated uniformly at random.
We normalize the query time to that on a perfectly balanced tree with the same set of points to illustrate the impact of imbalance.
Among all the 1000 batch queries, we record the maximum normalized query time by the blue rectangles in \cref{fig:inbaRatio},
which roughly represents the ``worst-case'' query time, and the geometric mean, which represents the ``average case''.
We use $x/y=(0.5-\alpha)/(0.5+\alpha)$ in the figure to indicate the degree of balance controlled by $\alpha$, which means that two sibling subtrees can differ by at most $x:y$.
We pick $1$-NN query here since the $1$-NN query performance is the most sensitive to how balance the \kdtree{} is, among all queries types tested in this paper.

For both input sequences, the overall trend for the incremental construction time decreases when less rebalance is required,
since rebalancing is triggered less frequently.
There is a slight rebound when the subtrees are excessively unbalanced (1/99 or worse)---the cost of traversing the tree in batch insertion also increases when the tree becomes skewed.
For queries, an unbalanced tree can significantly slow down the performance, as the searches need to go much deeper in the tree to touch the incident points.

Overall, the query performance is stable for a reasonably large range of $\alpha$.
The worst-case performance is negligibly affected all the way up to 30/70 ($\alpha=0.2$), and the average-case overhead is small until 10/90 ($\alpha=0.4$).
However, when we further relax the balancing criterion, then the performance can degenerate greatly on \type{1}---up to 4.48$\times$ slowdown on average for $\alpha=0.5$. The performance may still be reasonable on instances such as \type{2}.
To ensure better query performance in general, we choose $\alpha = 0.3$ (i.e., 20/80) as the default setting in \ourtree{} since it achieves a good tradeoff for the construction, update and query performance.

\ifconference{We also report the running time of batch updates with three specific values of $\alpha\in\{0.03,0.1,0.3\}$ with comparison to all baselines in the full paper.}\iffullversion{We also report the running time of batch updates with three specific values of $\alpha\in\{0.03,0.1,0.3\}$ in \cref{app:inbaratios}.}

\hide{
	Overall, in order to ensure the query efficiency of a \kdtree{} in the batch-dynamic scene, the value of $\alpha$ should be chosen carefully according to the I/O cost of the query. If the query requires extensive I/O operations, such as $100$-NN search, a larger value of $\alpha$ reduces the time for batch-update meanwhile brings small overhead for the query; otherwise, a smaller value of $\alpha$ is better.
}

\newcolumntype{J}[1]{>{\bf\raggedright\let\newline\\\arraybackslash\hspace{0pt}}m{#1}}
\hide{
  \begin{table*}[htbp]
    \centering \small
    \caption{Add caption}
    \begin{tabular}{J{1.6cm}L{1cm}rL{2cm}L{1.2cm}L{2cm}L{1.5cm}L{1.5cm}L{2cm}}
      \toprule
                                                    & \multicolumn{1}{l}{{base}} & \multicolumn{2}{c}{construction} & \multicolumn{2}{c}{Update}                                                                                    & \multicolumn{1}{c}{\multirow{2}[1]{*}{Property}} & \multicolumn{1}{c}{\multirow{2}[1]{*}{I/O opt.}} & \multicolumn{1}{c}{\multirow{2}[1]{*}{Implementation}}                                                                                              \\
                                                    & \multicolumn{1}{l}{tree}   & Par?                             & \multicolumn{1}{c}{Optimizations}                                                                             & Approach                                         & Parallel?                                        & \multicolumn{1}{c}{}                                   & \multicolumn{1}{c}{}                                       & \multicolumn{1}{c}{}          \\
      \midrule
      Bentley'75~\cite{bentley1975multidimensional} & single kd-tree             & seq                              & -                                                                                                             & Static                                           & -                                                & Perfectly balanced                                     & -                                                          & Implemented in CGAL           \\
      \midrule
      Willard'78~\cite{willard1978balanced}         & multiple k-d trees         & seq                              & -                                                                                                             & log method                                       & sequential, point update                         & Perfectly balanced                                     & -                                                          & Implemented later             \\
      \midrule
      kdb tree'80~\cite{robinson1981kdb}            & B-tree                     & Seq                              & -                                                                                                             & Overflow/ underflow                              & sequential, point update                         & no guanrantee                                          & multi-way tree nodes                                       & Implemented, unavailable      \\
      \midrule
      Overmas'83~\cite{overmars1983design}          & single kd-tree             & no                               & -                                                                                                             & partially rebuild                                & sequential, point update                         & relaxed (weight balanced)                              & -                                                          & -                             \\
      \midrule
      CGAL'97~\cite{cgal51}                  & single kd-tree             & Par                              & dimension choosing                                                                                            & fully rebuild                                    &                                                  & Perfectly balanced                                     & leaf wraping                                               & Implemented, available        \\
      \midrule
      Bkd tree'03~\cite{procopiuc2003bkd}           & multiple k-d trees         & Seq                              &                                                                                                               & log method                                       & Sequential, point update                         & perfectly balanced                                     & Optmized for construction and point update                 & Implemented in Julian's paper \\
      \midrule
      Arge et al.'03~\cite{agarwal2003cache}        & multiple k-d trees         & seq                              &                                                                                                               & log method                                       & Sequential, point update                         & perfectly balanced                                     & Optmized for construction\par and point update             &                               \\
      \midrule
      PODS’16~\cite{agarwal2016parallel}            & single kd-tree             & dist                             & Sampling, multi-level                                                                                         & static                                           & -                                                & relaxed (randomized)                                   & -                                                          & No                            \\
      \midrule
      ikd-tree'21~\cite{cai2021ikd}                 & single kd-tree             & seq                              & -                                                                                                             & Partial rebuild \par Mark tomb for deletion      & Sequential, point update                         & Relaxed [weight balanced]                              &                                                            & Implemented, unavailable      \\
      \midrule
      ParGeo'21~\cite{wang2022pargeo}               & multiple k-d trees         & par                              &                                                                                                               &                                                  & parallel, batch update                           & Perfectly balanced                                     & leaf wraping                                               & Implemented, available        \\
      \midrule
      zd-tree*'21~\cite{blelloch2022parallel}       & Oct-tree                   & par                              &                                                                                                               & BST rotations                                    & parallel, batch update                           & Weight-balanced BST                                    & leaf wraping                                               & Implemented, available        \\
      \midrule
      Ours'23                                       & single kd-tree             & par                              & \begin{itemize}[nosep,wide,leftmargin=*]\item dimension choosing \item Sampling\item multi-level\end{itemize} & partially rebuild                                & parallel, batch update                           & Relaxed [weight balance, randomized]                   & leaf wraping, optimized for construction and batch updates & Implemented, available        \\
      \bottomrule
    \end{tabular}%
    \label{tab:relatedwork}%
  \end{table*}%
}

\begin{table*}[htbp]
  \centering
 \setstretch{0.78}
  \small
  \vspace{-1em}
  \begin{tabular}{@{}J{0.8cm}@{}J{1.8cm}@{}C{0.7cm}L{3.2cm}L{2.1cm}L{3.2cm}L{3.8cm}}
    \toprule
         &                                                           &       &  \bf Layout \& Update                  & \multicolumn{1}{l}{\bf Balancing Criteria}             & \multicolumn{1}{@{}c@{}}{\bf Cache (I/O) Optimizations}                                                    & \bf Notes \& More Optimizations                                                                                         \\
    \midrule
    1975 & Bentley~\cite{bentley1975multidimensional}     & Seq.   & Single tree; No rebalancing              &    No balance                       & -                                                                          & \vspace{-1em}\begin{itemize}[topsep=-1em, itemsep=-0.2em, parsep=0em, wide, leftmargin=*, series=notes]
        \item Proposed \kdtree{}\vspace{-1em}
      \end{itemize} \\
\midrule
1978 & Bentley~\cite{bentley1979decomposable}          & Seq.   & Log-method; Tree merging         & Perfectly balanced                          & -                                                                          & \vspace{-1em}\begin{itemize}[resume*=notes]
        \item Proposed logarithmic method \vspace{-1em}
      \end{itemize}                                                    \\
    \midrule
    1980 & kdb-tree~\cite{robinson1981kdb}                      & Seq.   & B-tree; Overflow/underflow & Not shown                               &\vspace{-1em}\begin{itemize}[resume*=notes]
       \item  B-tree layout\vspace{-1em}
     \end{itemize}                                                                &-\\
    \midrule
    1983 & Overmars~\cite{overmars1983design}               & Seq.   & Single tree; Partial rebuild  & Relaxed \par(weight balanced)               & -                                                                          & -                                                                                                      \\
    \midrule
    2003 & Bkd-tree~\cite{procopiuc2003bkd}             & Seq.   & Log-method; Tree merging         & Perfectly balanced                          &
    \vspace{-1em}
    \begin{itemize}[resume*=notes]
      \item Cache opt. construction
      \item Cache opt. point update
      \vspace{-1em}
    \end{itemize}&-\\
    \midrule
    2003 & Agarwal et al.~\cite{agarwal2003cache}         & Seq.   & Log-method; Tree merging         & Perfectly balanced                          & \vspace{-1em}\begin{itemize}[resume*=notes]
                                                                                                                                   \item Cache opt. construction
                                                                                                                                   \item Cache opt. point update
                                                                                                                                   \item vEB layout
                                                                                                                                   \vspace{-1em}
                                                              \end{itemize} & -\\
    \midrule
    2016 & Agarwal et al.~\cite{agarwal2016parallel}       & Dist. & Single tree; Static (no update)            & Relaxed\par(randomized)                     & -                                                                          & \vspace{-1em}\begin{itemize}[resume*=notes]
        \item Sampling
        \item Multi-level construction\vspace{-1em}
      \end{itemize}                                                       \\
    \midrule
    2020 & CGAL~\cite{cgal51}                      & Par.  & Single tree; Full rebuild  \par{(sequential deletion)}    & Perfectly balanced                          &     \vspace{-1em}\begin{itemize}[resume*=notes]
       \item Leaf wrap\vspace{-1em}
     \end{itemize}                                                               & \vspace{-1em}
    \begin{itemize}[resume*=notes]
      \item ``\cgal'' tested in \cref{sec:exp}
      \item Implementation available
      \vspace{-1em}
    \end{itemize}\\
    \midrule
    2021 & ikd-tree~\cite{cai2021ikd}                      & Seq.   & Single tree; Partial rebuild  \par{(lazy deletion)}  & Relaxed \par(weight balanced)               &                      -                                                     & -\\
    \midrule
    2021 & ParGeo~\cite{wang2022pargeo}                & Par.   &   Log-method; Tree merging \par{(lazy deletion)}                & Perfectly balanced                          & \vspace{-1em}\begin{itemize}[resume*=notes]
       \item Leaf wrap
       \item vEB layout\vspace{-1em}
     \end{itemize}  &
     \vspace{-1em}\begin{itemize}[resume*=notes]
        \item ``\logtree{}'' tested in \cref{sec:exp}
        \item Implementation available\vspace{-1em}
      \end{itemize}
      \\
      \midrule
    2021 & ParGeo~\cite{wang2022pargeo}                & Par.   &   Single tree; Full rebuild                & Perfectly balanced                          & \vspace{-1em}\begin{itemize}[resume*=notes]
       \item Leaf wrap\vspace{-1em}
     \end{itemize}  &
     \vspace{-1em}\begin{itemize}[resume*=notes]
        \item ``\bhltree{}'' tested in \cref{sec:exp}
        \item Implementation available\vspace{-1em}
      \end{itemize}
      \\
    \midrule
\hide{      2021 & \zdtree{}~\cite{blelloch2022parallel}                  & Par.   & Quad/octree; Rotations      & Balanced BST                                &     \vspace{-1em}\begin{itemize}[resume*=notes]
       \item Leaf wrap\vspace{-1em}
     \end{itemize}
    &  \vspace{-1em}\begin{itemize}[resume*=notes]
        \item ``\zdtree{}'' tested in \cref{sec:exp}
        \item Open-source\vspace{-1em}
      \end{itemize}
   \\
    \midrule}
         & This paper (\ourtree{})                                           & Par.   & Single tree; Partial rebuild  & Relaxed \par{(weight balanced, randomized)}
         &\vspace{-1em}
         \begin{itemize}[resume*=notes]
           \item Leaf wrap
           \item Cache opt. construction
           \item Cache opt. batch update\vspace{-1em}
         \end{itemize}
         & \vspace{-1em}
         \begin{itemize}[resume*=notes]
           \item Sampling
           \item Multi-level construction
           \item Implementation available
           \vspace{-1em}
         \end{itemize}                                                                                              \\
    \bottomrule
  \end{tabular}%
  \caption{\textbf{Summary of related work.} ``Log-method'': logarithmic method (using $O(\log n)$ \kdtree{s}). ``Seq.'': sequential; ``Par.'': parallel; ``Dist.'': distributed. ``Cache opt.'': optimized for cache (or I/O) efficiency. Some of the designs are optimized for disk I/Os, but algorithmically the optimizations for disks or caches are almost identical. Therefore, we also denote both of them as ``cache-opt.'' in this table. 
  \vspace{-.8em}}
  \label{tab:related}%
\end{table*}%

\subsection{Parallel Scalability\label{sec:exp:scalability}}

\begin{figure}[t]
	\includegraphics[width=.42\textwidth]{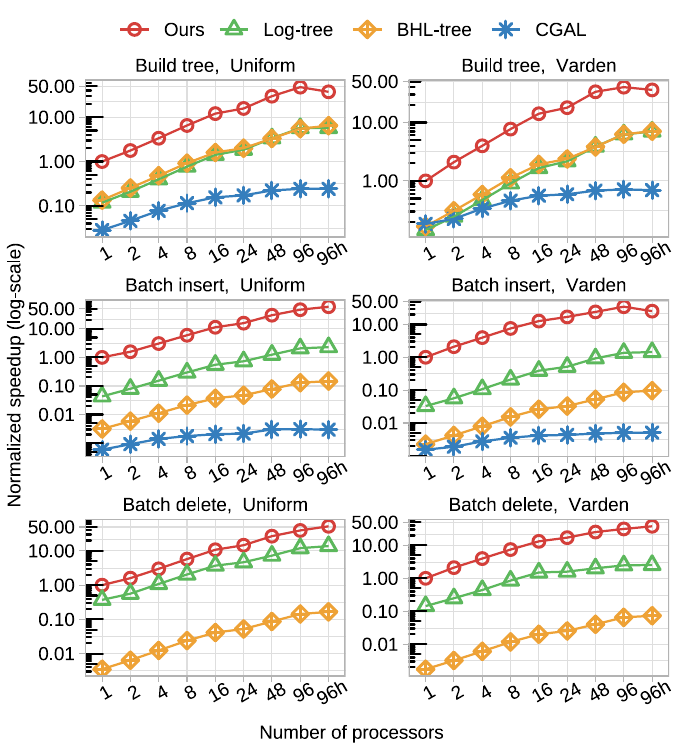}
	\caption{
		\textbf{Normalized parallel speedup of operations on \uniform and \varden for the \ourtree{} and other baselines on varying numbers of processors. Higher is better.}
		The curves show relative running time on different number of threads normalized to the \ourtree{} on one thread. 
		The benchmark contains 1000M points in 3D. 
		The ``batch insert'' inserts another 10M points from the same distribution into the tree, and the ``batch delete'' removes 10M points from the tree.
		``96h'': 96 cores with hyper-threads.
		There is no data for \cgal in batch delete since it deletes points sequentially.
	}
	\label{fig:scalability}
\end{figure}

We test the scalability for the tree construction and batch updates of \ourtree{} and other baselines on both 3D-U-1000M and 3D-V-1000M.
We normalize all running time to the \ourtree{} on one core, and show the scalability in \cref{fig:scalability}.
The \ourlib overall has very good scalability.
For \uniform{}, the \ourlib{} achieves 37.3$\times$ self-relative speedup in construction, 59.3$\times$ in batch insertion and 52.4$\times$ in batch deletion using 96 cores.
The numbers for \varden is 35.9$\times$ in tree construction, 25.3$\times$ in batch insertion, and 37.8$\times$ in batch deletion.
The speedup on \varden is lower in batch updates since the new points are unevenly distributed in the subtrees, resulting in more challenges in load balancing in constructing subtrees with various sizes.
Both \logtree{s} and \bhltree{s} have also shown good parallel speedup, while the performance difference is mainly due to the slow sequential (1 core) performance compared to the \ourtree{}.
The main reason that causes the advantage is the carefully-designed sieving algorithm introduced in \cref{sec:constr}, which is used in both construction and updates.
The only implementation that does not scale well is \cgal{}, which was also observed in previous work~\cite{blelloch2022parallel}.
\cgal{} only parallelizes the process of building two subtrees, but finds the splitters and partitions points into subtrees sequentially, limiting its parallelism.


\hide{
	One may notice that the complexity of tree construction algorithms for all baselines is the same,
	which is $O(n\log n)$, and all baselines have good scalability (except \cgal{}, see \cref{sec:exp:scalability}).
	However, as shown in \cref{table:summary}, the \ourlib{} outperforms other baselines for tree construction.
	We believe the reason is that the \ourlib{} uses the sampling and the sieving algorithm (see \cref{sec:constr-algo}) to optimize the cache utilization, making it the only cache-efficient implementation among all baselines.
	To verify this conclusion,
	we measure the tree construction in terms of time, number of cache misses, and runtime memory usage for all baselines regarding different input sizes.
	\cref{fig:cacheMemory} illustrates the result.

	As expected, the time needed for construction has a strong correlation with the cache misses and memory usage, and \ourlib{} has the best performance for all metrics.
	On benchmark from \texttt{Uniform}, compared with other baselines, the \ourlib{} uses 7.98--319$\times$ less time to finish the construction, meanwhile it incurs 5.88--61.4$\times$ fewer cache misses and 1.56--12.6$\times$ less runtime memory usage during the construction.
	The numbers on \texttt{Varden} are analogous to the \texttt{Uniform}.
}

\section{Related Work}\label{sec:related}
In this section, we review the literature of the \kdtree{}, with a summary of the most relevant ones in \cref{tab:related}.
The original algorithm proposed by Bentley et al.~\cite{bentley1975multidimensional,friedman1977algorithm} does not include a rebalancing scheme,
and assumes either static data, or inserting keys in a random order.
Since then, researchers have been developing rebalancing schemes for \kdtree{s}, mainly in the two categories: \emph{logarithmic method}, and \emph{partial rebuild}.
The logarithmic method was first proposed by Bentley in~\cite{bentley1979decomposable}, and has been followed up by later work, including optimizing the cache bounds~\cite{procopiuc2003bkd,agarwal2003cache} and parallelism~\cite{wang2022pargeo}.
However, as shown in \cref{table:summary}, maintaining $O(\log n)$ trees in the logarithmic method hampers the query efficiency significantly.
Another issue for this method is that insertions and deletions are asymmetric and need to be handled by different approaches,
which is more complicated.
An alternative idea is to maintain a single tree and partially rebuild the unbalanced subtrees, which
was proposed by Overmars~\cite{overmars1983design}.
Many papers followed up this idea, such as the KDB-tree~\cite{robinson1981kdb}, scapegoat $k$-d tree~\cite{galperin1993scapegoat}, ikd-tree~\cite{cai2021ikd}, and the divided $k$-d tree~\cite{van1991divided}.
Among them, only KDB-tree is cache-optimized.
None of them considered parallelism.

There have been many attempts to optimize the cache (or I/O) efficiency for \kdtree{s}.
Early work simply considers flattening the binary structure into a B-tree-like multiple-way tree~\cite{robinson1981kdb,procopiuc2003bkd}.
These papers are optimized for disk I/Os, and do not consider parallelism or batch updates.
\citet{procopiuc2003bkd} gave a (sequential) cache-efficient \kdtree{} construction algorithm, and dynamized the \kdtree{} using the logarithmic method.
\citet{agarwal2003cache} showed how to construct a cache-oblivious \kdtree{} using the vEB layout~\cite{bender2000cache}.
Motivated by \citet{procopiuc2003bkd}, \citet{wang2022pargeo} proposed parallel batch update algorithms on \kdtree{s}, and implemented them as the \logtree{} in the \pargeo{} library (called BDL-tree in their paper).
However, \citet{wang2022pargeo} did not show the cache complexity for their parallel construction or update algorithms.


There exist parallel \kdtree{} algorithms, but most of them are not cache-friendly or do not support a full interface.
Parallel construction for static \kdtree{s} has been well studied, mainly in two approaches.
The first approach~\cite{brown2014building,cao2020improved, yamasaki2018parallelizing} is to presort all points in all $D$ dimensions in parallel.
To compute the partition hyperplane, the median of the corresponding dimension is selected, and all elements are stably partitioned into two subtrees and recursively constructed.
The second approach~\cite{reif2022scalable, al2000parallel,hunt2006fast, shevtsov2007highly, choi2010parallel} finds the median as the splitter on the fly, and then constructs the sub-trees recursively in parallel.
Reif and Neumann's work~\cite{reif2022scalable} also proposed to support range-join using a \kdtree{} algorithm with the second approach.
Some of them also use sampling~\cite{al2000parallel,hunt2006fast} or constructing multiple levels~\cite{agarwal2016parallel,garanzha2011grid}.
\citet{agarwal2016parallel} showed a distributed algorithm for static \kdtree{s},
which also uses sampling and multi-level construction to optimize the number of rounds in the MPC model.
However, they did not show cache or span bounds, and have no implementations.
These approaches do not directly support updates. 
Although the \kdtree{} in \cgal{}~\cite{cgal51} supports updates, it simply rebuilds the tree after updates, which is inefficient.
The \pargeo{} library~\cite{wang2022pargeo} provides several \kdtree{} implementations, among which \logtree{s} and \bhltree{s} are generally the fastest.
The \bhltree{} is based on a single \kdtree{}, which fully rebuilds the tree on updates.
The \logtree{} uses the logarithmic method as discussed above to support parallel batch updates.
We compared to both of them in \cref{sec:exp}.
There also exist concurrent \kdtree{s}~\cite{chatterjee2018concurrent,ichnowski2020concurrent} that achieve linearizability and lock-freedom.
Our work focuses on batch-parallel setting,
which aims to support a batch of insertions or deletions with good work, span, and cache bounds.

There have been other data structures for multi-dimensional data such as R-trees (e.g.,~\cite{guttman1984r,kamel1992parallel,arge2002efficient, you2013parallel,prasad2015gpu}) and quad/octrees (e.g.,~\cite{blelloch2022parallel}).
Most of them do not support parallel updates.
There exist papers on parallel \rtree{s}, such as on GPUs~\cite{prasad2015gpu,you2013parallel} and on disks~\cite{kamel1992parallel,arge2002efficient}.
However, we are unaware of open-source in-memory implementations that support parallel construction and updates.
This is probably not surprising, given that the main use cases for \rtree{} are for external memory while the thread-level in-memory parallelism is a less relevant optimization.
\ifconference{For completeness, we compare the \ourlib{} with the sequential \rtree{} in Boost~\cite{schaling2011boost} in tree construction, \knn{} and range report in the full paper.}\iffullversion{For completeness, we compare the tree construction, \knn{} and range report time for \ourlib{} with the sequential \rtree{} in Boost~\cite{schaling2011boost} in \cref{app:rtree}.}
A recent paper~\cite{blelloch2022parallel} developed batch-parallel quad/octrees called \zdtree{s}.
However, through the correspondence with the authors, we confirmed that the released version has correctness issues in batch updates, so we cannot compare to their batch-update performance. \ifconference{We give a comparison to their construction time and \knn{} time in the full paper.}\iffullversion{We give a comparison to their construction time and \knn{} time in \cref{app:zdtree}.}



\hide{
	A static balanced \kdtree can be constructed within $O(n\log n)$ work by partition points evenly in each level~\cite{bentley1975multidimensional}. Since \kdtree is a space partition tree, there is no tree rotation operation like AVL tree~\cite{bayer1972symmetric} or Splay tree~\cite{sleator1985self} to retain a balanced tree after consecutive insertion or deletion. Commonly used scheme to keep a single tree balanced is the \bdita{partial rebuilding} strategy introduced in~\cite{overmars1982dynamic, overmars1983design} that rebuilds the imbalanced (sub)-tree after each update within $O(\log^2 n)$ amortized work, which forms the basis for the later K-D-B tree~\cite{robinson1981kdb}, scapegoat $k$-d tree~\cite{galperin1993scapegoat} and $ikd$-tree~\cite{cai2021ikd}. The divided $k$-d tree~\cite{van1991divided} reduces the update work to $O(\log n)$, while the tree structure was changed to two-level 2-3 trees.

	Bentley \textit{et al} introduced the \bdita{logarithmic method} to transform a static data structure to a dynamic one~\cite{bentley1979decomposable}.
	\yihan{looks like \cite{willard1978balanced} is earlier in 1978. We may want to make sure }
	The initial structure is replaced by logarithmic one, the I/O efficiency is improved while posing $O(\log n)$ overhead for the query operation. In this case, instead maintaining one \kdtree and dynamic rebalance it after each insertion, the Bkd-tree~\cite{procopiuc2003bkd} maintains logarithmic static K-D-B trees
	and update these structures periodically by reconstructing a specifically chosen subset of them. Given $N$ points, memory buffer $M$ and $B$ the number of points that fits in a disk block, the Bkd-tree perform an insertion in $O(\frac{1}{B}(\log_{M/B}\frac{N}{B})(\log_2{\frac{N}{M}}))$ amortized work.
	The Bkd-tree was further extended to be \bdita{cache-oblivious} by organize the logarithmic sub-trees through the van Emde Boas layout~\cite{agarwal2003cache}.

	For parallel tree construction, existing approaches either find the median and partition points in parallel as~\cite{reif2022scalable}, or presort the points and build the tree in parallel as~\cite{brown2014building,cao2020improved, yamasaki2018parallelizing}. When a perfect balanced \kdtree is required, parallel tree constructions focus on optimizing the evaluation of surface area heuristic (SAH) functions and memory throughput are proposed in~\cite{hunt2006fast, shevtsov2007highly, choi2010parallel}. In the massively parallel communication (MPC) model, \cite{al2000parallel} proposed that one can build the first $O(\log p)$ level of tree using $p$ non-shared memory machines in parallel, then construct the remaining levels locally. A similarly idea is presented in~\cite{agarwal2016parallel} as well.

	For those data structures supporting dynamic batch update and parallel rebalance, the BDL-tree~\cite{yesantharao2021parallel, wang2022pargeo} parallelizes the Bkd-tree and performs batch insertion/deletion within $O(B\log^2(n+B))$ amortized work and $O(\log (n+B)\log\log(n+B))$ depth, where $n$ is tree size before update and $B$ is the batch size.
	There is variant of \kdtree{}, called \zdtree~\cite{blelloch2022parallel}, which parallelizes the quad/oct-tree based on the Morton order. \zdtree supports tree construction in $O(n)$ work and $O(n^\epsilon)$ span for $\epsilon<1$. The batch update of size $m$ to \zdtree requires $O(m\log n/m)$ work with $O(m^\epsilon+\text{polylog(n)})$ span.
	The $ikd$-tree inserts batch points into the tree one by one and uses one main thread rebuild imbalanced sub-trees whose size is larger than a given threshold and another single thread rebuild the others, while their parallel reconstruction is not lock-free. The progressive kd-tree in~\cite{jo2017progressive} pipelines the query and insertion, the reconstruction is done in a parallel/interleaved task using a build queue, the most imbalanced tree in the queue is replaced by a new constructed tree.

	\myparagraph{Comparison to \cite{agarwal2016parallel}.}
	Agarwal et al.~\cite{agarwal2016parallel} also showed a parallel \kdtree construction based on sampling.
	Their work is purely theoretical and the main goal is to optimize the number of synchronization rounds in a distributed system.
	Meanwhile, it is unclear how to run multiple constructions simultaneously since their algorithm is bulk synchronized.
	Regarding the sampling phase, our result is also stronger.
	For instance, our oversampling rate is polylogarithmic for the tree height of $\log_2 n+O(1)$ in \cref{lem:tree-height}, while it can be polynomial in \cite{agarwal2016parallel}.
	While theoretically it makes no asymptotic difference, sorting these samples in practice can lead to significant overhead and hamper parallelism.

}

\section{Conclusion\label{sec:conclusion}}
We present \ourtree{}, a parallel \kdtree{} that has strong theoretical guarantees in work, span, and cache complexity for tree construction and batch update,
as well as high performance in practice.
Our main techniques include sampling, multi-level construction, the sieving algorithm, and the weight-balance scheme
to holistically optimize the work, span, cache-efficiency in both constructions and updates.
In this way, our approach relaxes the balancing criteria by a controllable manner, 
which allows for overall good performance considering construction, update, and various queries. 
In our experiments, the \ourtree{} significantly outperforms all the existing parallel \kdtree{} implementations on construction and updates, with competitive or better query performance. 


\section{Acknowledgements}
This work is supported by NSF grants CCF-2103483, TI-2346223 and IIS-2227669, NSF CAREER Awards CCF-2238358 and CCF-2339310, the UCR Regents Faculty Development Award, and the Google Research Scholar Program.

\bibliographystyle{ACM-Reference-Format}
\balance
\bibliography{bib/strings, bib/main}


\begin{thebibliography}{88}


\ifx \showCODEN    \undefined \def \showCODEN     #1{\unskip}     \fi
\ifx \showDOI      \undefined \def \showDOI       #1{#1}\fi
\ifx \showISBNx    \undefined \def \showISBNx     #1{\unskip}     \fi
\ifx \showISBNxiii \undefined \def \showISBNxiii  #1{\unskip}     \fi
\ifx \showISSN     \undefined \def \showISSN      #1{\unskip}     \fi
\ifx \showLCCN     \undefined \def \showLCCN      #1{\unskip}     \fi
\ifx \shownote     \undefined \def \shownote      #1{#1}          \fi
\ifx \showarticletitle \undefined \def \showarticletitle #1{#1}   \fi
\ifx \showURL      \undefined \def \showURL       {\relax}        \fi
\providecommand\bibfield[2]{#2}
\providecommand\bibinfo[2]{#2}
\providecommand\natexlab[1]{#1}
\providecommand\showeprint[2][]{arXiv:#2}

\bibitem[Agarwal et~al\mbox{.}(2016)]%
        {agarwal2016parallel}
\bibfield{author}{\bibinfo{person}{Pankaj Agarwal}, \bibinfo{person}{Kyle Fox}, \bibinfo{person}{Kamesh Munagala}, {and} \bibinfo{person}{Abhinandan Nath}.} \bibinfo{year}{2016}\natexlab{}.
\newblock \showarticletitle{Parallel algorithms for constructing range and nearest-neighbor searching data structures}. In \bibinfo{booktitle}{\emph{Principles of Database Systems (PODS)}}. \bibinfo{pages}{429--440}.
\newblock


\bibitem[Agarwal et~al\mbox{.}(2003)]%
        {agarwal2003cache}
\bibfield{author}{\bibinfo{person}{Pankaj~K Agarwal}, \bibinfo{person}{Lars Arge}, \bibinfo{person}{Andrew Danner}, {and} \bibinfo{person}{Bryan Holland-Minkley}.} \bibinfo{year}{2003}\natexlab{}.
\newblock \showarticletitle{Cache-oblivious data structures for orthogonal range searching}. In \bibinfo{booktitle}{\emph{Proceedings of the nineteenth annual symposium on Computational geometry}}. \bibinfo{pages}{237--245}.
\newblock


\bibitem[Aggarwal and Vitter(1988)]%
        {aggarwal1988input}
\bibfield{author}{\bibinfo{person}{Alok Aggarwal} {and} \bibinfo{person}{S Vitter, Jeffrey}.} \bibinfo{year}{1988}\natexlab{}.
\newblock \showarticletitle{The input/output complexity of sorting and related problems}.
\newblock \bibinfo{journal}{\emph{Commun. ACM}} \bibinfo{volume}{31}, \bibinfo{number}{9} (\bibinfo{year}{1988}), \bibinfo{pages}{1116--1127}.
\newblock


\bibitem[Al-Furajh et~al\mbox{.}(2000)]%
        {al2000parallel}
\bibfield{author}{\bibinfo{person}{I Al-Furajh}, \bibinfo{person}{Srinivas Aluru}, \bibinfo{person}{Sanjay Goil}, {and} \bibinfo{person}{Sanjay Ranka}.} \bibinfo{year}{2000}\natexlab{}.
\newblock \showarticletitle{Parallel construction of multidimensional binary search trees}.
\newblock \bibinfo{journal}{\emph{IEEE Transactions on Parallel and Distributed Systems}} \bibinfo{volume}{11}, \bibinfo{number}{2} (\bibinfo{year}{2000}), \bibinfo{pages}{136--148}.
\newblock


\bibitem[Anderson et~al\mbox{.}(2022)]%
        {anderson2022problem}
\bibfield{author}{\bibinfo{person}{Daniel Anderson}, \bibinfo{person}{Guy~E Blelloch}, \bibinfo{person}{Laxman Dhulipala}, \bibinfo{person}{Magdalen Dobson}, {and} \bibinfo{person}{Yihan Sun}.} \bibinfo{year}{2022}\natexlab{}.
\newblock \showarticletitle{The problem-based benchmark suite (PBBS), V2}. In \bibinfo{booktitle}{\emph{{ACM} Symposium on Principles and Practice of Parallel Programming (PPOPP)}}. \bibinfo{pages}{445--447}.
\newblock


\bibitem[Andersson(1989)]%
        {andersson1989improving}
\bibfield{author}{\bibinfo{person}{Arne Andersson}.} \bibinfo{year}{1989}\natexlab{}.
\newblock \showarticletitle{Improving partial rebuilding by using simple balance criteria}. In \bibinfo{booktitle}{\emph{Workshop on Algorithms and Data Structures (WADS)}}. Springer, \bibinfo{pages}{393--402}.
\newblock


\bibitem[Arge et~al\mbox{.}(2004)]%
        {arge2004cache}
\bibfield{author}{\bibinfo{person}{Lars Arge}, \bibinfo{person}{Gerth~St{\o}lting Brodal}, {and} \bibinfo{person}{Rolf Fagerberg}.} \bibinfo{year}{2004}\natexlab{}.
\newblock \showarticletitle{Cache-Oblivious Data Structures}.
\newblock \bibinfo{journal}{\emph{Handbook of Data Structures and Applications}}  \bibinfo{volume}{27} (\bibinfo{year}{2004}).
\newblock


\bibitem[Arge et~al\mbox{.}(2002)]%
        {arge2002efficient}
\bibfield{author}{\bibinfo{person}{Lars Arge}, \bibinfo{person}{Klaus~H Hinrichs}, \bibinfo{person}{Jan Vahrenhold}, {and} \bibinfo{person}{Jeffrey~Scott Vitter}.} \bibinfo{year}{2002}\natexlab{}.
\newblock \showarticletitle{Efficient bulk operations on dynamic R-trees}.
\newblock \bibinfo{journal}{\emph{Algorithmica}}  \bibinfo{volume}{33} (\bibinfo{year}{2002}), \bibinfo{pages}{104--128}.
\newblock


\bibitem[Arora et~al\mbox{.}(2001)]%
        {arora2001thread}
\bibfield{author}{\bibinfo{person}{Nimar~S Arora}, \bibinfo{person}{Robert~D Blumofe}, {and} \bibinfo{person}{C~Greg Plaxton}.} \bibinfo{year}{2001}\natexlab{}.
\newblock \showarticletitle{Thread scheduling for multiprogrammed multiprocessors}.
\newblock \bibinfo{journal}{\emph{Theory of Computing Systems (TOCS)}} \bibinfo{volume}{34}, \bibinfo{number}{2} (\bibinfo{year}{2001}), \bibinfo{pages}{115--144}.
\newblock


\bibitem[Bender et~al\mbox{.}(2000)]%
        {bender2000cache}
\bibfield{author}{\bibinfo{person}{Michael~A Bender}, \bibinfo{person}{Erik~D Demaine}, {and} \bibinfo{person}{Martin Farach-Colton}.} \bibinfo{year}{2000}\natexlab{}.
\newblock \showarticletitle{Cache-oblivious B-trees}. In \bibinfo{booktitle}{\emph{focs}}. IEEE, \bibinfo{pages}{399--409}.
\newblock


\bibitem[Bentley(1975)]%
        {bentley1975multidimensional}
\bibfield{author}{\bibinfo{person}{Jon~Louis Bentley}.} \bibinfo{year}{1975}\natexlab{}.
\newblock \showarticletitle{Multidimensional binary search trees used for associative searching}.
\newblock \bibinfo{journal}{\emph{Commun. {ACM}}} \bibinfo{volume}{18}, \bibinfo{number}{9} (\bibinfo{year}{1975}), \bibinfo{pages}{509--517}.
\newblock


\bibitem[Bentley(1979)]%
        {bentley1979decomposable}
\bibfield{author}{\bibinfo{person}{Jon~Louis Bentley}.} \bibinfo{year}{1979}\natexlab{}.
\newblock \showarticletitle{Decomposable searching problems}.
\newblock \bibinfo{journal}{\emph{Inform. Process. Lett.}} \bibinfo{volume}{8}, \bibinfo{number}{5} (\bibinfo{year}{1979}), \bibinfo{pages}{244--251}.
\newblock


\bibitem[Blelloch(1989)]%
        {Blelloch89}
\bibfield{author}{\bibinfo{person}{Guy~E. Blelloch}.} \bibinfo{year}{1989}\natexlab{}.
\newblock \showarticletitle{Scans as Primitive Parallel Operations}.
\newblock \bibinfo{journal}{\emph{IEEE Trans. on Comput.}} \bibinfo{volume}{38}, \bibinfo{number}{11} (\bibinfo{year}{1989}).
\newblock


\bibitem[Blelloch et~al\mbox{.}(2020a)]%
        {blelloch2020parlaylib}
\bibfield{author}{\bibinfo{person}{Guy~E. Blelloch}, \bibinfo{person}{Daniel Anderson}, {and} \bibinfo{person}{Laxman Dhulipala}.} \bibinfo{year}{2020}\natexlab{a}.
\newblock \showarticletitle{ParlayLib --- a toolkit for parallel algorithms on shared-memory multicore machines}. In \bibinfo{booktitle}{\emph{{ACM} Symposium on Parallelism in Algorithms and Architectures (SPAA)}}. \bibinfo{pages}{507--509}.
\newblock


\bibitem[Blelloch and Dobson(2022)]%
        {blelloch2022parallel}
\bibfield{author}{\bibinfo{person}{Guy~E Blelloch} {and} \bibinfo{person}{Magdalen Dobson}.} \bibinfo{year}{2022}\natexlab{}.
\newblock \showarticletitle{Parallel Nearest Neighbors in Low Dimensions with Batch Updates}. In \bibinfo{booktitle}{\emph{Algorithm Engineering and Experiments (ALENEX)}}. SIAM, \bibinfo{pages}{195--208}.
\newblock


\bibitem[Blelloch et~al\mbox{.}(2016)]%
        {blelloch2016just}
\bibfield{author}{\bibinfo{person}{Guy~E. Blelloch}, \bibinfo{person}{Daniel Ferizovic}, {and} \bibinfo{person}{Yihan Sun}.} \bibinfo{year}{2016}\natexlab{}.
\newblock \showarticletitle{Just Join for Parallel Ordered Sets}. In \bibinfo{booktitle}{\emph{{ACM} Symposium on Parallelism in Algorithms and Architectures (SPAA)}}.
\newblock


\bibitem[Blelloch et~al\mbox{.}(2020b)]%
        {blelloch2020optimal}
\bibfield{author}{\bibinfo{person}{Guy~E. Blelloch}, \bibinfo{person}{Jeremy~T. Fineman}, \bibinfo{person}{Yan Gu}, {and} \bibinfo{person}{Yihan Sun}.} \bibinfo{year}{2020}\natexlab{b}.
\newblock \showarticletitle{Optimal parallel algorithms in the binary-forking model}. In \bibinfo{booktitle}{\emph{{ACM} Symposium on Parallelism in Algorithms and Architectures (SPAA)}}. \bibinfo{pages}{89--102}.
\newblock


\bibitem[Blelloch et~al\mbox{.}(2010)]%
        {blelloch2010low}
\bibfield{author}{\bibinfo{person}{Guy~E. Blelloch}, \bibinfo{person}{Phillip~B. Gibbons}, {and} \bibinfo{person}{Harsha~Vardhan Simhadri}.} \bibinfo{year}{2010}\natexlab{}.
\newblock \showarticletitle{Low depth cache-oblivious algorithms}. In \bibinfo{booktitle}{\emph{{ACM} Symposium on Parallelism in Algorithms and Architectures (SPAA)}}.
\newblock


\bibitem[Blelloch and Gu(2020)]%
        {BG2020}
\bibfield{author}{\bibinfo{person}{Guy~E. Blelloch} {and} \bibinfo{person}{Yan Gu}.} \bibinfo{year}{2020}\natexlab{}.
\newblock \showarticletitle{Improved Parallel Cache-Oblivious Algorithms for Dynamic Programming}. In \bibinfo{booktitle}{\emph{{SIAM} Symposium on Algorithmic Principles of Computer Systems (APOCS)}}.
\newblock


\bibitem[Blelloch et~al\mbox{.}(2018)]%
        {blelloch2018geometry}
\bibfield{author}{\bibinfo{person}{Guy~E. Blelloch}, \bibinfo{person}{Yan Gu}, \bibinfo{person}{Julian Shun}, {and} \bibinfo{person}{Yihan Sun}.} \bibinfo{year}{2018}\natexlab{}.
\newblock \showarticletitle{Parallel Write-Efficient Algorithms and Data Structures for Computational Geometry}. In \bibinfo{booktitle}{\emph{{ACM} Symposium on Parallelism in Algorithms and Architectures (SPAA)}}.
\newblock


\bibitem[Blonder et~al\mbox{.}(2018)]%
        {blonder2018new}
\bibfield{author}{\bibinfo{person}{Benjamin Blonder}, \bibinfo{person}{Cecina~Babich Morrow}, \bibinfo{person}{Brian Maitner}, \bibinfo{person}{David~J Harris}, \bibinfo{person}{Christine Lamanna}, \bibinfo{person}{Cyrille Violle}, \bibinfo{person}{Brian~J Enquist}, {and} \bibinfo{person}{Andrew~J Kerkhoff}.} \bibinfo{year}{2018}\natexlab{}.
\newblock \showarticletitle{New approaches for delineating n-dimensional hypervolumes}.
\newblock \bibinfo{journal}{\emph{Methods in Ecology and Evolution}} \bibinfo{volume}{9}, \bibinfo{number}{2} (\bibinfo{year}{2018}), \bibinfo{pages}{305--319}.
\newblock


\bibitem[Blumofe and Leiserson(1998)]%
        {BL98}
\bibfield{author}{\bibinfo{person}{Robert~D. Blumofe} {and} \bibinfo{person}{Charles~E. Leiserson}.} \bibinfo{year}{1998}\natexlab{}.
\newblock \showarticletitle{Space-Efficient Scheduling of Multithreaded Computations}.
\newblock \bibinfo{journal}{\emph{{SIAM} J. on Computing}} \bibinfo{volume}{27}, \bibinfo{number}{1} (\bibinfo{year}{1998}).
\newblock


\bibitem[B{\"o}hm et~al\mbox{.}(2001)]%
        {bohm2001searching}
\bibfield{author}{\bibinfo{person}{Christian B{\"o}hm}, \bibinfo{person}{Stefan Berchtold}, {and} \bibinfo{person}{Daniel~A Keim}.} \bibinfo{year}{2001}\natexlab{}.
\newblock \showarticletitle{Searching in high-dimensional spaces: Index structures for improving the performance of multimedia databases}.
\newblock \bibinfo{journal}{\emph{ACM Computing Surveys (CSUR)}} \bibinfo{volume}{33}, \bibinfo{number}{3} (\bibinfo{year}{2001}), \bibinfo{pages}{322--373}.
\newblock


\bibitem[Brown(2014)]%
        {brown2014building}
\bibfield{author}{\bibinfo{person}{Russell~A Brown}.} \bibinfo{year}{2014}\natexlab{}.
\newblock \showarticletitle{Building a balanced kd tree in o (kn log n) time}.
\newblock \bibinfo{journal}{\emph{arXiv preprint arXiv:1410.5420}} (\bibinfo{year}{2014}).
\newblock


\bibitem[Cai et~al\mbox{.}(2021)]%
        {cai2021ikd}
\bibfield{author}{\bibinfo{person}{Yixi Cai}, \bibinfo{person}{Wei Xu}, {and} \bibinfo{person}{Fu Zhang}.} \bibinfo{year}{2021}\natexlab{}.
\newblock \showarticletitle{ikd-tree: An incremental kd tree for robotic applications}.
\newblock \bibinfo{journal}{\emph{arXiv preprint arXiv:2102.10808}} (\bibinfo{year}{2021}).
\newblock


\bibitem[Cao et~al\mbox{.}(2020)]%
        {cao2020improved}
\bibfield{author}{\bibinfo{person}{Yu Cao}, \bibinfo{person}{Xiaojiang Zhang}, \bibinfo{person}{Boheng Duan}, \bibinfo{person}{Wenjing Zhao}, {and} \bibinfo{person}{Huizan Wang}.} \bibinfo{year}{2020}\natexlab{}.
\newblock \showarticletitle{An improved method to build the KD tree based on presorted results}. In \bibinfo{booktitle}{\emph{International Conference on Software Engineering and Service Science (ICSESS)}}. IEEE, \bibinfo{pages}{71--75}.
\newblock


\bibitem[Chatterjee et~al\mbox{.}(2018)]%
        {chatterjee2018concurrent}
\bibfield{author}{\bibinfo{person}{Bapi Chatterjee}, \bibinfo{person}{Ivan Walulya}, {and} \bibinfo{person}{Philippas Tsigas}.} \bibinfo{year}{2018}\natexlab{}.
\newblock \showarticletitle{Concurrent linearizable nearest neighbour search in lock free-kd-tree}. In \bibinfo{booktitle}{\emph{Proceedings of the 19th International Conference on Distributed Computing and Networking}}. \bibinfo{pages}{1--10}.
\newblock


\bibitem[Chen et~al\mbox{.}(2016)]%
        {chen2016gene}
\bibfield{author}{\bibinfo{person}{Yifei Chen}, \bibinfo{person}{Yi Li}, \bibinfo{person}{Rajiv Narayan}, \bibinfo{person}{Aravind Subramanian}, {and} \bibinfo{person}{Xiaohui Xie}.} \bibinfo{year}{2016}\natexlab{}.
\newblock \showarticletitle{Gene expression inference with deep learning}.
\newblock \bibinfo{journal}{\emph{Bioinformatics}} \bibinfo{volume}{32}, \bibinfo{number}{12} (\bibinfo{year}{2016}), \bibinfo{pages}{1832--1839}.
\newblock


\bibitem[Choi et~al\mbox{.}(2010)]%
        {choi2010parallel}
\bibfield{author}{\bibinfo{person}{Byn Choi}, \bibinfo{person}{Rakesh Komuravelli}, \bibinfo{person}{Victor Lu}, \bibinfo{person}{Hyojin Sung}, \bibinfo{person}{Robert~L Bocchino~Jr}, \bibinfo{person}{Sarita~V Adve}, {and} \bibinfo{person}{John~C Hart}.} \bibinfo{year}{2010}\natexlab{}.
\newblock \showarticletitle{Parallel SAH kD tree construction}. In \bibinfo{booktitle}{\emph{High performance graphics}}. Citeseer, \bibinfo{pages}{77--86}.
\newblock


\bibitem[Deng et~al\mbox{.}(2016)]%
        {deng2016efficient}
\bibfield{author}{\bibinfo{person}{Zhenyun Deng}, \bibinfo{person}{Xiaoshu Zhu}, \bibinfo{person}{Debo Cheng}, \bibinfo{person}{Ming Zong}, {and} \bibinfo{person}{Shichao Zhang}.} \bibinfo{year}{2016}\natexlab{}.
\newblock \showarticletitle{Efficient kNN classification algorithm for big data}.
\newblock \bibinfo{journal}{\emph{Neurocomputing}}  \bibinfo{volume}{195} (\bibinfo{year}{2016}), \bibinfo{pages}{143--148}.
\newblock


\bibitem[Dhulipala et~al\mbox{.}(2022)]%
        {dhulipala2022pac}
\bibfield{author}{\bibinfo{person}{Laxman Dhulipala}, \bibinfo{person}{Guy~E. Blelloch}, \bibinfo{person}{Yan Gu}, {and} \bibinfo{person}{Yihan Sun}.} \bibinfo{year}{2022}\natexlab{}.
\newblock \showarticletitle{{PaC-trees}: Supporting Parallel and Compressed Purely-Functional Collections}. In \bibinfo{booktitle}{\emph{ACM Conference on Programming Language Design and Implementation (PLDI)}}.
\newblock


\bibitem[Dong et~al\mbox{.}(2024)]%
        {dong2024parallel}
\bibfield{author}{\bibinfo{person}{Xiaojun Dong}, \bibinfo{person}{Laxman Dhulipala}, \bibinfo{person}{Yan Gu}, {and} \bibinfo{person}{Yihan Sun}.} \bibinfo{year}{2024}\natexlab{}.
\newblock \showarticletitle{Parallel Integer Sort: Theory and Practice}. In \bibinfo{booktitle}{\emph{{ACM} Symposium on Principles and Practice of Parallel Programming (PPOPP)}}.
\newblock


\bibitem[Dong et~al\mbox{.}(2023)]%
        {dong2023high}
\bibfield{author}{\bibinfo{person}{Xiaojun Dong}, \bibinfo{person}{Yunshu Wu}, \bibinfo{person}{Zhongqi Wang}, \bibinfo{person}{Laxman Dhulipala}, \bibinfo{person}{Yan Gu}, {and} \bibinfo{person}{Yihan Sun}.} \bibinfo{year}{2023}\natexlab{}.
\newblock \showarticletitle{High-Performance and Flexible Parallel Algorithms for Semisort and Related Problems}. In \bibinfo{booktitle}{\emph{{ACM} Symposium on Parallelism in Algorithms and Architectures (SPAA)}}.
\newblock


\bibitem[Fonollosa et~al\mbox{.}(2015)]%
        {fonollosa2015reservoir}
\bibfield{author}{\bibinfo{person}{Jordi Fonollosa}, \bibinfo{person}{Sadique Sheik}, \bibinfo{person}{Ram{\'o}n Huerta}, {and} \bibinfo{person}{Santiago Marco}.} \bibinfo{year}{2015}\natexlab{}.
\newblock \showarticletitle{Reservoir computing compensates slow response of chemosensor arrays exposed to fast varying gas concentrations in continuous monitoring}.
\newblock \bibinfo{journal}{\emph{Sensors and Actuators B: Chemical}}  \bibinfo{volume}{215} (\bibinfo{year}{2015}), \bibinfo{pages}{618--629}.
\newblock


\bibitem[Friedman et~al\mbox{.}(1977)]%
        {friedman1977algorithm}
\bibfield{author}{\bibinfo{person}{Jerome~H Friedman}, \bibinfo{person}{Jon~Louis Bentley}, {and} \bibinfo{person}{Raphael~Ari Finkel}.} \bibinfo{year}{1977}\natexlab{}.
\newblock \showarticletitle{An algorithm for finding best matches in logarithmic expected time}.
\newblock \bibinfo{journal}{\emph{ACM Transactions on Mathematical Software (TOMS)}} \bibinfo{volume}{3}, \bibinfo{number}{3} (\bibinfo{year}{1977}), \bibinfo{pages}{209--226}.
\newblock


\bibitem[Frigo et~al\mbox{.}(1999)]%
        {Frigo99}
\bibfield{author}{\bibinfo{person}{Matteo Frigo}, \bibinfo{person}{Charles~E. Leiserson}, \bibinfo{person}{Harald Prokop}, {and} \bibinfo{person}{Sridhar Ramachandran}.} \bibinfo{year}{1999}\natexlab{}.
\newblock \showarticletitle{Cache-Oblivious Algorithms}. In \bibinfo{booktitle}{\emph{{IEEE} Symposium on Foundations of Computer Science (FOCS)}}.
\newblock


\bibitem[Galperin and Rivest(1993)]%
        {galperin1993scapegoat}
\bibfield{author}{\bibinfo{person}{Igal Galperin} {and} \bibinfo{person}{Ronald Rivest}.} \bibinfo{year}{1993}\natexlab{}.
\newblock \showarticletitle{Scapegoat Trees.}. In \bibinfo{booktitle}{\emph{{ACM-SIAM} Symposium on Discrete Algorithms (SODA)}}, Vol.~\bibinfo{volume}{93}. \bibinfo{pages}{165--174}.
\newblock


\bibitem[Gan and Tao(2017)]%
        {gan2017hardness}
\bibfield{author}{\bibinfo{person}{Junhao Gan} {and} \bibinfo{person}{Yufei Tao}.} \bibinfo{year}{2017}\natexlab{}.
\newblock \showarticletitle{On the hardness and approximation of Euclidean DBSCAN}.
\newblock \bibinfo{journal}{\emph{ACM Transactions on Database Systems (TODS)}} \bibinfo{volume}{42}, \bibinfo{number}{3} (\bibinfo{year}{2017}), \bibinfo{pages}{1--45}.
\newblock


\bibitem[Garanzha et~al\mbox{.}(2011)]%
        {garanzha2011grid}
\bibfield{author}{\bibinfo{person}{Kirill Garanzha}, \bibinfo{person}{Simon Premo{\v{z}}e}, \bibinfo{person}{Alexander Bely}, {and} \bibinfo{person}{Vladimir Galaktionov}.} \bibinfo{year}{2011}\natexlab{}.
\newblock \showarticletitle{Grid-based SAH BVH construction on a GPU}.
\newblock \bibinfo{journal}{\emph{The Visual Computer}}  \bibinfo{volume}{27} (\bibinfo{year}{2011}), \bibinfo{pages}{697--706}.
\newblock


\bibitem[Graefe(1993)]%
        {Graefe93}
\bibfield{author}{\bibinfo{person}{Goetz Graefe}.} \bibinfo{year}{1993}\natexlab{}.
\newblock \showarticletitle{Query Evaluation Techniques for Large Databases}.
\newblock \bibinfo{journal}{\emph{ACM Comput. Surv.}} \bibinfo{volume}{25}, \bibinfo{number}{2} (\bibinfo{year}{1993}), \bibinfo{pages}{73--170}.
\newblock


\bibitem[Gu et~al\mbox{.}(2023)]%
        {gu2023parallel}
\bibfield{author}{\bibinfo{person}{Yan Gu}, \bibinfo{person}{Ziyang Men}, \bibinfo{person}{Zheqi Shen}, \bibinfo{person}{Yihan Sun}, {and} \bibinfo{person}{Zijin Wan}.} \bibinfo{year}{2023}\natexlab{}.
\newblock \showarticletitle{Parallel Longest Increasing Subsequence and van Emde Boas Trees}. In \bibinfo{booktitle}{\emph{{ACM} Symposium on Parallelism in Algorithms and Architectures (SPAA)}}.
\newblock


\bibitem[Gu et~al\mbox{.}(2022)]%
        {gu2022analysis}
\bibfield{author}{\bibinfo{person}{Yan Gu}, \bibinfo{person}{Zachary Napier}, {and} \bibinfo{person}{Yihan Sun}.} \bibinfo{year}{2022}\natexlab{}.
\newblock \showarticletitle{Analysis of Work-Stealing and Parallel Cache Complexity}. In \bibinfo{booktitle}{\emph{{SIAM} Symposium on Algorithmic Principles of Computer Systems (APOCS)}}. SIAM, \bibinfo{pages}{46--60}.
\newblock


\bibitem[Guo et~al\mbox{.}(2013)]%
        {guo2013rotational}
\bibfield{author}{\bibinfo{person}{Yulan Guo}, \bibinfo{person}{Ferdous Sohel}, \bibinfo{person}{Mohammed Bennamoun}, \bibinfo{person}{Min Lu}, {and} \bibinfo{person}{Jianwei Wan}.} \bibinfo{year}{2013}\natexlab{}.
\newblock \showarticletitle{Rotational projection statistics for 3D local surface description and object recognition}.
\newblock \bibinfo{journal}{\emph{International journal of computer vision}}  \bibinfo{volume}{105} (\bibinfo{year}{2013}), \bibinfo{pages}{63--86}.
\newblock


\bibitem[G{\"u}ting(1994)]%
        {guting1994introduction}
\bibfield{author}{\bibinfo{person}{Ralf~Hartmut G{\"u}ting}.} \bibinfo{year}{1994}\natexlab{}.
\newblock \showarticletitle{An introduction to spatial database systems}.
\newblock \bibinfo{journal}{\emph{the VLDB Journal}}  \bibinfo{volume}{3} (\bibinfo{year}{1994}), \bibinfo{pages}{357--399}.
\newblock


\bibitem[Guttman(1984)]%
        {guttman1984r}
\bibfield{author}{\bibinfo{person}{Antonin Guttman}.} \bibinfo{year}{1984}\natexlab{}.
\newblock \showarticletitle{R-trees: A dynamic index structure for spatial searching}. In \bibinfo{booktitle}{\emph{ACM SIGMOD International Conference on Management of Data (SIGMOD)}}. \bibinfo{pages}{47--57}.
\newblock


\bibitem[Haklay and Weber(2008)]%
        {haklay2008openstreetmap}
\bibfield{author}{\bibinfo{person}{Mordechai Haklay} {and} \bibinfo{person}{Patrick Weber}.} \bibinfo{year}{2008}\natexlab{}.
\newblock \showarticletitle{Openstreetmap: User-generated street maps}.
\newblock \bibinfo{journal}{\emph{IEEE Pervasive computing}} \bibinfo{volume}{7}, \bibinfo{number}{4} (\bibinfo{year}{2008}), \bibinfo{pages}{12--18}.
\newblock


\bibitem[Hebrail and Berard(2012)]%
        {household}
\bibfield{author}{\bibinfo{person}{Georges Hebrail} {and} \bibinfo{person}{Alice Berard}.} \bibinfo{year}{2012}\natexlab{}.
\newblock \bibinfo{title}{{Individual household electric power consumption}}.
\newblock \bibinfo{howpublished}{UCI Machine Learning Repository}.
\newblock
\newblock
\shownote{{DOI}: https://doi.org/10.24432/C58K54}.


\bibitem[Huerta et~al\mbox{.}(2016)]%
        {huerta2016online}
\bibfield{author}{\bibinfo{person}{Ramon Huerta}, \bibinfo{person}{Thiago Mosqueiro}, \bibinfo{person}{Jordi Fonollosa}, \bibinfo{person}{Nikolai~F Rulkov}, {and} \bibinfo{person}{Irene Rodriguez-Lujan}.} \bibinfo{year}{2016}\natexlab{}.
\newblock \showarticletitle{Online decorrelation of humidity and temperature in chemical sensors for continuous monitoring}.
\newblock \bibinfo{journal}{\emph{Chemometrics and Intelligent Laboratory Systems}}  \bibinfo{volume}{157} (\bibinfo{year}{2016}), \bibinfo{pages}{169--176}.
\newblock


\bibitem[Hunt et~al\mbox{.}(2006)]%
        {hunt2006fast}
\bibfield{author}{\bibinfo{person}{Warren Hunt}, \bibinfo{person}{William~R Mark}, {and} \bibinfo{person}{Gordon Stoll}.} \bibinfo{year}{2006}\natexlab{}.
\newblock \showarticletitle{Fast kd-tree construction with an adaptive error-bounded heuristic}. In \bibinfo{booktitle}{\emph{IEEE Symposium on Interactive Ray Tracing}}. IEEE, \bibinfo{pages}{81--88}.
\newblock


\bibitem[Ichnowski and Alterovitz(2020)]%
        {ichnowski2020concurrent}
\bibfield{author}{\bibinfo{person}{Jeffrey Ichnowski} {and} \bibinfo{person}{Ron Alterovitz}.} \bibinfo{year}{2020}\natexlab{}.
\newblock \showarticletitle{Concurrent nearest-neighbor searching for parallel sampling-based motion planning in SO (3), SE (3), and euclidean spaces}. In \bibinfo{booktitle}{\emph{Algorithmic Foundations of Robotics XIII: Proceedings of the 13th Workshop on the Algorithmic Foundations of Robotics 13}}. Springer, \bibinfo{pages}{69--85}.
\newblock


\bibitem[{Intel Corporation}(2024)]%
        {IntelVTune2024}
\bibfield{author}{\bibinfo{person}{{Intel Corporation}}.} \bibinfo{year}{2024}\natexlab{}.
\newblock \bibinfo{title}{VTune Profiler}.
\newblock
\newblock
\urldef\tempurl%
\url{https://www.intel.com/content/www/us/en/developer/tools/oneapi/vtune-profiler.html}
\showURL{%
\tempurl}


\bibitem[Intel Threading Building Blocks({[n.\,d.]})]%
        {TBB}
Intel Threading Building Blocks \bibinfo{year}{[n.\,d.]}\natexlab{}.
\newblock \bibinfo{title}{Intel Threading Building Blocks ({TBB})}.
\newblock \bibinfo{howpublished}{\url{https://www.threadingbuildingblocks.org}}.
\newblock


\bibitem[Jo et~al\mbox{.}(2017)]%
        {jo2017progressive}
\bibfield{author}{\bibinfo{person}{Jaemin Jo}, \bibinfo{person}{Jinwook Seo}, {and} \bibinfo{person}{Jean-Daniel Fekete}.} \bibinfo{year}{2017}\natexlab{}.
\newblock \showarticletitle{A progressive kd tree for approximate k-nearest neighbors}. In \bibinfo{booktitle}{\emph{2017 IEEE Workshop on Data Systems for Interactive Analysis (DSIA)}}. IEEE, \bibinfo{pages}{1--5}.
\newblock


\bibitem[Kamel and Faloutsos(1992)]%
        {kamel1992parallel}
\bibfield{author}{\bibinfo{person}{Ibrahim Kamel} {and} \bibinfo{person}{Christos Faloutsos}.} \bibinfo{year}{1992}\natexlab{}.
\newblock \showarticletitle{Parallel R-trees}.
\newblock \bibinfo{journal}{\emph{ACM SIGMOD International Conference on Management of Data (SIGMOD)}} \bibinfo{volume}{21}, \bibinfo{number}{2} (\bibinfo{year}{1992}), \bibinfo{pages}{195--204}.
\newblock


\bibitem[Kanungo et~al\mbox{.}(2002)]%
        {kanungo2002efficient}
\bibfield{author}{\bibinfo{person}{Tapas Kanungo}, \bibinfo{person}{David~M Mount}, \bibinfo{person}{Nathan~S Netanyahu}, \bibinfo{person}{Christine~D Piatko}, \bibinfo{person}{Ruth Silverman}, {and} \bibinfo{person}{Angela~Y Wu}.} \bibinfo{year}{2002}\natexlab{}.
\newblock \showarticletitle{An efficient k-means clustering algorithm: Analysis and implementation}.
\newblock \bibinfo{journal}{\emph{IEEE transactions on pattern analysis and machine intelligence}} \bibinfo{volume}{24}, \bibinfo{number}{7} (\bibinfo{year}{2002}), \bibinfo{pages}{881--892}.
\newblock


\bibitem[Li et~al\mbox{.}(2018)]%
        {li2018so}
\bibfield{author}{\bibinfo{person}{Jiaxin Li}, \bibinfo{person}{Ben~M Chen}, {and} \bibinfo{person}{Gim~Hee Lee}.} \bibinfo{year}{2018}\natexlab{}.
\newblock \showarticletitle{So-net: Self-organizing network for point cloud analysis}. In \bibinfo{booktitle}{\emph{Proceedings of the IEEE conference on computer vision and pattern recognition}}. \bibinfo{pages}{9397--9406}.
\newblock


\bibitem[Li et~al\mbox{.}(2019)]%
        {li2019graph}
\bibfield{author}{\bibinfo{person}{Yujia Li}, \bibinfo{person}{Chenjie Gu}, \bibinfo{person}{Thomas Dullien}, \bibinfo{person}{Oriol Vinyals}, {and} \bibinfo{person}{Pushmeet Kohli}.} \bibinfo{year}{2019}\natexlab{}.
\newblock \showarticletitle{Graph matching networks for learning the similarity of graph structured objects}. In \bibinfo{booktitle}{\emph{International conference on machine learning}}. PMLR, \bibinfo{pages}{3835--3845}.
\newblock


\bibitem[Likas et~al\mbox{.}(2003)]%
        {likas2003global}
\bibfield{author}{\bibinfo{person}{Aristidis Likas}, \bibinfo{person}{Nikos Vlassis}, {and} \bibinfo{person}{Jakob~J Verbeek}.} \bibinfo{year}{2003}\natexlab{}.
\newblock \showarticletitle{The global k-means clustering algorithm}.
\newblock \bibinfo{journal}{\emph{Pattern recognition}} \bibinfo{volume}{36}, \bibinfo{number}{2} (\bibinfo{year}{2003}), \bibinfo{pages}{451--461}.
\newblock


\bibitem[Ma et~al\mbox{.}(2018)]%
        {ma2018query}
\bibfield{author}{\bibinfo{person}{Lin Ma}, \bibinfo{person}{Dana Van~Aken}, \bibinfo{person}{Ahmed Hefny}, \bibinfo{person}{Gustavo Mezerhane}, \bibinfo{person}{Andrew Pavlo}, {and} \bibinfo{person}{Geoffrey~J Gordon}.} \bibinfo{year}{2018}\natexlab{}.
\newblock \showarticletitle{Query-based workload forecasting for self-driving database management systems}. In \bibinfo{booktitle}{\emph{Proceedings of the 2018 International Conference on Management of Data}}. \bibinfo{pages}{631--645}.
\newblock


\bibitem[Malkov and Yashunin(2018)]%
        {malkov2018efficient}
\bibfield{author}{\bibinfo{person}{Yu~A Malkov} {and} \bibinfo{person}{Dmitry~A Yashunin}.} \bibinfo{year}{2018}\natexlab{}.
\newblock \showarticletitle{Efficient and robust approximate nearest neighbor search using hierarchical navigable small world graphs}.
\newblock \bibinfo{journal}{\emph{IEEE transactions on pattern analysis and machine intelligence}} \bibinfo{volume}{42}, \bibinfo{number}{4} (\bibinfo{year}{2018}), \bibinfo{pages}{824--836}.
\newblock


\bibitem[McInnes and Healy(2017)]%
        {mcinnes2017accelerated}
\bibfield{author}{\bibinfo{person}{Leland McInnes} {and} \bibinfo{person}{John Healy}.} \bibinfo{year}{2017}\natexlab{}.
\newblock \showarticletitle{Accelerated hierarchical density based clustering}. In \bibinfo{booktitle}{\emph{2017 IEEE International Conference on Data Mining Workshops (ICDMW)}}. IEEE, \bibinfo{pages}{33--42}.
\newblock


\bibitem[Men et~al\mbox{.}(2024)]%
        {pkdCode}
\bibfield{author}{\bibinfo{person}{Ziyang Men}, \bibinfo{person}{Zheqi Shen}, \bibinfo{person}{Yan Gu}, {and} \bibinfo{person}{Yihan Sun}.} \bibinfo{year}{2024}\natexlab{}.
\newblock \bibinfo{title}{Parallel $k$d-tree with Batch Updates}.
\newblock \bibinfo{howpublished}{\url{https://github.com/ucrparlay/Pkd-tree}}.
\newblock


\bibitem[Muja and Lowe(2014)]%
        {muja2014scalable}
\bibfield{author}{\bibinfo{person}{Marius Muja} {and} \bibinfo{person}{David~G Lowe}.} \bibinfo{year}{2014}\natexlab{}.
\newblock \showarticletitle{Scalable nearest neighbor algorithms for high dimensional data}.
\newblock \bibinfo{journal}{\emph{IEEE transactions on pattern analysis and machine intelligence}} \bibinfo{volume}{36}, \bibinfo{number}{11} (\bibinfo{year}{2014}), \bibinfo{pages}{2227--2240}.
\newblock


\bibitem[Overmars(1983)]%
        {overmars1983design}
\bibfield{author}{\bibinfo{person}{Mark~H Overmars}.} \bibinfo{year}{1983}\natexlab{}.
\newblock \bibinfo{booktitle}{\emph{The design of dynamic data structures}}. Vol.~\bibinfo{volume}{156}.
\newblock \bibinfo{publisher}{Springer Science \& Business Media}.
\newblock


\bibitem[Overmars and Van~Leeuwen(1981)]%
        {overmars1981maintenance}
\bibfield{author}{\bibinfo{person}{Mark~H Overmars} {and} \bibinfo{person}{Jan Van~Leeuwen}.} \bibinfo{year}{1981}\natexlab{}.
\newblock \showarticletitle{Maintenance of configurations in the plane}.
\newblock \bibinfo{journal}{\emph{Journal of computer and System Sciences}} \bibinfo{volume}{23}, \bibinfo{number}{2} (\bibinfo{year}{1981}), \bibinfo{pages}{166--204}.
\newblock


\bibitem[Prasad et~al\mbox{.}(2015)]%
        {prasad2015gpu}
\bibfield{author}{\bibinfo{person}{Sushil~K Prasad}, \bibinfo{person}{Michael McDermott}, \bibinfo{person}{Xi He}, {and} \bibinfo{person}{Satish Puri}.} \bibinfo{year}{2015}\natexlab{}.
\newblock \showarticletitle{GPU-based Parallel R-tree Construction and Querying}. In \bibinfo{booktitle}{\emph{2015 IEEE International Parallel and Distributed Processing Symposium Workshop}}. IEEE, \bibinfo{pages}{618--627}.
\newblock


\bibitem[Procopiuc et~al\mbox{.}(2003)]%
        {procopiuc2003bkd}
\bibfield{author}{\bibinfo{person}{Octavian Procopiuc}, \bibinfo{person}{Pankaj~K Agarwal}, \bibinfo{person}{Lars Arge}, {and} \bibinfo{person}{Jeffrey~Scott Vitter}.} \bibinfo{year}{2003}\natexlab{}.
\newblock \showarticletitle{Bkd-tree: A dynamic scalable kd-tree}. In \bibinfo{booktitle}{\emph{International Symposium on Spatial and Temporal Databases (SSTD)}}. Springer, \bibinfo{pages}{46--65}.
\newblock


\bibitem[Rajasekaran and Reif(1989)]%
        {RR89}
\bibfield{author}{\bibinfo{person}{Sanguthevar Rajasekaran} {and} \bibinfo{person}{John~H. Reif}.} \bibinfo{year}{1989}\natexlab{}.
\newblock \showarticletitle{Optimal and sublogarithmic time randomized parallel sorting algorithms}.
\newblock \bibinfo{journal}{\emph{{SIAM} J. on Computing}} \bibinfo{volume}{18}, \bibinfo{number}{3} (\bibinfo{year}{1989}), \bibinfo{pages}{594--607}.
\newblock


\bibitem[Reif and Neumann(2022)]%
        {reif2022scalable}
\bibfield{author}{\bibinfo{person}{Maximilian Reif} {and} \bibinfo{person}{Thomas Neumann}.} \bibinfo{year}{2022}\natexlab{}.
\newblock \showarticletitle{A scalable and generic approach to range joins}.
\newblock \bibinfo{journal}{\emph{Proceedings of the VLDB Endowment}} \bibinfo{volume}{15}, \bibinfo{number}{11} (\bibinfo{year}{2022}), \bibinfo{pages}{3018--3030}.
\newblock


\bibitem[Robinson(1981)]%
        {robinson1981kdb}
\bibfield{author}{\bibinfo{person}{John~T Robinson}.} \bibinfo{year}{1981}\natexlab{}.
\newblock \showarticletitle{The KDB-tree: a search structure for large multidimensional dynamic indexes}. In \bibinfo{booktitle}{\emph{ACM SIGMOD International Conference on Management of Data (SIGMOD)}}. \bibinfo{pages}{10--18}.
\newblock


\bibitem[Sch{\"a}ling(2011)]%
        {schaling2011boost}
\bibfield{author}{\bibinfo{person}{Boris Sch{\"a}ling}.} \bibinfo{year}{2011}\natexlab{}.
\newblock \bibinfo{booktitle}{\emph{The boost C++ libraries}}.
\newblock \bibinfo{publisher}{Boris Sch{\"a}ling}.
\newblock


\bibitem[Schubert et~al\mbox{.}(2017)]%
        {schubert2017dbscan}
\bibfield{author}{\bibinfo{person}{Erich Schubert}, \bibinfo{person}{J{\"o}rg Sander}, \bibinfo{person}{Martin Ester}, \bibinfo{person}{Hans~Peter Kriegel}, {and} \bibinfo{person}{Xiaowei Xu}.} \bibinfo{year}{2017}\natexlab{}.
\newblock \showarticletitle{DBSCAN revisited, revisited: why and how you should (still) use DBSCAN}.
\newblock \bibinfo{journal}{\emph{ACM Transactions on Database Systems (TODS)}} \bibinfo{volume}{42}, \bibinfo{number}{3} (\bibinfo{year}{2017}), \bibinfo{pages}{1--21}.
\newblock


\bibitem[Scoville et~al\mbox{.}(2007)]%
        {scoville2007cosmic}
\bibfield{author}{\bibinfo{person}{Nick Scoville}, \bibinfo{person}{H Aussel}, \bibinfo{person}{Marcella Brusa}, \bibinfo{person}{Peter Capak}, \bibinfo{person}{C~Marcella Carollo}, \bibinfo{person}{M Elvis}, \bibinfo{person}{M Giavalisco}, \bibinfo{person}{L Guzzo}, \bibinfo{person}{G Hasinger}, \bibinfo{person}{C Impey}, {et~al\mbox{.}}} \bibinfo{year}{2007}\natexlab{}.
\newblock \showarticletitle{The cosmic evolution survey (COSMOS): overview}.
\newblock \bibinfo{journal}{\emph{The Astrophysical Journal Supplement Series}} \bibinfo{volume}{172}, \bibinfo{number}{1} (\bibinfo{year}{2007}), \bibinfo{pages}{1}.
\newblock


\bibitem[Shakhnarovich et~al\mbox{.}(2005)]%
        {shakhnarovich2005nearest}
\bibfield{author}{\bibinfo{person}{Gregory Shakhnarovich}, \bibinfo{person}{Trevor Darrell}, {and} \bibinfo{person}{Piotr Indyk}.} \bibinfo{year}{2005}\natexlab{}.
\newblock \bibinfo{booktitle}{\emph{Nearest-neighbor methods in learning and vision: theory and practice}}. Vol.~\bibinfo{volume}{3}.
\newblock \bibinfo{publisher}{MIT press Cambridge, MA, USA:}.
\newblock


\bibitem[Shevtsov et~al\mbox{.}(2007)]%
        {shevtsov2007highly}
\bibfield{author}{\bibinfo{person}{Maxim Shevtsov}, \bibinfo{person}{Alexei Soupikov}, {and} \bibinfo{person}{Alexander Kapustin}.} \bibinfo{year}{2007}\natexlab{}.
\newblock \showarticletitle{Highly parallel fast KD-tree construction for interactive ray tracing of dynamic scenes}. In \bibinfo{booktitle}{\emph{Computer Graphics Forum}}, Vol.~\bibinfo{volume}{26}. Wiley Online Library, \bibinfo{pages}{395--404}.
\newblock


\bibitem[Silva et~al\mbox{.}(2013)]%
        {silva2013data}
\bibfield{author}{\bibinfo{person}{Jonathan~A Silva}, \bibinfo{person}{Elaine~R Faria}, \bibinfo{person}{Rodrigo~C Barros}, \bibinfo{person}{Eduardo~R Hruschka}, \bibinfo{person}{Andr{\'e} CPLF~de Carvalho}, {and} \bibinfo{person}{Jo{\~a}~o Gama}.} \bibinfo{year}{2013}\natexlab{}.
\newblock \showarticletitle{Data stream clustering: A survey}.
\newblock \bibinfo{journal}{\emph{ACM Computing Surveys (CSUR)}} \bibinfo{volume}{46}, \bibinfo{number}{1} (\bibinfo{year}{2013}), \bibinfo{pages}{1--31}.
\newblock


\bibitem[Sleator and Tarjan(1985)]%
        {Sleator85}
\bibfield{author}{\bibinfo{person}{Daniel~D. Sleator} {and} \bibinfo{person}{Robert~E. Tarjan}.} \bibinfo{year}{1985}\natexlab{}.
\newblock \showarticletitle{Amortized Efficiency of List Update and Paging Rules}.
\newblock \bibinfo{journal}{\emph{Commun. {ACM}}} \bibinfo{volume}{28}, \bibinfo{number}{2} (\bibinfo{year}{1985}), \bibinfo{numpages}{7}~pages.
\newblock
\showISSN{0001-0782}
\urldef\tempurl%
\url{https://doi.org/10.1145/2786.2793}
\showDOI{\tempurl}


\bibitem[Smith et~al\mbox{.}(2016)]%
        {smith2016structure}
\bibfield{author}{\bibinfo{person}{Mark~William Smith}, \bibinfo{person}{Jonathan~L Carrivick}, {and} \bibinfo{person}{Duncan~J Quincey}.} \bibinfo{year}{2016}\natexlab{}.
\newblock \showarticletitle{Structure from motion photogrammetry in physical geography}.
\newblock \bibinfo{journal}{\emph{Progress in physical geography}} \bibinfo{volume}{40}, \bibinfo{number}{2} (\bibinfo{year}{2016}), \bibinfo{pages}{247--275}.
\newblock


\bibitem[Sun et~al\mbox{.}(2018)]%
        {sun2018pam}
\bibfield{author}{\bibinfo{person}{Yihan Sun}, \bibinfo{person}{Daniel Ferizovic}, {and} \bibinfo{person}{Guy~E Blelloch}.} \bibinfo{year}{2018}\natexlab{}.
\newblock \showarticletitle{{PAM}: Parallel Augmented Maps}. In \bibinfo{booktitle}{\emph{{ACM} Symposium on Principles and Practice of Parallel Programming (PPOPP)}}.
\newblock


\bibitem[Tang et~al\mbox{.}(2016)]%
        {tang2016visualizing}
\bibfield{author}{\bibinfo{person}{Jian Tang}, \bibinfo{person}{Jingzhou Liu}, \bibinfo{person}{Ming Zhang}, {and} \bibinfo{person}{Qiaozhu Mei}.} \bibinfo{year}{2016}\natexlab{}.
\newblock \showarticletitle{Visualizing large-scale and high-dimensional data}. In \bibinfo{booktitle}{\emph{Proceedings of the 25th international conference on world wide web}}. \bibinfo{pages}{287--297}.
\newblock


\bibitem[{The CGAL Project}(2020)]%
        {cgal51}
\bibfield{author}{\bibinfo{person}{{The CGAL Project}}.} \bibinfo{year}{2020}\natexlab{}.
\newblock \bibinfo{booktitle}{\emph{{CGAL} User and Reference Manual} (\bibinfo{edition}{{5.1}} ed.)}.
\newblock \bibinfo{publisher}{{CGAL Editorial Board}}.
\newblock
\urldef\tempurl%
\url{https://doc.cgal.org/5.1/Manual/packages.html}
\showURL{%
\tempurl}


\bibitem[van Kreveld and Overmars(1991)]%
        {van1991divided}
\bibfield{author}{\bibinfo{person}{Marc~J van Kreveld} {and} \bibinfo{person}{Mark~H Overmars}.} \bibinfo{year}{1991}\natexlab{}.
\newblock \showarticletitle{Divided kd trees}.
\newblock \bibinfo{journal}{\emph{Algorithmica}}  \bibinfo{volume}{6} (\bibinfo{year}{1991}), \bibinfo{pages}{840--858}.
\newblock


\bibitem[Wang et~al\mbox{.}(2022)]%
        {wang2022pargeo}
\bibfield{author}{\bibinfo{person}{Yiqiu Wang}, \bibinfo{person}{Shangdi Yu}, \bibinfo{person}{Laxman Dhulipala}, \bibinfo{person}{Yan Gu}, {and} \bibinfo{person}{Julian Shun}.} \bibinfo{year}{2022}\natexlab{}.
\newblock \showarticletitle{ParGeo: a library for parallel computational geometry}. In \bibinfo{booktitle}{\emph{European Symposium on Algorithms (ESA)}}.
\newblock


\bibitem[Wang et~al\mbox{.}(2021)]%
        {wang2021fast}
\bibfield{author}{\bibinfo{person}{Yiqiu Wang}, \bibinfo{person}{Shangdi Yu}, \bibinfo{person}{Yan Gu}, {and} \bibinfo{person}{Julian Shun}.} \bibinfo{year}{2021}\natexlab{}.
\newblock \showarticletitle{Fast parallel algorithms for euclidean minimum spanning tree and hierarchical spatial clustering}. In \bibinfo{booktitle}{\emph{ACM SIGMOD International Conference on Management of Data (SIGMOD)}}. \bibinfo{pages}{1982--1995}.
\newblock


\bibitem[Yamasaki et~al\mbox{.}(2018)]%
        {yamasaki2018parallelizing}
\bibfield{author}{\bibinfo{person}{Hiroki Yamasaki}, \bibinfo{person}{Atsushi Nunome}, {and} \bibinfo{person}{Hiroaki Hirata}.} \bibinfo{year}{2018}\natexlab{}.
\newblock \showarticletitle{Parallelizing the Construction of a k-Dimensional Tree}. In \bibinfo{booktitle}{\emph{2018 IEEE International Conference on Big Data, Cloud Computing, Data Science \& Engineering (BCD)}}. IEEE, \bibinfo{pages}{23--30}.
\newblock


\bibitem[You et~al\mbox{.}(2013)]%
        {you2013parallel}
\bibfield{author}{\bibinfo{person}{Simin You}, \bibinfo{person}{Jianting Zhang}, {and} \bibinfo{person}{Le Gruenwald}.} \bibinfo{year}{2013}\natexlab{}.
\newblock \showarticletitle{Parallel spatial query processing on gpus using r-trees}. In \bibinfo{booktitle}{\emph{Proceedings of the 2Nd ACM SIGSPATIAL international workshop on analytics for big geospatial data}}. \bibinfo{pages}{23--31}.
\newblock


\bibitem[Yue et~al\mbox{.}(2016)]%
        {yue2016healthcare}
\bibfield{author}{\bibinfo{person}{Xiao Yue}, \bibinfo{person}{Huiju Wang}, \bibinfo{person}{Dawei Jin}, \bibinfo{person}{Mingqiang Li}, {and} \bibinfo{person}{Wei Jiang}.} \bibinfo{year}{2016}\natexlab{}.
\newblock \showarticletitle{Healthcare data gateways: found healthcare intelligence on blockchain with novel privacy risk control}.
\newblock \bibinfo{journal}{\emph{Journal of medical systems}}  \bibinfo{volume}{40} (\bibinfo{year}{2016}), \bibinfo{pages}{1--8}.
\newblock


\bibitem[Zheng et~al\mbox{.}(2008)]%
        {geolife}
\bibfield{author}{\bibinfo{person}{Yu Zheng}, \bibinfo{person}{Like Liu}, \bibinfo{person}{Longhao Wang}, {and} \bibinfo{person}{Xing Xie}.} \bibinfo{year}{2008}\natexlab{}.
\newblock \showarticletitle{Learning transportation mode from raw gps data for geographic applications on the web}. In \bibinfo{booktitle}{\emph{International World Wide Web Conference (WWW)}}. \bibinfo{pages}{247--256}.
\newblock


\end{thebibliography}
\iffullversion{
	\clearpage
	\appendix

\section{Proof for Tree Height\label{app:treeHeightSupport}}
\begin{lemma}
	Function $f(n)=-\log n/\log(1/2+1/\log n)-\log n=O(1)$ for $n> 4$.
\end{lemma}

\begin{proof}
	Let $t=\log n$, we have:
	\begin{align*}
		f(t) & = - \frac{t}{\log(t+2)-\log 2t}- t                \\
		     & = - \left(\frac{t}{\log(t+2)-\log t-1} + t\right) \\
		     & = -h(t)
	\end{align*}
	Clearly, $f(t)=O(1)$ when $t\to 2^+$. We then show $f(t)>0$ holds for $t> 2$, which is
	:
	\begin{align*}
		\log\frac{t+2}{2t}>\log\frac{1}{2} & \implies \frac{1}{\log(t+2)-\log2t}+1<0 \\
		                                   & \implies h(t)<0
	\end{align*}
	as desired.
	The remaining is to show $f(t)$ is decreasing, which equivalents to show $h(t)$ is increasing over $t>2$. The derivative of $h(t)$ is
	\begin{align*}
		h'(t) & = \frac{\log(t+2)-\log t -1 -t\left(\frac{c}{t+2}-\frac{c}{t}\right)}{\left(\log(t+2)-\log t-1\right)^2}+1 \\
		      & = \frac{\log\frac{t+2}{t}-1 - \frac{c\cdot t}{t+2}+c}{\left(\log\frac{t+2}{t}-1\right)^2}+1
	\end{align*}
	where $c=1/\ln 2$. Let $k = (t+2)/t$. We wish to show that $h'(k)>0$ holds for $k\to 1^+$, namely,
	\begin{align*}
		                     & \frac{\log k-1 -c/k +c}{\left(\log k -1\right)^2} +1 >0 \\
		\Longleftrightarrow~ & \log^2 k - \log k -c/k+c>0                              \\
		\Longleftrightarrow~ & g(k) >0
	\end{align*}
	Since $g(1^+) > 0$, therefore, it is sufficient to show $g'(k)>0$ holds for $k$, i.e.,
	\begin{align*}
		                     & 2c^2\cdot \ln k/k -c/k+c/k^2 >0 \\
		\Longleftrightarrow~ & 2c\cdot k\ln k-k >-1            \\
		\Longleftrightarrow~ & k(2c\ln k -1) > -1
	\end{align*}
	The function w.r.t $k$ in LHS is increasing and equals to $-1$ (the RHS) when $k=1$. Proof follows then.
\end{proof}

\section{Proof for Batch Updates\label{app:batchUpdate}}
\begin{theorem}[Updates]
	A batch update (insertions or deletions) of a batch of size $m$ on a \ourkdtree of size $n$ has $O(\log n\log_M n)$ span \whp; the amortized work and cache complexity per element in the batch is $O(\log^2 n)$ and $O(\log(n/m)+(\log n\log_M n)/B)$ \whp, respectively.
\end{theorem}

\begin{proof}
	We will start with the span bound.
	According to \cref{thm:constr-polylog}, the sieve process and trees rebuilding (rebalancing) all have $O(\log n\log_M n)$ span \whp.
	Note that the tree rebuilding (\cref{line:update-build-tree}) can only be triggered once on any tree path.
	The span for other parts is $O(\log n)$---the \textsc{InsertToSkeleton} function can be recursively call for $\levels = O(\log n)$ levels each with constant cost.
	In total, the span is the same as the construction algorithm, since in the extreme case, the entire tree can be rebuilt.

	Then we show the work bound.
	The cost to traverse the \ourtree and find the corresponding leaves to update is $O(\log n)$ per point \whp{}, proportional to the tree height.
	Once a rebuild is triggered (on \cref{line:update-build-tree}), the cost is $O(n'\log n')$ where $n'$ is the subtree size.
	After that (or the original construction), each subtree will contains $(1/2\pm \sqrt{(12c\log n)/\os}/4)n'=(1/2\pm \balpara/4)n'$ points \whp (\cref{lem:sampling}).
	We need to insert at least another $3\balpara n'/4$ points for this subtree to be sufficiently imbalance that triggers the next rebuilding of this subtree.
	The amortized cost per point in this subtree is hence $O(\log n'/\balpara) = O(\log n')$ on this tree node assuming $\balpara$ is a constant.
	Note that \ourtree has the tree height of $O(\log n)$, so overall amortized work per inserted/deleted point is $O(\log^2 n)$.

	We can analyze the cache complexity similarly.
	We first show the rebuilding cost.
	For a subtree of size $n'$, the cost is $O((n'/B)\log_M n')$ (\cref{thm:constr}).
	The amortized cost per updated point, using the same analysis above, is $O((1/B)\log_M n')$.
	Again since the tree height is $O(\log n)$, the overall amortized work per inserted/deleted point is $O((\log n\log_M n)/B)$.
	Then, we consider the cost to traverse the tree and find the corresponding leaves to update.
	Finding $m$ leaves in a tree of size $n$ will touch $O(m\log(n/m))$ tree nodes~\cite{blelloch2016just}, so the amortized block transfer per point is $O(\log(n/m))$.
	Putting both cost together gives the stated cache complexity.
\end{proof}

\hide{
	\begin{proof}
		Let $t=\log n$, proving the lemma is to show function
		\begin{align*}
			f(t) & = - \frac{t}{\log(t-2)-\log 2t}- t \\
			     & = \frac{t}{1+\log t-\log(t-2)} - t
		\end{align*}
		is decreasing when $t\to\infty$, which is equivalent to prove that $f'(t)<0$ holds for $t\to\infty$, namely:
		\begin{align*}
			f'(t) = \frac{1+\log \frac{t}{t-2}-c+c\cdot \frac{t}{t-2}}{(1+\log\frac{t}{t-2})^2}-1<0\label{firstDiff}
		\end{align*}
		where $c=1/\ln2$. Let $k=t/(t-2)$, then \cref{firstDiff} can be transformed to:
		\begin{align*}
			1+\log k-c+c\cdot k<(1+\log k)^2
		\end{align*}
		which equivalents to show:
		\begin{align*}
			g(k) = k -c\ln^2k-\ln k<1 \label{secondIneqaul}
		\end{align*}
		holds for $k\to1^+$. Note that $g(k) = 1$ when $k=1$, thus it is sufficient to show $g'(k)<0$ for $k\in(1,1+\epsilon)$, where $\epsilon\to 0$. Formally:
		\begin{align*}
			g'(k) = 1-\frac{2c}{k}\ln k-\frac{1}{k} < 0
		\end{align*}
		Since $k>1$, this equivalents to show:
		\begin{align*}
			h(k) = k-2c\ln k  < 1
		\end{align*}
		Again, when $k=1$, $h(k) = 0$, we then show $h'(k)< 0$ holds, clearly:
		\begin{align*}
			\lim_{k\to1^+}h'(k) = \lim_{k\to1^+}1-\frac{2c}{k}= 1-\frac{2}{\ln 2} < 0
		\end{align*}
		Proof finish.
	\end{proof}
}

\section{Handling of Duplicates\label{app:heavyleaf}}
One issue we observed in some existing \kdtree implementations (e.g., \cgal{}) is
the inefficiency in dealing with duplicate points.
Because many points may fall onto the split hyperplane and a default approach will put all of them on one side of
the tree, the tree height (and thus the query performance) may degenerate significantly.
\ourtree uses a special \emph{heavy leaf} to handle this.
When all points in a node are duplicates, we use
a heavy leaf to store the coordinate and the count.

\section{Comparison to the \rtree{}\label{app:rtree}}
The \rtree{} is a commonly seen spatial index structure in practice.
The Boost library~\cite{schaling2011boost} provides an optimized sequential implementation of the \rtree{}.
We compare the \ourlib{} with the Boost \rtree{} in terms of tree construction, \knn{}, and range report.
The \rtree{} in Boost only supports sequential tree construction and updates.
For \knn{} and range report, we parallelize all queries for both \ourtree{s} and \rtree{s}.
We note that the \rtree{} supports more general functionalities than \kdtree{s}, such as range queries with arbitrary shapes, spatial overlap/intersection queries, handling non-point and high-dimensional objects, and may not be specifically optimized for \knn{} or range queries.
\cref{table:app_rtree_zdtree} summarizes the results.

With 96 cores, our parallel \ourlib{} is 91.9--115$\times$ faster than the sequential Boost \rtree{} in construction,
30--580$\times$ faster in batch insertions, and 133--2259$\times$ faster in deletion.
Since Boost \rtree{} only supports point updates, the speedup of \ourtree{} comes from both handling multiple updates in a batch
and parallelism.
Regarding the \knn{} query, the \ourlib{} is 10.4--14.8$\times$ faster than the Boost \rtree{}.
For the range report query, the \ourlib{} is 4.13--4.17$\times$ faster.
One possible reason for the speedup is that \ourtree{} can output the candidates of a subtree in parallel when its associated subspace is fully contained in the query box by running a parallel flatten.

\begin{table*}[t]
	\centering
	\small
	\setlength\tabcolsep{6pt} 
	\renewcommand{\arraystretch}{0.8} 

	\begin{tabular}{cc|c|cccc|cccc|c|c}
		\toprule
		\textbf{Benchmark}                    & \multirow{2}[2]{*}{\textbf{Baselines}} & \multirow{2}[2]{*}{\textbf{Build}} & \multicolumn{4}{c|}{\textbf{Batch Insert}} & \multicolumn{4}{c|}{\textbf{Batch Delete}} & \textbf{10-NN }  & \textbf{Range Report}                                                                                                                     \\
		\textbf{(1000M-2D)}                   &                                        &                                    & \textbf{0.01\%}                            & \textbf{0.1\%}                             & \textbf{1\%}     & \textbf{10\%}         & \textbf{0.01\%}  & \textbf{0.1\%}   & \textbf{1\%}     & \textbf{10\%}    & \textbf{(1\%)}   & \textbf{(10K, 1M]} \\
		\midrule
		\multirow{3}[2]{*}{\textbf{\uniform}} & Ours                                   & \underline{3.15}                   & \underline{.004}                           & \underline{.020}                           & \underline{.104} & \underline{.495}      & \underline{.004} & \underline{.022} & \underline{.121} & \underline{.526} & \underline{.381} & \underline{.391}   \\
		                                      & \rtree{} (seq.)                        & 363                                & .282                                       & 2.86                                       & 28.6             & 287                   & 1.10             & 11.6             & 121              & 1188             & 5.64             & 1.62               \\
		                                      & \zdtree                                & 6.70                               & n.a.                                       & n.a.                                       & n.a.             & n.a.                  & n.a.             & n.a.             & n.a.             & n.a.             & .870             & n.a.               \\
		\midrule
		\multirow{3}[2]{*}{\textbf{\varden}}  & Ours                                   & \underline{3.66}                   & \underline{.002}                           & \underline{.007}                           & \underline{.055} & \underline{.473}      & \underline{.002} & \underline{.006} & \underline{.049} & \underline{.477} & \underline{.172} & \underline{.382}   \\
		                                      & \rtree{} (seq.)                        & 336                                & .060                                       & .606                                       & 6.21             & 64.7                  & .267             & 2.23             & 24.7             & 308              & 1.79             & 1.60               \\
		                                      & \zdtree                                & 6.70                               & n.a.                                       & n.a.                                       & n.a.             & n.a.                  & n.a.             & n.a.             & n.a.             & n.a.             & .331             & n.a.               \\
		\bottomrule
	\end{tabular}%

	\vspace{.2em}
	\caption{
		\textbf{Running time (in seconds) for the \ourtree{} and other baselines on $10^9$ points in 2 dimensions. Lower is better.} \normalfont
		``\rtree{} (seq.)'': the serial \rtree{} implementation from Boost~\cite{schaling2011boost}.
		``\zdtree'': the parallel Quad/Oct-tree implementation using the Morton order from~\cite{blelloch2022parallel}.
		``10-NN'': 10-nearest-neighbor queries on $10^7$ points.
		``Range report'': orthogonal range report queries on $10^4$ rectangles, with output sizes in $10^4$--$10^6$.
		The fastest time for each test is underlined.
		``n.a.'': not applicable.
	}
	\label{table:app_rtree_zdtree}%
\end{table*}%

\section{Comparison to the \zdtree{}\label{app:zdtree}}
We also compare the \ourlib with the \zdtree~\cite{blelloch2022parallel}, which is a parallel quad/octree based on the Morton order (aka. the space-filling curve).
The \zdtree maps each point to an integer by interleaving the bits of the coordinates and uses integer sort as preprocessing
so that the tree construction has $O(n)$ work and $O(n^\epsilon)$ span, where $n$ is the input size and $\epsilon<1$ a constant.

The implementation of the \zdtree is part of the PBBS~\cite{anderson2022problem}, which should support parallel tree construction, batch updates, and \knn{}.
Due to integer precision issues, \zdtree{'s} implementation only supports inputs in 2 or 3 dimensions.
We omit the evaluation for batch updates since their code has problems that are causing it to produce incorrect results.
\cref{table:app_rtree_zdtree} summarizes the comparison between the \ourlib{} and the \zdtree{}.

For construction, the \ourlib{} is $1.34$--$2.12\times$ faster than the \zdtree{} on both datasets and dimensions.
This is surprising as \zdtree{} is a quad/octree based on the space-filling curve, which handles multi-dimensional points as integers.
Therefore, the computation is simpler than the \kdtree{-based} structures in construction and updates
(and is thus reasonable to achieve higher performance).
We believe the reason is the I/O optimizations we designed in \ourlib{s}, and applying them to quad/octree can be an interesting future work.
For \knn{} query, the \ourlib{} is $1.22$--$2.40\times$ faster than the \zdtree{} on both datasets.
The reason is that the \ourlib{} uses object median as a cutting plane, which enables more efficient special prune during the searches.
Besides, the low memory usage of the tree structure of the \ourlib{} achieves higher cache utilization and thus better query performance.


\begin{table}[t]
	\centering
	\small
	\setlength\tabcolsep{2.4pt}
	\renewcommand{\arraystretch}{1.0}

	\begin{tabular}{c|c|ccc|ccc}
		\toprule
		\multirow{3}[2]{*}{\textbf{Tree}}     & \multirow{3}[2]{*}{\textbf{Baselines}} & \multicolumn{6}{c}{\textbf{Query points}}                                                                                                                    \\
		\cline{3-8}
		                                      &                                        & \multicolumn{3}{c|}{\textbf{\uniform}}    & \multicolumn{3}{c}{\textbf{\varden}}                                                                             \\
		                                      &                                        & \textbf{1-NN}                             & \textbf{10-NN}                       & \textbf{100-NN}  & \textbf{1-NN}    & \textbf{10-NN}   & \textbf{100-NN}  \\
		\midrule
		\multirow{4}[2]{*}{\textbf{\uniform}} & Ours                                   & \underline{.209}                          & \underline{.765}                     & \underline{2.66} & \underline{.086} & \underline{.162} & \underline{.774} \\
		                                      & \logtree{}                             & 2.43                                      & 4.48                                 & 16.9             & 1.21             & 2.00             & 15.7             \\
		                                      & \bhltree{}                             & .315                                      & 1.02                                 & 4.24             & .204             & .911             & 8.89             \\
		                                      & \cgal{}                                & .333                                      & 2.32                                 & 13.0             & .108             & .223             & 1.57             \\
		\midrule
		\multirow{4}[2]{*}{\textbf{\varden}}  & Ours                                   & \underline{.938}                          & \underline{1.71}                     & \underline{4.15} & \underline{.056} & \underline{.190} & \underline{.888} \\
		                                      & \logtree{}                             & t.o.                                      & t.o.                                 & t.o.             & 2.43             & 4.48             & 16.9             \\
		                                      & \bhltree{}                             & t.o.                                      & t.o.                                 & t.o.             & .315             & 1.02             & 4.24             \\
		                                      & \cgal{}                                & 1.68                                      & 3.07                                 & 7.80             & .083             & .219             & 1.28             \\
		\bottomrule
	\end{tabular}%

	\caption{
		\textbf{In-distribution and out-of-distribution \knn{} query time (in seconds) for \ourlib{} and other baselines on synthetic datasets with 3 dimensions. Lower is better.} The tree contains $10^9$ points, and the query points contains $10^7$ candidates. ``t.o.'': time out after 600s.}
	\label{table:ood}%

	\small
	\setlength\tabcolsep{2.4pt}
	\renewcommand{\arraystretch}{1.0}
	\begin{tabular}{c|cc|cccc|cccc}
		\toprule
		\multirow{2}[2]{*}{\textbf{Bench.}}                                                                       & \multirow{2}[2]{*}{\textbf{Baselines}}                & \multirow{2}[2]{*}{{$\boldsymbol{\alpha}$}} & \multicolumn{4}{c|}{\textbf{Batch Insert (1\%)}} & \multicolumn{4}{c}{\textbf{Batch Delete (1\%)}}                                                                                                                   \\
&                                              &                              & \textbf{2}                                       & \textbf{3}                                      & \textbf{5}                & \textbf{9}                & \textbf{2}                & \textbf{3}                & \textbf{5}                & \textbf{9}                \\
		\midrule
		\multicolumn{1}{c|}{\multirow{6}[2]{*}{\begin{tabular}[c]{@{}c@{}}Uniform\\ 1000M\end{tabular}}} & \multicolumn{1}{c}{\multirow{3}[1]{*}{Ours}} & 0.03                         & 2.95                                    & 4.23                                   & 6.05             & 10.3             & 3.46             & 4.60             & 6.66             & 10.9             \\
		                                                                                                 &                                              & 0.1                          & .686                                    & .799                                   & 1.36             & 1.97             & .776             & .986             & 1.52             & 2.19             \\
		                                                                                                 &                                              & 0.3                          & \underline{.104}                        & \underline{.107}                       & \underline{.123} & \underline{.152} & \underline{.121} & \underline{.134} & \underline{.171} & \underline{.232} \\
		                                                                                                 & Log-tree                                     & -                            & 2.16                                    & 2.66                                   & 3.67             & 6.19             & .396             & .485             & 1.94             & 2.39             \\
		                                                                                                 & BHL-tree                                     & -                            & 31.4                                    & 40.3                                   & 57.1             & 103              & 30.9             & 39.3             & 68.7             & 114              \\
		                                                                                                 & CGAL                                         & -                            & 1660                                    & 1815                                   & 1863             & 2145             & 41.2             & 41.3             & 45.0             & 40.2             \\
		\midrule
		\multicolumn{1}{c|}{\multirow{6}[2]{*}{\begin{tabular}[c]{@{}c@{}}Varden\\ 1000M\end{tabular}}}  & \multicolumn{1}{c}{\multirow{3}[1]{*}{Ours}} & 0.03                         & 2.81                                    & 3.05                                   & 4.21             & 9.65             & 7.79             & 10.7             & 10.9             & 20.5             \\
		                                                                                                 &                                              & 0.1                          & .485                                    & .690                                   & .978             & 1.83             & 2.09             & 2.09             & 2.65             & 4.11             \\
		                                                                                                 &                                              & 0.3                          & \underline{.055}                        & \underline{.107}                       & \underline{.157} & \underline{.350} & \underline{.049} & \underline{.112} & \underline{.127} & \underline{.237} \\
		                                                                                                 & Log-tree                                     & -                            & 2.01                                    & 2.60                                   & 3.72             & 6.07             & 1.06             & 1.14             & 1.92             & 2.30             \\
		                                                                                                 & BHL-tree                                     & -                            & 29.4                                    & 39.1                                   & 57.3             & 102              & 29.0             & 38.4             & 67.0             & 123              \\
		                                                                                                 & CGAL                                         & -                            & 849                                     & 700                                    & 582              & 599              & 13.0             & 9.53             & 23.1             & 3.90             \\
		\bottomrule
	\end{tabular}%

	\caption{
		\textbf{Batch update time (in seconds) for \ourlib{} and other baselines on synthetic datasets with dimensions 2, 3, 5, and 9. Lower is better.} The tree contains $10^9$ points, and the batch contains $10^7$ points from the same distribution as the points in the tree. Parameter $\alpha$ is the imbalance ratio used in \ourlib{}.}
	\label{table:3inba}%
\end{table}%

\section{Out-of-distribution Queries\label{app:ood}}
In the context of tree queries, the out-of-distribution (OOD) query refers to the query with a distribution that is significantly different from the data distribution used to construct the tree. This can lead to inefficiencies because the tree structure may not be optimized for such queries.
We perform both the in-distribution and out-of-distribution query for all baselines by first constructing the tree using the points from one distribution and then performing the \knn{} query using points from another distribution.
\cref{table:ood} illustrates the results.

In general, \ourtree{} still achieves the best performance.
Both the \ourlib{} and \cgal{} are relatively resistant to OOD queries, which is reasonable as the \kdtree{s} use of the object median as the cutting plane is less sensitive to the data distribution.
However, both the \logtree{} and the \bhltree{} suffer from timeouts on one of the OOD queries.
This is due to the prune heuristic used in their implementation, which skips a (sub-)tree only when the distance between the query point and the cutting plane (rather than the bounding box) is larger than the current best candidate.
While this is a commonly used heuristic that trades off more tree nodes traversed for faster computing per node, we note that this heuristic leads to much worse performance for some OOD queries.

\section{Imbalance Ratios for Batch Update\label{app:inbaratios}}
We measure the batch update time of the \ourlib{} with different values of imbalance ratio $\alpha$,
which serves as a complementary for \cref{table:summary}.
The results are summarized in \cref{table:3inba}.
Note that we omit the tree construction and queries since these metrics tested in \cref{table:summary} are not affected by the imbalance ratio.

When $\alpha=0.03$, batch insertion in the \ourlib{} almost always rebuilds the entire tree, as does batch deletion.
However, the \ourlib{} is still two orders of magnitude faster than the \bhltree{}, which rebuilds the whole tree for batch insertion as well.
The \ourlib{} is also one order of magnitude faster than the \bhltree{} for batch deletions.
When the imbalance is more tolerated and $\alpha=0.1$, the \ourlib{} becomes the fastest among all baselines for batch insertion.
For batch deletion, however, the \logtree{} is faster than the \ourlib{} on 6 out of 8 instances by a factor of 1.39--2.03$\times$, as the \logtree{} only needs to build a small portion of the tree when the batch size is small.
Finally, with $\alpha$ set to 0.3, the \ourlib{} achieves a speedup of 5.24--67.6$\times$ for batch insertion and 6.42--159$\times$ for batch deletion, compared to $\alpha=0.03$ or $\alpha=0.1$. With $\alpha=0.3$, the \ourtree{} is also the fastest among all baselines.

\begin{table}[t]

	\small
	\setlength\tabcolsep{2.2pt}
	\renewcommand{\arraystretch}{1.2}

	\begin{tabular}{cc|ccccccc}
		\toprule
		                                                       & \textbf{Tree} & \textbf{CC(M)}      & \textbf{Inst(M)}   & \textbf{IPC}     & \textbf{CRs(M)} & \textbf{CMs(M)} & \textbf{BR(M)}    & \textbf{BMs(M)}  \\
		\midrule
		\multirow{2}[2]{*}{\begin{sideways}HT\end{sideways}}   & \ours{}       & 95,300              & 31,084             & .326             & 926             & 508             & 6,832             & 58.0             \\
		                                                       & \oursbb{}     & \underline{33,885}  & \underline{17,143} & \underline{.506} & \underline{245} & \underline{134} & \underline{2,279} & \underline{18.0} \\
		\midrule
		\multirow{2}[2]{*}{\begin{sideways}HH\end{sideways}}   & \ours{}       & 69,078              & 22,127             & .320             & 557             & 361             & 4,114             & 117              \\
		                                                       & \oursbb{}     & \underline{27,987}  & \underline{10,831} & \underline{.387} & \underline{242} & \underline{150} & \underline{1,636} & \underline{44.0} \\
		\midrule
		\multirow{2}[2]{*}{\begin{sideways}CHEM\end{sideways}} & \ours{}       & 243,014             & 65,919             & \underline{.271} & 1,662           & 1,450           & 14,318            & 170              \\
		                                                       & \oursbb{}     & \underline{139,701} & \underline{32,887} & .235             & \underline{954} & \underline{820} & \underline{6,583} & \underline{107}  \\
		\midrule
		\multirow{2}[2]{*}{\begin{sideways}GL\end{sideways}}   & \ours{}       & 71,844              & 21,767             & \underline{.303} & 439             & 345             & 3,376             & 118              \\
		                                                       & \oursbb{}     & \underline{62,637}  & \underline{13,438} & .215             & \underline{407} & \underline{317} & \underline{1,981} & \underline{101}  \\
		\midrule
		\multirow{2}[2]{*}{\begin{sideways}CM\end{sideways}}   & \ours{}       & \underline{120,478} & 22,407             & \underline{.186} & \underline{692} & \underline{649} & 3,397             & \underline{173}  \\
		                                                       & \oursbb{}     & 133,104             & \underline{21,913} & .165             & 752             & 703             & \underline{3,296} & 176              \\
		\midrule
		\multirow{2}[2]{*}{\begin{sideways}OSM\end{sideways}}  & \ours{}       & \underline{71,679}  & 12,247             & \underline{.171} & \underline{460} & \underline{426} & 1,366             & 50.0             \\
		                                                       & \oursbb{}     & 75,185              & \underline{11,124} & .148             & 473             & 441             & \underline{1,178} & \underline{48.0} \\
		\bottomrule
	\end{tabular}%

	\caption{
		\textbf{Hardware profiling of vanilla \ourlib{} (\ours{}) and a variant with bounding box optimizations (\oursbb{}) for range report query on real-world datasets.
			Underlined values indicate better performance.
		}
		The range report query contains $10^4$ rectangles each with output size $10^4$--$10^6$. Different queries are performed in parallel, and each query searches the tree in serial.
		``CC'': Cycles, ``Inst'': Instructions, ``IPC'': Instructions per cycle, ``CR'': Cache reference, ``CMs'': Cache misses, ``BR'': Branches, ``BMs'': Branch misses.
	}

	\label{table:appendix:perf-hardware}%

	\small
	\setlength\tabcolsep{6pt}
	\renewcommand{\arraystretch}{1.2}
	\begin{tabular}{cc|ccccc}
		\toprule
		                                                       & \multirow{2}[2]{*}{\textbf{Tree}} & \multirow{2}[2]{*}{\textbf{Time (sec.)}} & \multicolumn{4}{c}{\textbf{Average \# of nodes proceed}}                                                            \\
		                                                       &                                   &                                          & \textbf{Leaf}                                            & \textbf{Interior} & \textbf{Skip}     & \textbf{Flatten} \\
		\midrule
		\multirow{2}[2]{*}{\begin{sideways}HT\end{sideways}}   & \ours{}                           & .587                                     & 2,675                                                    & 3,957             & \underline{1,282} & 0                \\
		                                                       & \oursbb{}                         & \underline{.268}                         & \underline{207}                                          & \underline{891}   & 605               & \underline{79}   \\
		\midrule
		\multirow{2}[2]{*}{\begin{sideways}HH\end{sideways}}   & \ours{}                           & .385                                     & 2,817                                                    & 3,621             & \underline{802}   & 3                \\
		                                                       & \oursbb{}                         & \underline{.192}                         & \underline{615}                                          & \underline{1,135} & 455               & \underline{65}   \\
		\midrule
		\multirow{2}[2]{*}{\begin{sideways}CHEM\end{sideways}} & \ours{}                           & 1.15                                     & 4,276                                                    & 5,701             & \underline{1,425} & 0                \\
		                                                       & \oursbb{}                         & \underline{.837}                         & \underline{1,330}                                        & \underline{2,506} & 1,091             & \underline{84}   \\
		\midrule
		\multirow{2}[2]{*}{\begin{sideways}GL\end{sideways}}   & \ours{}                           & .329                                     & 3,268                                                    & 5,478             & \underline{2,042} & 167              \\
		                                                       & \oursbb{}                         & \underline{.291}                         & \underline{1,285}                                        & \underline{3,484} & 1,993             & \underline{206}  \\
		\midrule
		\multirow{2}[2]{*}{\begin{sideways}CM\end{sideways}}   & \ours{}                           & \underline{.531}                         & 2,456                                                    & 3,939             & 1,053             & 429              \\
		                                                       & \oursbb{}                         & .577                                     & \underline{2,195}                                        & \underline{3,785} & \underline{1,120} & \underline{470}  \\
		\midrule
		\multirow{2}[2]{*}{\begin{sideways}OSM\end{sideways}}  & \ours{}                           & \underline{.326}                         & 529                                                      & 1,243             & 435               & 278              \\
		                                                       & \oursbb{}                         & .363                                     & \underline{236}                                          & \underline{959}   & \underline{439}   & \underline{284}  \\
		\bottomrule
	\end{tabular}%

	\caption{
		\textbf{Algorithmic statistic of vanilla \ourlib{} (\ours{}) and with bounding-box optimizations (\oursbb{}) for range report on real-world datasets. Underlined values indicate better performance.}
		The range report query contains $10^4$ rectangles each with output size $10^4$--$10^6$. Different queries are performed in parallel, and each query searches the tree in serial.
		``Leaf'': average number of leaf nodes visited per query, ``Interior'': average number of non-leaf nodes visited per query, ``Skip'': average number of nodes skipped per query, i.e., the associated sub-space does not intersect with the query box, ``Flatten'': average number of nodes flattened per query, i.e., the associated sub-space is fully contained in the query box.
	}
	\label{table:appendix:perf-alg}%
\end{table}%

\section{Profiling for Range Report Queries\label{app:profiling}}
As we mentioned, \ourlib{} does not store the bounding box for subtrees in each node, but instead computes them on-the-fly during the query.
This reduces the memory usage of the tree, while increasing the computation cost and potentially more number of nodes to be explored during the query.
To verify this trade-off, we perform the hardware profiling, as well as summarizing the algorithmic statistics, for both the \ourlib and the variant with the bounding box stored in each node.
The results are illustrated in \cref{table:appendix:perf-hardware} and \cref{table:appendix:perf-alg}, respectively.

Storing the bounding box within the node can facilitate pruning, resulting in fewer nodes being explored during searches.
However, this advantage diminishes as the dimensionality of the input points decreases.
On the contrary, reduced memory usage enhances search speed more, as the dynamically computed bounding box is nearly as tight as the exact one when the dimensionality of inputs is low.
We have discussed this in the main paper, and show the full results here in \cref{table:appendix:perf-hardware} and \cref{table:appendix:perf-alg}.

\section{High-dimensional Synthetic Dataset\label{app:synthetic}}
\begin{table*}[h]
	\centering
	\small
	\setlength\tabcolsep{6pt} 
	\renewcommand{\arraystretch}{0.8} 

	\begin{tabular}{cc|c|cccc|cccc|c|c}
		\toprule
		\textbf{Benchmark}                    & \multirow{2}[2]{*}{\textbf{Baselines}} & \multirow{2}[2]{*}{\textbf{Build}} & \multicolumn{4}{c|}{\textbf{Batch Insert}} & \multicolumn{4}{c|}{\textbf{Batch Delete}} & \textbf{10-NN }  & \textbf{Range Report}                                                                                                                                                                                 \\
		\textbf{(100M-12D)}                   &                                        &                                    & \textbf{0.01\%}                            & \textbf{0.1\%}                             & \textbf{1\%}     & \textbf{10\%}         & \textbf{0.01\%}  & \textbf{0.1\%}   & \textbf{1\%}     & \textbf{10\%}    & \boldmath{}\textbf{$10^7$ queries}\unboldmath{} & \boldmath{}\textbf{$10^4$ queries}\unboldmath{} \\
		\midrule
		\multirow{4}[2]{*}{\textbf{\uniform}} & Ours                                   & \underline{.938}                   & \underline{.001}                           & \underline{.003}                           & \underline{.017} & \underline{.115}      & \underline{.002} & \underline{.005} & \underline{.029} & \underline{.148} & \underline{51.4}                                & \underline{17.4}                                \\
		                                      & \logtree                               & 19.9                               & .013                                       & .017                                       & 4.09             & 4.22                  & .345             & .347             & .345             & .427             & t.o.                                            & s.f.                                            \\
		                                      & \bhltree                               & 14.4                               & 13.1                                       & 13.0                                       & 13.3             & 14.6                  & 10.9             & 11.1             & 11.4             & 13.1             & t.o.                                            & s.f.                                            \\
		                                      & \cgal                                  & 112                                & 164                                        & 168                                        & 157              & 171                   & .032             & .278             & 2.83             & 30.0             & 63.7                                            & 109                                             \\
		\midrule
		\multirow{4}[2]{*}{\textbf{\varden}}  & Ours                                   & \underline{1.07}                   & \underline{.002}                           & \underline{.004}                           & \underline{.025} & \underline{.269}      & \underline{.002} & \underline{.005} & \underline{.023} & \underline{.169} & .048                                            & \underline{10.5}                                \\
		                                      & \logtree                               & 19.7                               & .012                                       & .013                                       & 3.99             & 4.00                  & .342             & .349             & .346             & .427             & t.o.                                            & s.f.                                            \\
		                                      & \bhltree                               & 14.4                               & 12.8                                       & 12.7                                       & 12.8             & 14.1                  & 11.3             & 11.8             & 11.9             & 13.0             & t.o.                                            & s.f.                                            \\
		                                      & \cgal                                  & 46.7                               & 70.7                                       & 63.7                                       & 59.5             & 63.9                  & .003             & .032             & .452             & 4.82             & \underline{.046}                                & 114                                             \\
		\bottomrule
	\end{tabular}%

	\caption{
		\textbf{Running time (in seconds) for the \ourlib{} and other baselines on $10^8$ points in 12 dimensions. Lower is better.}
			The 10-NN queries $10^7$ points in parallel, and the range report contains $10^4$ rectangle range queries in parallel with output size $10^4$ to $10^6$.
			All queries searches the tree in serial.
			``t.o.'': time out after 600s.
			``s.f.'': segmentation fault.
	}
	\label{table:highdimsyn}%
\end{table*}%

We compare the \ourlib{} with other baselines on synthetic datasets with 12 dimensions on 100 million points, and present the results in \cref{table:highdimsyn}.
\ourlib{} remains its advantage on tree construction and batch update, due to higher cache efficiency.
For tree construction, the \ourlib is 18.3--21.2$\times$ faster than the \logtree, 13.4--15.4$\times$ faster than the \bhltree and 43.5--120$\times$ faster than the \cgal.
Regarding batch insertion, \ourlib achieves 3.27--240$\times$ speedup than \logtree{}, and orders of magnitude faster than \bhltree{} and \cgal{}.
For batch deletion, compared with the best baseline, \ourlib{} is 1.5--202$\times$ faster than the \cgal{}, 2.53--172$\times$ faster than \logtree{} and orders of magnitude faster than \bhltree{}.

For 10-NN queries, \ourlib is 1.24$\times$ faster than \cgal on \uniform{}, while it is slightly slower on \varden{}, where \cgal{} is 1.04$\times$ faster.
Regarding the range report, \ourlib is 6.27--10.9$\times$ faster than \cgal{}.
Both \logtree{} and \bhltree{} timeouts on \knn{} queries as explained in \cref{app:ood}. They also experience the segmentation fault on the range report for the same reason.

}

\end{document}
\endinput